\documentclass[11pt]{article}
\usepackage{geometry}  
\usepackage{graphicx}
\usepackage{amssymb}
\usepackage{amsmath}
\usepackage{amsthm}

\usepackage{xcolor}
\usepackage{epstopdf}
 \usepackage{xr-hyper}
\usepackage{mathtools}
\usepackage[unicode=true,pdfusetitle,
 bookmarks=true,bookmarksnumbered=false,bookmarksopen=false,
 breaklinks=false,pdfborder={0 0 1},backref=false,colorlinks=true]
 {hyperref}
 \hypersetup{urlcolor=blue}
\hypersetup{citecolor=blue}
\hypersetup{filecolor=red}
\usepackage{url}
\usepackage{setspace}
\makeatletter
\def\@setthanks{\vspace{-\baselineskip}\def\thanks##1{\@par##1\@addpunct.}\thankses}
\makeatother
\usepackage{natbib}
\usepackage{caption}
\usepackage{tabularx}
 \usepackage{float}
 \usepackage{booktabs}
 \usepackage{multirow}
\usepackage{bm}
\newtheorem{thm}{Theorem}[section]

\newtheorem{lem}{Lemma}[section]

\newtheorem{assum}{Assumption}

\numberwithin{equation}{section}
\newtheorem{exm}{Example}[section]
\theoremstyle{remark}
\newtheorem{remark}{Remark}[section]
\newtheorem{innercustomgeneric}{\customgenericname}
\providecommand{\customgenericname}{}
\newcommand{\newcustomtheorem}[2]{%
  \newenvironment{#1}[1]
  {%
   \renewcommand\customgenericname{#2}%
   \renewcommand\theinnercustomgeneric{##1}%
   \innercustomgeneric
  }
  {\endinnercustomgeneric}
}
\numberwithin{equation}{section}
\newcustomtheorem{customthm}{Theorem}
\newcustomtheorem{customlem}{Lemma}
\newcustomtheorem{customassumption}{Assumption}
\newcustomtheorem{customprop}{Proposition}
\newcustomtheorem{customexample}{Example}
\newcustomtheorem{customdef}{Definition}
\newcustomtheorem{customcor}{Corollary}
\newcustomtheorem{customrem}{Remark}

\usepackage[title]{appendix}
\usepackage[ruled,vlined]{algorithm2e}
\SetKwInput{KwInput}{initialize}
\SetKwInput{KwOutput}{output}
\allowdisplaybreaks
\makeatletter
\newcommand*{\addFileDependency}[1]{
  \typeout{(#1)}
  \@addtofilelist{#1}
  \IfFileExists{#1}{}{\typeout{No file #1.}}
}
\makeatother
 
\begin{document}
\begin{titlepage}

\title{Universal Factor Models\thanks{A previous version of this paper was circulated
with the title “Robust Quantile Factor Analysis”. We have been encouraged by and benefited from comments from Jushan Bai, Yue Fang, Bryan Graham, Sokbae Lee, Whitney Newey, Myung Hwan Seo, Ke-Li Xu, Andrei Zeleneev, and the attendees of the 2023 Symposium for High Dimensional Econometric and Machine Learning, 2023 AFR International Conference of Economics and Finance, the Inaugural Meeting of the Great Bay Econometrics Study Group (2024), 2024 First Macau International Conference on Business Intelligence and Analytics, 2025 Harbin International Conference on Econometrics and Statistics, the 2025 World Congress of the Econometric Society, and the econometrics seminars in 2023, 2024 and 2025 at Peking University, Shanghai University of Finance and Economics, Singapore Management University, Sun Yat-Sen University, Southwestern University of Finance and Economics, and Xiamen University.}}
\author{Songnian Chen\thanks{School of Economics, Zhejiang University. Email: \href{mailto:snchen2022@zju.edu.cn}{snchen2022@zju.edu.cn}.}\ \ \ \ \ \  Junlong Feng\thanks{Department of Economics, the Hong Kong University of Science and Technology. Email: \href{mailto:jlfeng@ust.hk}{jlfeng@ust.hk}.}}

\date{February 2026}                                      

\maketitle
\vspace{-3em}
\begin{abstract} 
We propose a new factor analysis framework and estimators of the factors and loadings that are robust to certain weak factors in a large $N$ and large $T$ setting. Our framework, by simultaneously considering all quantile levels of the outcome variable, induces standard mean and quantile factor models, but the factors can have an arbitrarily weak influence on the outcome’s mean or quantile at most quantile levels. Our method estimates the factor space at the $\sqrt{N}$-rate as long as each factor is strong at some unknown quantile level, and achieves  $\sqrt{N}$- and $\sqrt{T}$-asymptotic normality for the factors and loadings based on a novel sample splitting approach that handles incidental nuisance parameters. 
We also develop a weak-factor-robust estimator of the number of factors and consistent selectors of factors of any tolerated level of influence on the outcome’s mean or quantiles. Monte Carlo simulations demonstrate the effectiveness of our method. 
\vspace{0.1in}
\vspace{0in}\\
\end{abstract}
\end{titlepage}

\section{Introduction}

Factor models have a wide application in economics and finance. In these models, observable outcomes are affected by a few number of latent economic variables called \textit{factors}. In asset pricing, researchers often assume that there are a small number of latent common shocks that drive asset returns. In macroeconomics, common shocks may have heterogeneous effects on cyclical variations. In synthetic control, variations in individuals' potential outcomes are assumed to be affected by latent time trends. 

A frequent challenge for researchers is the possible existence of weak factors. When the outcome variable has little exposure to some factors, learning about them becomes difficult. Moreover, their presence can distort inference regarding other factors that exert a more prominent influence \citep{onatski2012asymptotics,bai2023approximate,giglio2025test}. Because of this issue, popular methods, such as the principal component analysis (PCA, \cite{bai2002determining}, \cite{bai2003inferential}) and quantile factor analysis (QFA, \cite{chen2021quantile}), often assume that all factors have strong impact on the mean or quantile of $Y$ at the quantile level of interest, usually referred to as the \textit{strong factor assumptions}.

However, we argue in this paper that violating this type of strong factor assumptions \textit{does not necessarily imply} that the factors are inherently \textit{weak}; it could only imply the poor choice of the moment or quantile of the outcome variable used in estimation; this argument echoes a similar spirit in \cite{giglio2025test} where they state that ``the strength or weakness of a factor ... should
not be viewed as a property of the factor itself; rather, it should be viewed as a property of the set of test assets used in estimation'' (p.250). 
As a consequence, failure of PCA and QFA \textit{does not necessarily imply} that the factors cannot be well-estimated. For illustration, consider a random coefficient model for an outcome variable $Y_{it},i=1,\ldots,N,t=1,\ldots,T$, $$Y_{it}=\sum_{j=1}^{r}\beta_{j}(U_{it})\lambda_{0,ij}^{*}f_{0,tj}^{*},$$ 
where $\lambda_{0,ij}^{*}$ and $f_{0,tj}^{*}$, $j=1,\ldots,r$, are $r$ latent factor loadings and factors for $i$ and $t$, respectively, and the unobserved $U_{it}$ is independent of the factors and uniformly distributed. Suppose $\beta_{j}(\cdot)$ is integrable on $(0,1)$ for all $j$ with the integral denoted by $\bar{\beta}_{j}$ and $\sum_{j=1}^{r}\beta_{j}(\cdot)\lambda_{0,ij}^{*}f_{0,tj}^{*}$ is strictly increasing on $(0,1)$, then the conditional mean and $\tau$-th quantile of $Y_{it}$ are 
\begin{equation*}
\mathbb{E}(Y_{it}|f^{*}_{0,t})=\sum_{j=1}^{r}\bar{\beta}_{j}\lambda^{*}_{0,ij}f^{*}_{0,tj},\ \ \ \ q_{Y_{it}|f^{*}_{0,t}}(\tau)=\sum_{j=1}^{r}\beta_{j}(\tau)\lambda_{0,ij}^{*}f^{*}_{0,tj},
\end{equation*}
respectively. Let $f_{0,t}^{*}=(f_{0,t1}^{*},\ldots,f^{*}_{0,tr})'$ and $\lambda_{0,i}^{*}=(\lambda_{0,i1}^{*},\ldots,\lambda_{0,ir}^{*})'$. Suppose $\sum_{t=1}^{T}f^{*}_{0,t}f^{*'}_{0,t}/T$ and $\sum_{i=1}^{n}\lambda^{*}_{0,i}\lambda^{*'}_{0,i}/N$ are both positive definite uniformly in $T$ and $N$. Then the strength of factors in the sense of PCA and QFA respectively refers to the magnitude of the $\bar{\beta}_{j}$s (i.e., $\int_{0}^{1}\beta_{j}(\tau)d\tau$) and $\beta_{j}(\tau)$. One can see from here that a \textit{weak factor} in the sense of PCA can be a \textit{strong factor} in the sense of QFA, vice versa, and a \textit{weak factor} at some $\tau$ in the sense of QFA can be a \textit{strong factor} in the same sense just at a different $\tau'$. In fact, factor $f_{0,tj}^{*}$ is not inherently weak as long as its overall strength $\int_{\mathcal{U}}[\beta_{j}(\tau)]^{2}d\tau$ is not small for some $\mathcal{U}\subset (0,1)$.

This paper proposes a new factor analysis framework, the \textit{universal factor model} (UFM), that only imposes restrictions on the overall strength of factors rather than on the strength particularly regarding the mean or quantile of $Y_{it}$ at a certain $\tau$. Our model has the following form, which nests the previous example
\begin{equation}\label{eq.model.intro}
Y_{it}=\lambda_{0,i}^{*'}(U_{it})f_{0,t}^{*}, \ \ U_{it}\sim \text{Unif}[0,1].
\end{equation}
Let $\Lambda_{0}^{*}(\cdot)$ and $F_{0}^{*}$ be $N\times r$ and $T\times r$ matrices collecting all the loading functions and factors, respectively. We require that the factors have an overall strong impact on $Y_{it}$ in the sense that $\int_{\mathcal{U}}\Lambda^{*'}_{0}(\tau)\Lambda^{*}_{0}(\tau)d\tau/N$ is positive definite uniformly in $N$ on a fixed compact interval $\mathcal{U}\subseteq (0,1)$; since we do not require integrability to hold on the entire $(0,1)$, this restriction does not rule out $Y_{it}$ with heavy or fat tails. We refer to $F_{0}^{*}$ as \textit{universal factors}.

One can view our model \eqref{eq.model.intro} as a generalization of standard panel data models with fixed effects without observable regressors. For instance, by letting $f_{0,t}^{*}=(1,f_{0,t})'$ and $\lambda^{*}_{0,i}(U_{it})\coloneqq g(\lambda_{0,i},U_{it})\coloneqq (\lambda_{0,i}+\Phi^{-1}(U_{it}),1)$ for some inverse cumulative distribution function $\Phi^{-1}$, we obtain an additive fixed effect model $Y_{it}=\lambda_{0,i}+f_{0,t}+\varepsilon_{it}$ where $\varepsilon_{it}\coloneqq \Phi^{-1}(U_{it})$. Interactive fixed effects are also allowed by, for instance, redefining $\lambda^{*}_{0}(U_{it})$ as $(\lambda_{0,i}+\Phi^{-1}(U_{it}),\lambda_{0,i}+1)'$. Since fixed effects are commonly viewed as genuine latent economic variables, our $f_{0,t}^{*}$ captures latent time-specific variables while $\lambda^{*}_{it}(U_{it})$ absorbs the time-invariant unobservable individual heterogeneity $\lambda_{0,i}$ and time-varying idiosyncratic shocks $U_{it}$. One can also compare our model with the standard quantile regression model with observed interactive terms: $Y_{it}=\beta_{0}(U_{it})+X_{1i}\beta_{1}(U_{it})+X_{2t}\beta_{2}(U_{it})+X_{1i}X_{2t}\beta_{3}(U_{it})$, where $X_{1i}$ and $X_{2t}$ are \textit{observable}. One can reparameterize this model as $Y_{it}=X_{i}^{*'}(U_{it})X_{t}^{*}$ by letting $X_{i}^{*}(U_{it})\coloneqq (\beta_{0}(U_{it})+X_{1i}\beta_{1}(U_{it}),\beta_{2}(U_{it})+X_{1i}\beta_{3}(U_{it}))'$ and $X_{t}^{*}=(1,X_{2t})'$. Hence, just like $X_{1i}$ and $X_{2t}$, the factors and loadings in our model, although latent, are still genuine economic variables. A factor is weak only when the corresponding coefficient (marginal effect) is weak, very much like the case of weak instruments, which are themselves not weak (small); indeed, instruments are weak when they have small effects (correlation) on the endogenous variable.

Our model \eqref{eq.model.intro} induces a quantile factor model (QFM, \cite{chen2021quantile}) by assuming that $U_{it}$ is independent of $F_{0}^{*}$ and every component in $\Lambda_{0}^{*}(\cdot)F_{0}^{*'}$ is strictly increasing. Out assumption on the strength of the factors is weaker. To be specific, the strong factor assumption in \cite{chen2021quantile} requires all factors that have any impact on the $\tau$-th quantile of $Y_{it}$ to have a strong impact. In contrast, our assumption holds as long as each factor is strong at some potentially different $\tau\in\mathcal{U}$. 

A second advantage of our UFM modeling over QFM is that the factors in QFM are $\tau$-dependent; QFM thus only captures a subset of the universal factors $F_{0}^{*}$, and we show in this paper that in many models, such $\tau$-dependence is an unavoidable artificial parameterization only for their strong factor assumption to hold. Such ad hoc parameterization leads to difficulties in interpretation since in economics and finance, factors are simply unobserved latent variables, not to be $\tau$-dependent. Moreover, confusion may arise when the researcher wishes to use the recovered factors to augment regression in a different scenario. Our UFM modeling approach together with our relaxed requirement on the strength of factors avoids such artifact while entertaining a wider class of data generating processes.

Our model also induces a mean factor model, or, approximate factor model (AFM) in the form of \cite{bai2003inferential} if $\Lambda^{*}(\cdot)$ is integrable on $(0,1)$. We show that our requirement on the strength of the factors is also weaker than their strong factor assumption for PCA to obtain a $\sqrt{N}$-asymptotically normal estimator of the factors.

To estimate the universal factors, we first propose our baseline estimator, referred to as the \textit{universal factor analysis} (UFA). It begins by obtaining preliminary estimators of $F^{*}_{0}$ and $\Lambda^{*}_{0}(\tau)$ at various $\tau$s by iteratively conducting standard quantile regression or the smoothed version such as \cite{fernandes2021smoothing} and \cite{he2023smoothed}, which delivers both nice theoretical properties and computational efficiency. Using them, we estimate $F^{*}_{0}\int_{\tau\in\mathcal{U}}\Lambda_{0}^{*'}(\tau)\Lambda_{0}^{*}(\tau)d\tau F_{0}^{*'}/NT$. Its eigenvectors multiplied by $\sqrt{T}$ are the final estimated factors under our normalization. UFA consistently estimates the space spanned by the factors at the $\sqrt{N}$-rate regardless of whether smoothed quantile regression is adopted. In this paper, we only present the smoothed estimator for the sake of space. 

The asymptotic expansions of the smoothed UFA estimator involve the densities of $Y_{it}$ evaluated at the $(i,t)$-th common component $\lambda_{0,i}^{*'}(\tau)f_{0,t}^{*}$ conditional on the universal factors. These $(i,t)$-specific nuisance parameters impose theoretical challenges to obtain sharper theoretical results.
To sharpen the results to achieve asymptotic normality, we propose a two-stage inverse density weighted estimator, referred to as IDW-UFA, based on a new sample splitting strategy. In the first stage, we estimate the factors, loadings and thus $(i,t)$-specific inverse conditional densities by UFA using appropriate subsamples; we directly estimate these inverse densities exploiting the factor structure of our model instead of inverting an estimated density, achieving numerical stability. We then estimate the factors and loadings again in the second stage using the \textit{full dataset} by weighting the original objective function in UFA by those estimated inverse densities to eliminate the incidental nuisance parameters in the asymptotic expansions. 

The new sample splitting approach differs from the standard approach (e.g. \cite{chernozhukov2018double}) in two ways. First, unlike the standard approach where different subsamples are usually governed by the same nuisance parameters, our nuisance parameters here are $(i,t)$-specific, so the ones obtained using one subsample do not apply to another. Our new approach solves this problem by delicately constructing subsamples. Second, our estimated nuisance parameters have a negligible impact on the second stage where the \textit{full dataset} is used. Because of this, the IDW-UFA estimator does not suffer from the loss of efficiency or instability. This sample splitting strategy may be of independent interest; it can be applied to other scenarios where one needs to estimate a latent factor structure prior to estimating the main model; it also applies to estimating QFMs under the strong factor assumption; \cite{bai2021matrix} tackle missing data problems in factor models in a similar spirit.

The UFA-IDW estimator achieves asymptotic normality (up to rotation) at the $\sqrt{N}$- or $\sqrt{T}$-rate for individual factors, loadings and common components, assuming that $N$ and $T$ are of the same order. The rotation matrix matches those in \cite{bai2003inferential} and \cite{bai2023approximate} for PCA. 

In addition to our estimator of the $\tau$-specific loadings, we also develop a $\sqrt{T}$-asymptotically normal estimator of the factor loadings in the UFM-induced AFM based on our estimated universal factors. Hence, our estimator is useful even if the researcher is only interested in recovering the mean factors and their effects. Unlike PCA, it does not require the mean factors to be strong. 

Finally, we propose an eigen/singular-value-thresholding type estimator for the total number of universal factors. We also propose consistent selectors of factors that have an arbitrary tolerated level of strength on the mean or quantile of $Y$. 

Our paper adds to the growing literature on weak factors. \cite{onatski2012asymptotics}, \cite{bai2023approximate}, \cite{jiang2023revisiting}, \cite{fan2024can}, and \cite{choi2025high} provide theoretical and simulation evidence showing that the PCA may be inconsistent or have a slower rate of convergence when the factors in an AFM are weak. \cite{bai2019rank} propose a ridge penalized estimator which selects out the relevant strong factors in an AFM in a data driven way. \cite{onatski2010determining} estimate the number of factors in an AFM with weak mean factors. Other methods developed in the presence of weak factors in mean models often impose conditions requiring the eigenvalues of $\Lambda_{0}^{\mu*'}\Lambda_{0}^{\mu*}/N$ not too small or the loading matrix to be sparse, e.g. \cite{de2008forecasting,lettau2020estimating,bailey2021measurement, freyaldenhoven2022factor,uematsu2022estimation}.  In our paper, these eigenvalues can be arbitrarily small as long as our identification condition mentioned earlier is satisfied. Meanwhile, we do not need to assume sparsity for $\Lambda_{0}^{\mu *}$. Moreover, our method handles both mean and quantile models. The weak factor problem is also relevant in the literature of risk premium estimation. \cite{anatolyev2022factor} and \cite{giglio2025test} develop methods to tackle weak factors in that setting. The focus and models are different from our paper. In panel data regression with interactive fixed effects, \cite{armstrong2022robust} propose a robust estimation approach to construct bias-aware confidence intervals for the slope coefficients on the regressors when the fixed effects are weak. In contrast, we focus on the factors and loadings.

 The rest of the paper is organized as follows. We formally set up the model in Section \ref{sec.model}. We introduce our baseline estimator UFA and provide its rate of convergence in Section \ref{sec.estimator}. Section \ref{sec.est2} presents the inverse density weighted estimator and derives its asymptotic distribution. Section \ref{sec.mean} demonstrates how to estimate the mean loadings. Section \ref{sec.r} provides a consistent estimator of the total number of factors and consistent selectors of the factors in a QFM or an AMF with arbitrary tolerated strength. Section \ref{sec.simu} examines the finite sample performance of the estimator by Monte Carlo simulations. Section \ref{sec.con} concludes. The Appendix and Online Appendix collect all the proofs.

 \section{The Model}\label{sec.model}

Assume that observable $Y_{it}$ ($i=1,\ldots,N;t=1,\ldots,T$) follows the model in the Introduction,
\begin{equation}\label{eq.dgp.model}
Y_{it}=\lambda_{0,i}^{*'}(U_{it})f_{0,t}^{*}
\end{equation}
where $\lambda^{*}_{0,i}(\tau)$ is an $r\times 1$ vector of factor loadings. We treat $\lambda_{0,i}^{*}(\tau)$ as deterministic whereas $f_{0,t}^{*}$ as realizations of some underlying random variables $f_{0,t}^{0*}$. All the statements in the rest of the paper are implicitly conditional on $f_{0,t}^{0*}=f_{0,t}^{*}$.\footnote{Alternatively, we can view the loading functions as random and condition on the realization of them as well.} Let the $T\times r$ matrix $F^{*}_{0}$ collect all the factors. We call our factors, i.e., the columns in $F_{0}^{*}$, the \textit{universal factors} in contrast to the mean factors in AFMs such as in \cite{bai2003inferential} and quantile factors in the QFM in \cite{chen2021quantile}. Those factors are subsets of the universal factors since they only include the universal factors that have strong impact on the mean and a certain quantile of $Y_{it}$, respectively.

Throughout the paper, we maintain the following assumption which admits a linear conditional quantile representation for $Y_{it}$:
\begin{assum}\label{assum.foundation}
For all $i$ and $t$, $U_{it}\sim \text{Unif}[0,1]$ and $\lambda_{0,i}^{*'}(\cdot)f_{0,t}^{*}$ is strictly increasing on $(0,1)$.
\end{assum}
Let $q_{Y_{it}|f_{0,it}^{*}}(\tau),\tau\in (0,1)$, denote the $\tau$-th conditional quantile of $Y_{it}$. Then by Assumption \ref{assum.foundation}, 
 \begin{equation}\label{eq.model}
   q_{Y_{it}|f^{*}_{0,t}}(\tau)=\lambda^{*'}_{0,i}(\tau)f^{*}_{0,t},\tau\in (0,1).
 \end{equation}
Let the $N\times T$, $N\times r$ and $T\times r$ matrices $Y$, $\Lambda_{0}^{*}(\tau)$ and $F^{*}_{0}$ collect all the observables, loadings and factors, respectively. We can rewrite \eqref{eq.model} in matrix form as:
\begin{equation}\label{eq.model.matrix}
  q_{Y|F^{*}_{0}}(\tau)=\Lambda^{*}_{0}(\tau)F^{*'}_{0}\eqqcolon L_{0}(\tau),\tau\in(0,1).
\end{equation}

Model \eqref{eq.dgp.model} also implies a factor structure for the conditional mean of $Y_{it}$, provided the existence of the latter, or equivalently, integrability of $\lambda_{0,i}^{*}(\cdot)$ on $(0,1)$. Let $\bar{\lambda}_{0,i}^{*}\coloneqq \int_{0}^{1}\lambda_{0,i}^{*}(\tau)d\tau$. Then
\begin{equation*}
\mathbb{E}(Y_{it}|f_{0,t}^{*})=\bar{\lambda}_{0,i}^{*'}f_{0,t}^{*}.
\end{equation*}
Therefore, by letting $\nu_{it}\coloneqq Y_{it}-\bar{\lambda}_{0,i}^{*'}f^{*}_{0,t}$,
\begin{equation*}
  Y_{it}=\bar{\lambda}_{0,i}^{*'}f^{*}_{0,t}+\nu_{it},\ \ \mathbb{E}(\nu_{it}|f^{*}_{0,t})=0.
\end{equation*}

In the rest of this section, we first introduce our assumption on the strength of the universal factors. We then compare our model and assumption with QFM and AFM.

\subsection{A Relaxed Assumption on the Strength of Factors}

Let $\mathcal{B}$ be a compact subset of $\mathbb{R}$ and $\mathcal{U}$ be a compact subset of $(0,1)$. 

\begin{assum}\label{assum.normalization}
(i) The $r$ eigenvalues of $\int_{\mathcal{U}}\Lambda_{0}^{*'}(\tau)\Lambda_{0}^{*}(\tau)d\tau/N$ are bounded away from infinity and from 0 for sufficiently large $N$. (ii) The $r$ eigenvalues of $F_{0}^{*'}F_{0}^{*}/T$ are bounded away from infinity and from 0 for sufficiently large $T$. (iii) The $r$ largest eigenvalues of matrix $F^{*}_{0}\int_{\mathcal{U}}\Lambda_{0}^{*'}(\tau)\Lambda_{0}^{*}(\tau)d\tau F^{*'}_{0}/(NT)$ are distinct. (iv) $f^{*}_{0,t},\lambda^{*}_{0,i}(\tau)\in\mathcal{B}^{r}$ for each $i,t$ and $\tau\in \mathcal{U}$. 
 \end{assum}
Part (ii) of Assumption \ref{assum.normalization} is standard in the factor model literature. Under Parts (i) and (ii), the $r$ largest eigenvalues of the matrix considered in (iii) are bounded away from 0. Part (iii) further imposes distinctiveness of eigenvalues to guarantee uniqueness (up to column signs) of the eigenvectors. The boundedness of factors and loadings in part (iv) is identical to \cite{chen2021quantile}.

Part (i) of Assumption \ref{assum.normalization}, on the other hand, is new. It imposes restrictions on the strength of universal factors. It allows the singular values of the loading matrix to be small at any given $\tau$. Note that integrability of the loadings implied by this assumption does not rule out heavy- or fat-tailed $Y_{it}$ because of the compactness of $\mathcal{U}$. Now for better illustration, we compare our model and assumption with the QFM in \cite{chen2021quantile} and AFM in \cite{bai2002determining}, \cite{bai2003inferential}, and \cite{bai2023approximate}.

 \subsection{Comparison with the QFM}
 \cite{chen2021quantile} study the following quantile factor model:
 \begin{equation}\label{eq.model.chen}
    q_{Y_{it}|f^{q*}_{0,t}(\tau)}(\tau)=\lambda_{0,i}^{q*'}(\tau)f^{q*}_{0,t}(\tau),\tau\in (0,1),
 \end{equation}
 where the factors are $\tau$-dependent. Let the number of factors at $\tau$ be $r(\tau)$. They require all the $r(\tau)$ factors to be strong at every $\tau$ of interest: For sufficiently large $N$,
 \begin{equation}\label{eq.norm.Lambda.chen}
   \text{the eigenvalues of }\frac{\Lambda_{0}^{q*'}(\tau)\Lambda^{q*}_{0}(\tau)}{N}\text{ are bounded away from 0.}
 \end{equation}
Similar to our Assumption \ref{assum.normalization}-(ii), they also maintain that for sufficiently large $T$,
   \begin{equation}\label{eq.norm.F.chen}
   \text{the eigenvalues of }\frac{F_{0}^{q*'}(\tau)F^{q*}_{0}(\tau)}{T}\text{ are bounded away from 0.}
 \end{equation}

Under model \eqref{eq.dgp.model} and by \eqref{eq.model}, the columns in quantile factors $F^{q*}_{0}(\tau)$ form a subspace of the column space of universal factors $F^{*}_{0}$. The strong factor assumption \eqref{eq.norm.Lambda.chen} thus says that, for every $\tau$, the factors in $F^{*}_{0}$ have either sufficiently large or exactly zero impact on the $\tau$-th conditional quantile of $Y_{it}$. However, our Assumption \ref{assum.normalization}-(i) allows any universal factor to have arbitrarily weak impact at any $\tau\in\mathcal{U}$, as long as each universal factor is ``strong'' in the sense of \cite{chen2021quantile} in a neighborhood of some $\tau$, and these neighborhoods can be different for different factors.

Besides the more stringent requirements on the strength of factors, another drawback of the QFM is the ad hoc dependence of the factors on the quantile level $\tau$. In \cite{chen2021quantile}, they provide four examples where $Y_{it}$ is modeled as a function of factors, loadings and idiosyncratic shocks. As latent economic variables, none of the factor in those models depends on $\tau$, and all those four models can be nested by our model \eqref{eq.dgp.model}. However, once transformed into the quantile model, dependence in $\tau$ can become unavoidable to satisfy their strong factor assumption \eqref{eq.norm.Lambda.chen}. Such artificial dependence on $\tau$ creates barriers to understand the underlying model and to use the obtained factors to augment regression in other scenarios. In contrast, our quantile representation can avoid such $\tau$-dependence because of our milder restrictions on the factors' strength.

The following example, taken from \cite{chen2021quantile}, illustrates these points.
\begin{exm}[Example 4 in \cite{chen2021quantile}, p.879]\label{ex.chen}
$Y_{it}=\alpha_{i}f^{*}_{1t}+f_{2t}^{*}\epsilon_{it}+c_{i}f_{3t}^{*}\epsilon^{3}_{it}$, where the $\epsilon_{it}$s are independent standard normal random variables whose cumulative distribution function (CDF) is denoted as $\Phi(\cdot)$. Let $f_{2t}^{*}$, $f_{3t}^{*}$ and $c_{i}$ be positive for all $i,t$. Assume $\epsilon_{it}$ is independent of all the factors. Then letting $\lambda_{0,i}^{*}(\tau)=[\alpha_{i},\Phi^{-1}(\tau),c_{i}(\Phi^{-1}(\tau))^{3}]'$, we have
\begin{align*}
\frac{1}{N}\Lambda_{0}^{*'}(\tau)\Lambda_{0}^{*}(\tau)=\begin{pmatrix}\frac{1}{N}\sum_{i=1}^{N}\alpha_{i}^{2}&\frac{\sum_{i=1}^{N}\alpha_{i}}{N}\Phi^{-1}(\tau)& \frac{\sum_{i=1}^{N}\alpha_{i}c_{i}}{N}\left(\Phi^{-1}(\tau)\right)^{3}\\
\frac{\sum_{i=1}^{N}\alpha_{i}}{N}\Phi^{-1}(\tau)&\left(\Phi^{-1}(\tau)\right)^{2}& \frac{\sum_{i=1}^{N}c_{i}}{N}\left(\Phi^{-1}(\tau)\right)^{4}\\
\frac{\sum_{i=1}^{N}\alpha_{i}c_{i}}{N}\left(\Phi^{-1}(\tau)\right)^{3}& \frac{\sum_{i=1}^{N}c_{i}}{N}\left(\Phi^{-1}(\tau)\right)^{4}&  \frac{\sum_{i=1}^{N}c_{i}^{2}}{N}\left(\Phi^{-1}(\tau)\right)^{6}
\end{pmatrix}.
\end{align*}
\cite{chen2021quantile} show that for the strong factor assumption \eqref{eq.norm.Lambda.chen} to hold, one has to treat the factors as quantile-level dependent. To see it, consider the following two cases.

At $\tau=0.5$, by $\Phi^{-1}(0.5)=0$, $rank\left(\Lambda^{*'}_{0}(0.5)\Lambda^{*}_{0}(0.5)/N\right)=1<3$,
 violating \eqref{eq.norm.Lambda.chen}. Hence, one should treat the factors as $\tau$-dependent so that $f_{0,t}^{*}(0.5)=f_{1t}^{*}$.

At $\tau\neq 0.5$ but if $c_{i}=c$ for all $i$, $rank\left(\Lambda^{*'}_{0}(\tau)\Lambda^{*}_{0}(\tau)/N\right)=2<3$
because the second and third columns in $\Lambda^{*'}_{0}(\tau)\Lambda^{*}_{0}(\tau)/N$ are linearly dependent, still violating \eqref{eq.norm.Lambda.chen}. \cite{chen2021quantile} thus treat the factors in this case as $(f_{1t}^{*},f_{2t}^{*}+f_{3t}^{*}(\Phi^{-1}(\tau))^{2})$ for $\tau\neq 0.5$. However, for $\tau$ very close to $0.5$, the loading of their second factor, $\Phi^{-1}(\tau)$, is close to 0, still not satisfying their strong factor assumption. Moreover, $f_{2t}^{*}$ and $f_{3t}^{*}$ are never separately identifiable at any $\tau$.

 On the other hand, for our Assumption \ref{assum.normalization}-(i) to hold, we can treat $f_{0,t}^{*}=(f_{1t}^{*},f_{2t}^{*},f_{3t}^{*})'$ for all $\tau\in (0,1)$ regardless of whether $c_{i}$ is constant in $i$ or not. One can verify that all the three eigenvalues of $\int_{\underline{u}}^{\bar{u}}\Lambda^{*'}_{0}(\tau)\Lambda^{*}_{0}(\tau)d\tau/N$ are bounded away from 0 if $\sum_{i}\alpha_{i}^{2}/N$ and $\sum_{i}c_{i}^{2}/N$ do not shrink to 0 as $N\to\infty$ for all $\underline{u}$ and $\bar{u}$ that are sufficiently close to 0 and 1, respectively. 
\end{exm}

\subsection{Comparison with the AFM}
Recall that our model \eqref{eq.model} implies the following AFM provided that $\mathbb{E}(Y_{it}|f^{*}_{0,t})$ exists:
\begin{equation*}
  Y_{it}=\bar{\lambda}_{0,i}^{*'}f^{*}_{0,t}+\nu_{it},\ \bar{\lambda}_{0,i}^{*}=\int_{0}^{1}\lambda_{0,i}^{*}(\tau)d\tau,\ \mathbb{E}(\nu_{it}|f^{*}_{0,t})=0.
\end{equation*}

Compared with the standard AFM under the strong factor assumption:
\begin{equation*}
  Y_{it}=\lambda_{0,i}^{\mu *'}f^{\mu *}_{0,t}+\nu_{it},\ \mathbb{E}(\nu_{it}|f^{*}_{0,t})=0;\ \text{eigenvalues of }\frac{\Lambda^{\mu *'}\Lambda^{\mu *}}{N}\text{ are positive for large }N,
\end{equation*}
we can see that the strong factor assumption for AFM is equivalent to:
\begin{equation*}
  \text{The eigenvalues of }\frac{\int_{0}^{1}\Lambda_{0}^{*'}(\tau)d\tau \int_{0}^{1}\Lambda_{0}^{*}(\tau)d\tau}{N}\text{ are either bounded away from 0 or exactly 0}.
\end{equation*}
Therefore, when all the factors are strong in the sense of \cite{bai2002determining} and \cite{bai2003inferential}, our Assumption \ref{assum.normalization}-(i) holds by the Cauchy-Schwarz inequality.

Now we consider the case when some universal factors are weak mean factors. Let $\kappa(N)$ be the order of the smallest eigenvalue of $\left(\int_{0}^{1}\Lambda_{0}^{*}(\tau)d\tau\right)'\left(\int_{0}^{1}\Lambda_{0}^{*}(\tau)d\tau\right)$. If all factors all strong, $\kappa(N)\asymp N$. \cite{onatski2012asymptotics} show that PCA is inconsistent if $\kappa(N)\asymp 1$ and  \cite{bai2023approximate}, \cite{jiang2023revisiting}, \cite{fan2024can} and \cite{choi2025high} relax the strong factor assumption by allowing for $\kappa(N)=o(N)$ while still requiring $\kappa(N)\to\infty$. Under different requirements on $\kappa(N)$, these works show that PCA achieves $\sqrt{\kappa(N)}$-average and pointwise-in-$t$ rate of convergence for factors. In particular, \cite{fan2024can} obtains $\sqrt{\kappa(N)}$-asymptotic normality for PCA when $\kappa(N)$ can be as small as $\log(N)$.

In contrast, our Assumption \ref{assum.normalization} puts no restrictions on $\kappa(N)$ and \textit{even allows for} $\kappa(N)=O(1)$ or $o(1)$, while still achieves $\sqrt{N}$-asymptotic normality. We illustrate this in the following example.
\begin{exm}[A scale model with diminishing location shift]
Let $Y_{it}=f^{*}_{0,t}(1/N^{(1-\beta)/2}+\epsilon_{it})$ with $(1-\beta)>0$ where $\epsilon_{it}$ are i.i.d standard normal and independent of $f^{*}_{0,t}$ for all $t$. Suppose $f^{*}_{0,t}>0$. Then $\lambda_{0,i}^{*}(\tau)=1/N^{(1-\beta)/2}+\Phi^{-1}(\tau)$. Our Assumption \ref{assum.normalization} is satisfied because 
$$\lim_{N\to\infty}\frac{1}{N}\int_{0}^{1}\Lambda_{0}^{*'}(\tau)\Lambda^{*}_{0}(\tau)d\tau>0, \forall \beta,$$
which, by continuity, further implies the existence of a compact $\mathcal{U}\subset (0,1)$ such that 
$$\lim_{N\to\infty}\frac{1}{N}\int_{\mathcal{U}}\Lambda_{0}^{*'}(\tau)\Lambda^{*}_{0}(\tau)d\tau>0, \forall \beta.$$
However,
\[
  \left(\int_{0}^{1}\Lambda_{0}^{*}(\tau)d\tau\right)'\left(\int_{0}^{1}\Lambda_{0}^{*}(\tau)d\tau\right)=N^{\beta},
\]
so $\kappa(N)=O(1)$ if $\beta=0$ and is $o(1)$ if $\beta<0$, leading to inconsistency of PCA.
\end{exm}

 \section{The Baseline Estimator}\label{sec.estimator}
 Under our conditional quantile model \eqref{eq.model}, one can in principle use quantile regression to estimate the unknown factors and loadings. Since $\Lambda^{*}_{0}(\cdot)$ as a function on $\mathcal{U}$ is an infinite dimensional parameter to compute, we discretize $\mathcal{U}$ for feasibility. As for quantile regression, $\sqrt{N}$-rate can be obtained when estimating the space spanned by the factors no matter whether one uses the standard or some smoothed version of quantile regression. To save space, in the paper we only present the estimator based on smoothed quantile regression following \cite{fernandes2021smoothing} and \cite{he2023smoothed} because similar to \cite{chen2021quantile}, smoothed quantile regression also leads to uniform consistency for factors at each $t$.  

 \subsection{Discretize $\mathcal{U}$}\label{sec.M}
Let $(\tau_{1},\ldots,\tau_{M})$ be an equally spaced grid on $\mathcal{U}$.
\begin{assum}\label{assum.lip}
Function $\lambda^{*}_{0,i}(\cdot)$ is Lipschitz continuous on $\mathcal{U}$ with a Lipschitz constant uniform in $i$.
\end{assum}

\begin{lem}\label{lem.discrete}
Assumptions \ref{assum.normalization} and \ref{assum.lip} imply that for sufficiently large $M,N$ and $T$, the eigenvalues of $\sum_{m=1}^{M}\Lambda^{*'}_{0}(\tau_{m})\Lambda^{*}_{0}(\tau_{m})/MN$ are bounded away from 0, and the eigenvalues of matrix $F_{0}^{*}\sum_{m=1}^{M}\Lambda^{*'}_{0}(\tau_{m})\Lambda^{*}_{0}(\tau_{m})F_{0}^{*'}/MNT$ are distinct.
\end{lem}

By this observation, by letting $M$ be a large fixed constant or $M=h(T)$ where $h$ is a known increasing function, we can focus on estimating the factors and loadings at $\tau_{m},m=1,\ldots,M$, which leads to both an implementable estimator and tractable theoretical derivations.  In what follows, we let $M$ be a sufficiently large fixed constant for simplicity. It is straightforward to extend the results in the paper to the case where $M$ is slowing growing with $T$.  

Next, we impose normalization to the true parameters for identification. Let
\begin{equation}\label{eq.eigen}
  F^{*}_{0}\left(\frac{1}{MNT}\sum_{m=1}^{M}\Lambda^{*'}_{0}(\tau_{m})\Lambda^{*}_{0}(\tau_{m})\right)F_{0}^{*'}=F_{0}\left(\frac{1}{MNT}\sum_{m=1}^{M}\Lambda^{'}_{0}(\tau_{m})\Lambda_{0}(\tau_{m})\right)F_{0}',
\end{equation}
where the right-hand side is an eigendecomposition of the left-hand side; the matrix $F_{0}/\sqrt{T}$ collects all the eigenvectors of the left-hand side, and $\sum_{m=1}^{M}\Lambda^{'}_{0}(\tau_{m})\Lambda_{0}(\tau_{m})/MN$ is a diagonal matrix with diagonal entries $\sigma_{1}^{2}\geq \cdots\geq \sigma_{r}^{2}$ equal to the eigenvalues.\footnote{The diagonal matrix of the eigenvalues still has the additive structure in $m$ as on the right-hand side of equation \eqref{eq.eigen}. This is because by definition of eigenvectors, there exists an $m$-independent full-rank matrix $H$ such that $F_{0}^{*}/\sqrt{T}=F_{0}H/\sqrt{T}$. So, the eigenvalue matrix is equal to $\sum_{m=1}^{M}H\Lambda^{*'}_{0}(\tau_{m})\Lambda^{*}_{0}(\tau_{m})H'/MN$, and our $\Lambda_{0}(\tau_{m})=\Lambda_{0}^{*}(\tau_{m})H'$.} Then Lemma \ref{lem.discrete} implies that
\[\frac{F_{0}'F_{0}}{T}=I_{r}\text{ and }\sigma_{1}^{2}>\cdots>\sigma_{r}^{2}>0.\]
Moreover, $f_{0,t}$ and $\lambda_{0,i}(\tau_{m})$ are also uniformly bounded; with a bit abuse of notation, still assume that both are in $\mathcal{B}^{r}$ for all $i,t$ and $m$. Meanwhile, $\lambda_{0,i}(\cdot)$ is also Lipschitz uniformly on $\mathcal{U}$ under Assumption \ref{assum.lip}. Similar to \cite{chen2021quantile}, in the rest of the paper, we will treat these diagonalized $F_{0}$ and $\Lambda_{0}(\tau_{m})$s as our parameters of interest for simplicity.

\subsection{Smooth Quantile Regression}\label{sec.estimator1.estimator}
To achieve uniform consistency for the factors and loadings, we adopt the smoothed quantile regression in \cite{fernandes2021smoothing} and \cite{he2023smoothed}. Unlike \cite{horowitz1998bootstrap}, \cite{galvao2016smoothed} and \cite{chen2021quantile} where the indicator function in the check function is smoothed by some CDF kernel, this version of smoothed quantile regression keeps the check function but smooths the empirical distribution. \cite{fernandes2021smoothing} and \cite{he2023smoothed} show that it has superior theoretical and computational properties compared to the traditional smoothing techniques. 

Specifically, let $k(\cdot)$ be some smooth kernel function and $h$ be some bandwidth converging to 0 as $N,T\to\infty$. Under diagonalization of the true factors and loadings, we impose the following normalization for the estimator. Let $\Lambda(\cdot)\coloneqq (\Lambda(\tau_{m}))_{m=1,\ldots,M}$. Define the following parameter spaces:
 \begin{align*}
\mathcal{F}\coloneqq&\Bigg\{F\in \mathcal{B}^{T\times r}:\frac{F'F}{T}=I_{r}\Bigg\},\\
\Xi\coloneqq&\Bigg\{\Lambda(\cdot)\in \left(\mathcal{B}^{N\times r}\right)^{M}:\sum_{m=1}^{M}\frac{\Lambda'(\tau_{m})\Lambda(\tau_{m})}{MN}\text{ is diagonal}\\
& \ \ \ \ \ \ \ \ \ \ \ \ \ \ \ \ \ \ \ \ \text{with the diagonal entries in decreasing order}\Bigg\}.\label{eq.norm.Lambda}
\end{align*}
By construction, the diagonalized true loadings and factors satisfy $(\Lambda_{0}(\cdot),F_{0})\in\Xi\times\mathcal{F}$.

Our baseline estimator, \textit{universal factor analysis} (UFA), is defined as follows.
\begin{equation}\label{eq.estimator}
  \left(\hat{\Lambda}(\cdot),\hat{F}\right)\coloneqq \arg\min_{\Lambda(\cdot)\in\Xi, F\in\mathcal{F}}\frac{1}{M}\sum_{m=1}^{M}\int\rho_{\tau_{m}}(s)\frac{1}{NTh}\sum_{i,t}k\left(\frac{s-(Y_{it}-\lambda_{i}'(\tau_{m})f_{t})}{h}\right)ds, 
\end{equation}
where $\rho_{\tau}(\cdot)$ is the check function at $\tau$. Note that we treat $r$ as known for the moment. A consistent estimator of $r$ is introduced in Section \ref{sec.r}.

\sloppy Let $\hat{R}_{h,t}(f;\Lambda(\cdot))\coloneqq\sum_{m}\int \rho_{\tau_{m}}(s)\sum_{i}k\left((s-(Y_{it}-\lambda_{i}'(\tau_{m})f))/h\right)ds/(MNh)$ and $\hat{R}_{h,i,\tau}(\lambda;F)\coloneqq\int \rho_{\tau}(s)\sum_{t}k\left((s-(Y_{it}-\lambda'f_{t}))/h\right)ds/(Th)$. We propose the following algorithm.
\vspace{0.5em}

\begin{algorithm}[H]\label{alg.1}
\SetAlgoLined
\SetKwComment{Comment}{ \%\%\% }{}
\KwInput{Set the initial values $F^{0}$ and $\Lambda^{0}(\tau_{m})$ for every $\tau_{m}$.}
\While{not converged}{
\textbf{1.1} For each $t$, $f^{(k+1)}_{t,temp}\leftarrow \arg\min_{f\in\mathcal{B}^{r}}\hat{R}_{h,t}(f;\Lambda^{(k)}_{temp}(\cdot))$; $F^{(k+1)}_{temp}\leftarrow(f^{(k+1)}_{t,temp})$\;
\textbf{1.2} For each $i$ and $\tau_{m}$, $\lambda^{(k+1)}_{i,temp}(\tau_{m})\leftarrow \arg\min_{\lambda\in\mathcal{B}^{r}}\hat{R}_{h,i,\tau_{m}}(\lambda;F^{(k+1)}_{temp})$; $\Lambda^{(k+1)}_{temp}(\tau_{m})\leftarrow(\lambda^{(k+1)}_{i,temp}(\tau_{m}))$.
}
\Return $F_{temp}^{(\infty)},\Lambda_{temp}^{(\infty)}(\tau_{m})$ for every $m$\;
\textbf{2.} Normalization\;
  \begin{itemize}
\item[\textbf{2.1}] $L^{(\infty)}(\tau_{m})\leftarrow \Lambda^{(\infty)}_{temp}(\tau_{m})F^{(\infty)'}_{temp}$.
\item[\textbf{2.2}] $\mathcal{L}^{(\infty)}\leftarrow\sum_{m}L^{(\infty)'}(\tau_{m})L^{(\infty)}(\tau_{m})/MNT$; eigendecompose $\mathcal{L}^{(\infty)}$.
\item[\textbf{2.3}] $F^{(\infty)}\leftarrow \sqrt{T}$ times the $T\times r$ eigenvector matrix corresponding to the $r$-largest eigenvalues of $\mathcal{L}^{(\infty)}$.
\item[\textbf{2.4}] $\Lambda^{(\infty)}(\tau_{m})\leftarrow L^{(\infty)}(\tau_{m})F^{(\infty)}/T,\forall m$.
\end{itemize}
\KwOutput{$F^{(\infty)},\Lambda^{(\infty)}(\tau_{m})$ for every $m$.}
\caption{Universal Factor Analysis (UFA)}
\end{algorithm}
\vspace{0.5em}
In Algorithm \ref{alg.1}, steps 1.1 and 1.2 can be carried out by gradient descent type methods by the smoothness of the objective function. 
Note that in these steps, $\lambda$ and $f$ are simply treated as $r\times 1$ real vectors without normalization. Compared to \cite{ando2019quantile} and \cite{chen2021quantile}, the most distinctive feature of Algorithm \ref{alg.1} is the normalization from steps 2.1 to 2.4. The rationale behind step 2.4 is that $L_{0}(\tau)F_{0}/T=\Lambda_{0}(\tau)$ for all $\tau$ by definition.  

\begin{remark}
Under step 2.4, $\Lambda^{(\infty)}_{temp}(\tau)F^{(\infty)'}_{temp}=\Lambda^{(\infty)}(\tau)F^{(\infty)'}$. An important implication is that the pair $(\Lambda^{(\infty)}(\cdot),F^{(\infty)})$, although not obtained by directly solving the minimization problem, yield the same value of the objective function as under $(\Lambda^{(\infty)}_{temp}(\cdot),F^{(\infty)}_{temp})$, which is the minimum. Therefore, 
\begin{align*}
  \hat{f}_{t}=\arg\min_{f\in\mathcal{B}^{r}}\hat{R}_{h,t}\left(f;\hat{\Lambda}(\cdot)\right),\ \ \hat{\lambda}_{i}(\tau_{m})=\arg\min_{\lambda\in\mathcal{B}^{r}}R_{h,i,\tau_{m}}\left(\lambda;\hat{F}\right),\forall m=1,\ldots,M.
\end{align*}
Similar to \cite{chen2021quantile}, this observation is important when deriving the asymptotic properties of our estimator. 
\end{remark}

\begin{remark}\label{rem.initial}
The minimization problem \eqref{eq.estimator} is nonconvex. A good initial guess can improve the performance of the algorithm. Our recommendation is to use a nuclear-norm penalization type estimator of the factors and loadings as the initial guess; this estimator is a by-product of our estimator of $r$; it is consistent in the average squared Frobenius norm and obtained by solving a convex problem; see details in Remark \ref{rem.initial.confirm} in Section \ref{sec.r}. Monte Carlo simulations suggest that the algorithm works well under this initial guess.
\end{remark}

\subsection{Rate of Convergence}\label{sec.roc}
We show in this section that, under the assumption $N$ and $T$ having the same order, the estimator \eqref{eq.estimator} can estimate the space spanned by the factors at $\sqrt{N}$-rate. This is equal to the rate of the QFA estimator in \cite{chen2021quantile} and PCA in \cite{bai2003inferential} under strong factors. We also derive uniform rates for $\hat{\lambda}_{i}(\tau_{m})$ and $\hat{f}_{t}$ that match the preliminary pointwise rates in Lemma S.5 in \cite{chen2021quantile}.

Define $\varepsilon_{it}(\tau_{m})=Y_{it}-\lambda_{0,i}'(\tau_{m})f_{0,t}$. Denote the density of $\varepsilon_{it}(\tau)$ conditional on $F_{0}$ by $\textsf{f}_{\tau,it}(\cdot)$; note that $\textsf{f}_{\tau,it}(0)$ is also equal to the density of $Y_{it}$ conditional on the factors evaluated at the true common component. We now impose the following assumptions.
\begin{assum}\label{assum.iid}
Conditional on $F_{0}$, the $\varepsilon_{it}(\tau_{m})$s are independent across $i$ and $t$ for each $m$.
\end{assum}
\begin{assum}\label{assum.interior}
The true factors and loadings $(\lambda_{0,i}'(\tau_{m}),f_{0,t}')_{m,i,t}$ lie in the interior of $\mathcal{B}^{MN\times r}\times \mathcal{B}^{T\times r}$.
\end{assum}
\begin{assum}\label{assum.density}
(i) For any compact set $C\subset \mathbb{R}$, there exists a $\underline{\textsf{f}}_{C}>0$ such that the conditional density satisfies $\inf_{i,t,\tau\in\mathcal{U},c\in C}\textsf{f}_{\tau,it}(c)\geq \underline{\textsf{f}}_{C}$. (ii) For some positive integer $\gamma\geq 14$, $\textsf{f}_{\tau,it}$ is $\gamma+2$ times continuously differentiable. For $j=0,\ldots,\gamma+2$, the absolute value of the $j$-th derivative $\textsf{f}^{(j)}_{\tau,i,t}(u)$ is uniformly bounded in $i,t$ and $u$.
\end{assum}

\begin{assum}\label{assum.kernel}
(i) The kernel $k$ is symmetric around 0, twice continuously differentiable with $\int_{-\infty}^{\infty} |k^{(1)}(z)|dz<\infty$, $\int_{-\infty}^{\infty}k(z)dz=1$, $\int_{-\infty}^{\infty}s^{j}k(s)ds=0$ for $j=1,\ldots,\gamma-1$ and $\int_{-\infty}^{\infty}s^{\gamma}k(s)ds\neq 0$. (ii) As $N,T\to\infty$, $N\asymp T$ and the bandwidth $h\propto T^{-c}$ where $\gamma^{-1}<c<1/12$.
\end{assum}

Similar to \cite{ando2019quantile} and \cite{chen2021quantile}, the independence assumption in Assumption \ref{assum.iid} is made so that we can adopt some concentration inequalities from the random matrix theory. Note that only conditional independence is assumed, so serial or cross-sectional correlation among $Y_{it}$s are allowed, captured by the correlation among the factors and loadings. Assumption \ref{assum.interior} ensures that the estimators satisfy the first order conditions. Assumptions \ref{assum.density} to \ref{assum.kernel} are similar to \cite{galvao2016smoothed} and \cite{chen2021quantile}. Note that Assumptions \ref{assum.density}-(i) does not rule out the case of unbounded support since the lower bound $\underline{\textsf{f}}_{C}$ is $C$ specific. Assumption \ref{assum.density} implies differentiability of $\lambda_{0,i}(\cdot)$. Assumption \ref{assum.kernel} says that the kernel function is of bounded variation on $\mathbb{R}$ and has order $\gamma$. Compared with the aforementioned literature, our requirements on $c$ and $\gamma$ are stronger, needed for asymptotic normality for the inverse density weighted estimator in the next section; in this section, we can relax them to be $\gamma \geq 4$ and $\gamma^{-1}<c<1/2$. Due to Assumption \ref{assum.kernel}-(ii), which is made for simplicity, we will use $N$ and $T$ exchangeably when discussing rates of convergence. 

Let $\zeta_{NT}\coloneqq \sqrt{1/N}+\sqrt{1/T}$. Let $\|\cdot\|_{F}$ denote the Frobenius norm of a matrix. For a real number $a$, let $\text{sgn}(a)=1$ if $a\geq 0$ and $\text{sgn}(a)=-1$ if $a< 0$. Let $\hat{F}_{j}$ and $F_{0,j}$ be the $j$-th column in the $T\times r$ matrices $\hat{F}$ and $F_{0}$, respectively. Let $H_{NT,1}\coloneqq\text{diag}(\text{sgn}(\hat{F}_{j}'F_{0,j}))$ be an $r\times r$ diagonal matrix.
\begin{thm}\label{thm.roc}
We have the following results under Assumptions \ref{assum.foundation} to \ref{assum.kernel}. 
\begin{itemize}
  \item[(i)] Average rate:
\begin{align*}
  \frac{1}{T}\left\|\hat{F}-F_{0}H_{NT,1}\right\|^{2}_{F}=O_{p}\left(\zeta_{NT}^{2}\right),\ \ \ \ \frac{1}{MN}\sum_{m=1}^{M}\left\|\hat{\Lambda}(\tau_{m})-\Lambda_{0}(\tau_{m})H_{NT,1}^{'-1}\right\|^{2}_{F}=O_{p}\left(\zeta_{NT}^{2}\right).
\end{align*}
\item[(ii)] Uniform rate:
\begin{align*}
\max_{t=1,\ldots,T}\left\|\hat{f}_{t}-H_{NT,1}'f_{0,t}\right\|_{F}=&O_{p}\left(\frac{\zeta_{NT}}{h}\right),\\
\max_{i=1,\ldots,N;m=1,\ldots,M}\left\|\hat{\lambda}_{i}(\tau_{m})-H_{NT,1}^{-1}\lambda_{0,i}(\tau_{m})\right\|_{F}=&O_{p}\left(\frac{\zeta_{NT}}{h}\right),\\
\max_{i=1,\ldots,N;m=1,\ldots,M;t=1,\ldots,T}\left|\hat{\lambda}_{i}'(\tau_{m})\hat{f}_{t}-\lambda_{0,i}'(\tau_{m})f_{0,t}\right|=&O_{p}\left(\frac{\zeta_{NT}}{h}\right).
\end{align*}
\end{itemize}
\end{thm}
Theorem \ref{thm.roc}-(i) shows that, even if (some of) the factors are weak to the mean or some quantiles of $Y$, the space they span can be consistently estimated at the optimal rate, in contrast to \cite{bai2023approximate} and \cite{chen2021quantile}. Note that similar to \cite{chen2021quantile}, Theorem \ref{thm.roc}-(i) can also be shown by replacing the smoothed objective function with the standard check-function-based objective function; this result does not rely on smoothing. 

The uniform rate in Theorem \ref{thm.roc}-(ii), on the other hand, is slower than the optimal rate $\sqrt{\log N}/\sqrt{N}$. However, the current rate is sufficient to establish $\sqrt{N}$-asymptotic normality and to achieve the optimal uniform rate for the inverse density estimator to be introduced in the following section.

\section{The Inverse Density Weighted Estimator}\label{sec.est2}
\subsection{The Value of Inverse Density Weighting}\label{sec.value_of_density}

The main challenge to derive $\sqrt{N}$- and $\sqrt{T}$-asymptotic normality for $\hat{f}_{t}$ and $\hat{\lambda}_{i}(\tau_{m})$ defined in \eqref{eq.estimator} is the heterogeneity of the conditional density $\textsf{f}_{\tau,it}(0)$ in $i$ and $t$.
To see it, let $\eta_{h,\tau_{m},it}\coloneqq K\left((\lambda_{0,i}'(\tau_{m})f_{0,t}-Y_{it})/h\right)-\mathbb{E}\left[K\left((\lambda_{0,i}'(\tau_{m})f_{0,t}-Y_{it})/h\right)\right]$ where $K(c)=\int_{-\infty}^{c}k(z)dz$. For simplicity, assume $H_{NT,1}=I_{r}$. Let $Q_{F,t}\coloneqq \sum_{m=1}^{M}\sum_{i=1}^{N}\textsf{f}_{\tau_{m},it}(0)\lambda_{0,i}(\tau_{m})\lambda_{0,i}'(\tau_{m})/MN$ and $Q_{\Lambda,mi}\coloneqq \sum_{t=1}^{T}\textsf{f}_{\tau_{m},it}(0)f_{0,t}f_{0,t}'/T$. Assumptions \ref{assum.normalization} and \ref{assum.density} ensure that $Q_{F,t}$ and $Q_{\Lambda,mi}$ are invertible for sufficiently large $N$ and $T$. We can show that the Taylor expansion of the first order conditions leads to the following stochastic expansion for $\hat{f}_{t}$:
\begin{align*}
Q_{F,t}\left(\hat{f}_{t}-Q_{F,t}^{-1}A_{1}f_{0,t}\right)
=-\frac{1}{MN}\sum_{m=1}^{M}\sum_{i=1}^{N}\eta_{h,\tau_{m},it}\lambda_{0,i}(\tau_{m})+\text{remainder},
\end{align*}
where $A_{1}=\sum_{i=1}^{N}\sum_{m=1}^{M}\sum_{s=1}^{T}\textsf{f}_{\tau_{m},it}(0)\textsf{f}_{\tau_{m},is}(0)\lambda_{0,i}(\tau_{m})\lambda_{0,i}'(\tau_{m})\hat{f}_{s}f'_{0,s}Q_{\Lambda,mi}^{-1}/MNT$. The first term on the right-hand side is $\sqrt{N}$-asymptotically normal under our assumptions. 
The term $Q_{F,t}^{-1}A_{1}$ serves as an $r\times r$ rotation matrix on $f_{0,t}$. However, this rotation matrix loses the simplicity and interpretability compared to the rotation matrix in PCA \citep{bai2003inferential,bai2023approximate}.

More importantly, terms such as \[A_{2t}=\sum_{i=1}^{N}\sum_{m=1}^{M}\sum_{s=1}^{T}\eta_{h,\tau_{m},it}\cdot\textsf{f}_{\tau_{m},is}(0)Q_{\Lambda,mi}^{-1}f_{0,s}\lambda_{0,i}'(\tau_{m})\left(\hat{f}_{s}-f_{0,s}\right)/MNT\]
appear in the remainder. It is unclear whether the order of $A_{2}$ is $o_{p}(1/\sqrt{N})$: Although $A_{2}$ contains a mean zero random variable $\eta_{h,\tau_{m},it}$, the variable depends on $t$ instead of $s$, so it is not to be averaged out in the time dimension, and is correlated with $\tilde{f}_{s}$. The conditional density $\textsf{f}_{\tau_{m},is}(0)$, on the other hand, depends on all indices to be summed over; it thus prevents us to separately handle the average of $\eta_{h,\tau_{m},it}$ over $i$ and the average of $\hat{f}_{s}-f_{0,s}$ over $s$. Hence, we can only show that $A_{2}$ is $O_{p}(\zeta_{NT})$ by Theorem \ref{thm.roc}.\footnote{An important special case is when $\textsf{f}_{\tau_{m},is}(0)$ also has a factor structure, i.e., there exist some $a_{\tau_{m},i}$ and $b_{\tau_{m},s}$ such that $\textsf{f}_{\tau_{m},is}(0)=a_{\tau_{m},i}'b_{\tau_{m},s}$. In that case, $A_{2}$ and in fact the whole remainder term are $o_{p}(\zeta_{NT})$ and thus UFA is asymptotically normal. One such an example is a model where $\textsf{f}_{\tau_{m},is}(0)$ only depends on $i$ or $s$; the Monte Carlo design for normal approximation in Section 5.3 in \cite{chen2021quantile} satisfies this condition. Another example is when $r=1$: One can verify that $\textsf{f}_{\tau_{m},is}(0)=1/(\lambda_{0,i}'(\tau_{m})f_{0,t})$ (see equation \eqref{lem.density.formula}), so it has a factor structure if $r=1$.} 

Now if, instead, the objective function in \eqref{eq.estimator} is weighted by $1/\textsf{f}_{\tau_{m},it}(0)$ for each $m,i$ and $t$, then $Q_{F,t}$ becomes $\Phi\coloneqq \sum_{m=1}^{M}\sum_{i=1}^{N}\lambda_{0,i}(\tau_{m})\lambda_{0,i}'(\tau_{m})/MN$ whereas $Q_{\Lambda,mi}=I_{r}$. 
 For $A_{2}$, since $\textsf{f}_{\tau,is}(0)$ is now cancelled out, we can rewrite the summations as follows
\begin{align*}
A_{2}=&\left(\frac{1}{T}\sum_{s=1}^{T}\left[f_{0,s}\left(\hat{f}_{s}-f_{0,s}\right)'\right]\right)\cdot\left(\frac{1}{MN}\sum_{i=1}^{N}\sum_{m=1}^{M}\frac{\eta_{h,\tau_{m},it}}{\textsf{f}_{\tau_{m},it}(0)}\lambda_{0,i}(\tau_{m})\right)=o_{p}\left(\frac{1}{\sqrt{N}}\right).
\end{align*}
Indeed, we can show that the whole remainder is $o_{p}(1/\sqrt{N})$. Moreover, $Q_{F,t}^{-1}A_{1}$ now becomes $\left(F_{0}'\hat{F}/T\right)'$. We thus have the following equation which admits $\sqrt{N}$-asymptotic normality of the estimator.
\begin{equation}\label{eq.expansion.eg.weighted}
   \Phi\cdot\left(\hat{f}_{t}-\left(\frac{F_{0}'\hat{F}}{T}\right)'f_{0,t}\right)=-\frac{1}{MN}\sum_{m=1}^{M}\sum_{i=1}^{N}\frac{\eta_{h,\tau_{m},it}}{\textsf{f}_{\tau,it}(0)}\lambda_{0,i}(\tau_{m})+o_{p}\left(\frac{1}{\sqrt{N}}\right).
 \end{equation} 
 
 One nice feature of this expansion is that the rotation matrix $F_{0}'\hat{F}/T$ is exactly equal to the rotation matrix $H_{NT,2}$ for PCA in Lemma 3 in \cite{bai2023approximate} under $F_{0}'F_{0}/T=I_{r}$. \cite{bai2023approximate} show that this matrix is equivalent to\footnote{Equivalence is in the sense that the difference of those two rotation matrices is $o_{p}(1/\sqrt{N})$.} the rotation matrix in \cite{bai2003inferential} for PCA for an AFM under strong factors. Hence, it draws a close analogy between our inverse density weighted quantile estimator and PCA. We will further explain the intuition behind after we formally introduce our estimator in this section.
\begin{remark}
The difficulties in determining the stochastic order of $A_{2}$ and in simplifying $A_{1}$ are not caused by the possibility of the existence of weak factors. Nor is it related to the fact that we simultaneously estimate the loadings at multiple $\tau_{m}$s or the specific smoothing method we adopt. The same difficulties exist in \cite{chen2021quantile} as well; indeed, the inverse density weighted estimator we introduce later in this section applies to their setup as well. 
\end{remark}

The above analysis is based on that $\textsf{f}_{\tau,it}(0)$ is known which in most applications is not the case. This motivates us to consider how to estimate those densities in a way such that the estimation error does not lead to further complications. 

 \subsection{Inverse Density Estimation}
The observation in Section \ref{sec.value_of_density} motivates us to reconstruct the objective function \eqref{eq.estimator} by weighting the kernel function at each $(m,i,t)$ by a consistent estimator of $1/\textsf{f}_{\tau_{m},it}(0)$. However, the estimation error of these inverse densities can be correlated with $\eta_{h,\tau_{m},it}$, causing technical challenges. Specifically, in the asymptotic expansion \eqref{eq.expansion.eg.weighted}, the first term on the right-hand side, which drives the asymptotic distribution, is now modified as follows:
\begin{equation}\label{eq.prob}
  \frac{1}{MN}\sum_{m=1}^{M}\sum_{i=1}^{N}\frac{\eta_{h,\tau_{m},it}}{\textsf{f}_{\tau_{m},it}(0)}\lambda_{0,i}(\tau_{m})+\frac{1}{MN}\sum_{m=1}^{M}\sum_{i=1}^{N}\left(\widehat{\frac{1}{\textsf{f}_{\tau_{m},it}(0)}}-\frac{1}{\textsf{f}_{\tau_{m},it}(0)}\right)\eta_{h,\tau_{m},it}\lambda_{0,i}(\tau_{m}),
\end{equation}
where $\widehat{1/\textsf{f}_{\tau_{m},it}(0)}$ is some uniformly consistent estimator of $1/\textsf{f}_{\tau_{m},it}(0)$. If $\widehat{1/\textsf{f}_{\tau_{m},it}(0)}$ and $\eta_{h,\tau_{m},it}$ are correlated, it is unclear whether the second term is $o_{p}(1/\sqrt{N})$. 

In this subsection, we introduce a novel subsample estimation approach to estimate $1/\textsf{f}_{\tau,it}(0)$ so that the problem above is avoided; the key observation is that we only need \textit{three quarters} of data to estimate the whole set of factors and loadings. This idea may be of independent interest and can be applied to other two-stage estimation procedures that involve estimating some latent factor structure via in the first stage. An idea sharing a similar spirit can be found in \cite{bai2021matrix} where they focus on factor analysis with missing data. 

\subsubsection*{Representation of the Inverse Density}

Recall that $\textsf{f}_{\tau,it}(0)$ is equal to the conditional density of $Y_{it}$ at $L_{0,it}(\tau)\coloneqq\lambda_{0,i}'(\tau)f_{0,t}$ given $f_{0,t}$. Since only $\lambda_{0,i}(\cdot)$ depends on $\tau$, one can show that  (e.g. \cite{koenker1999goodness})
\begin{equation}\label{lem.density.formula}
  \frac{1}{\textsf{f}_{\tau,it}(0)}=\lambda_{0,i}^{(1)'}(\tau)f_{0,t},
\end{equation}
where $\lambda_{0,i}^{(1)}(\tau)$ is the derivative of $\lambda_{0,i}(\cdot)$ evaluated at $\tau$. For each $\tau$, we can approximate this derivative by numerical differentiation. \cite{koenker1999goodness} approximate it by the three-point central difference formula; letting $h_{d}$ be some bandwidth, the approximation error is $O(h_{d}^{2})$. We propose to use a five-point difference formula (FPDF) to reduce the approximation error to $O(h_{d}^{4})$. For instance, we can use the central-difference-FPDF to approximate $\lambda_{0,i}^{(1)}(\tau)$ as follows:\footnote{For $\tau$ that is close to 0 or 1, one can instead use forward-difference-FPDF or backward-difference-FPDF, respectively, to mitigate the performance drop near the boundaries. The forward-difference formula reads \[
  \frac{-25\lambda_{0,i}(\tau)+48\lambda_{0,i}(\tau+h_{d})-36\lambda_{0,i}(\tau+2h_{d})+16\lambda_{0,i}(\tau+3h_{d})-3\lambda_{0,i}(\tau+4h_{d})}{12h_{d}},
\]
whereas the backward-difference formula is obtained by flipping the sign of $h_{d}$. Both forward and backward formulae have approximation error $O(h_{d}^{4})$ as the central-difference version.}
\begin{equation}\label{eq.FPDF}
  \frac{-\lambda_{0,i}(\tau+2h_{d})+8\lambda_{0,i}(\tau+h_{d})-8\lambda_{0,i}(\tau-h_{d})+\lambda_{0,i}(\tau-2h_{d})}{12h_{d}}.
\end{equation}

We can thus estimate the inverse density $1/\textsf{f}_{\tau_{m},it}(0)$ by substituting proper estimators of the factors and loadings at $\tau_{m}$ into \eqref{eq.FPDF}. 
\begin{remark}
We estimate the inverse density by treating it as an entity and directly estimating it without actually taking the inverse of some estimated density. This guarantees numerical stability in implementation. An alternative representation of the density is $1/\textsf{f}_{\tau,it}(0)=(\sum_{s=1}^{T}f_{0,s}'/T)(\sum_{s=1}^{T}\textsf{f}_{\tau,is}(0)f_{0,s}f_{0,s}'/T)^{-1}f_{0,t}$; see \cite{fernandes2021smoothing}. However, an estimator based on this representation involves inverting a matrix so is numerically less stable.
\end{remark}

\subsubsection*{A New Sample Splitting Strategy}
To avoid dependence between density estimation and the second stage estimation, we develop the following subsample estimation approach. Let $\mathcal{N}_{1}\coloneqq\{1,\ldots,\lfloor N/2\rfloor\}$, $\mathcal{N}_{2}\coloneqq\{\lfloor N/2\rfloor+1,\ldots,N\}$, $\mathcal{T}_{1}\coloneqq\{1,\ldots,\lfloor T/2\rfloor\}$, $\mathcal{T}_{2}\coloneqq\{\lfloor T/2\rfloor+1,\ldots,T\}$, where $\lfloor c\rfloor$ equals the largest integer that is no greater than $c$. Let $N_{a}=|\mathcal{N}_{a}|$ and $T_{b}=|\mathcal{T}_{b}|,a,b\in\{1,2\}$. Divide the $N\times T$ data matrix $Y$ into four regions:
\begin{align*}
Y=\begin{pmatrix}
Top\ Left&Top\ Right\\
Bottom\ Left& Bottom\ Right
\end{pmatrix},
\end{align*}
where formally, for instance, $Top\ Left=\{Y_{it}:i\in\mathcal{N}_{1},t\in\mathcal{T}_{1}\}$. Let the subsample $Top\coloneqq Top\ Left\cup Top\ Right$. Subsamples $Bottom$, $Left$ and $Right$ are defined similarly.

The key observation is that we can estimate \textit{the full set} of factors $\{f_{0,t}\}$ using, for example, $Top$ only, and estimate the full set of loadings $\{\lambda_{0,i}(\tau_{m}\pm h_{d}),\lambda_{0,i}(\tau_{m}\pm 2h_{d})\}$ using $Left$ only. Together, they consist of only approximately $3/4$ of the full data set. Under Assumption \ref{assum.iid}, these estimated factors and loadings are by construction independent of all $Y_{it}$ in $Bottom\ Right$ even if they share the same $t$ and $i$ indices. 

One subtlety when applying this idea is the mismatch of rotation matrices. If we separately obtain the factors and loadings using $Top$ and $Left$ by UFA, respectively, their rotation matrices may not cancel out when computing the product. Hence, we propose a sequential process.
For illustration, suppose we are to estimate $f_{0,t}$ and $\lambda_{0,i}(\cdot)$ for every $(i,t)\in\mathcal{N}_{2}\times \mathcal{T}_{2}$, i.e., $Y_{it}\in Bottom\ Right$. We first obtain the whole set $\{\hat{f}_{s}^{top}\}$ where $s=1,\ldots,T$ by UFA using $Top$. Then, taking out the $i$-th row in $Left$, regress these $T_{1}$ data points of $Y$ onto the first half of the estimated factors $\{\hat{f}^{top}_{1},\ldots,\hat{f}^{top}_{\lfloor T/2\rfloor}\}$ by smoothed quantile regression. Denote the obtained loading by $\hat{\lambda}^{(t,l)}_{i}(\cdot)$, where $(t,l)$ indicates that the estimate is obtained using $Y$ in $Left$ and the estimated factors using $Top$. The estimators $\hat{\lambda}^{(t,l)}_{i}(\tau)$ and $\hat{f}^{top}_{t}$ are independent of any data point in $Bottom\ Right$ even though $(i,t)$ falls into that region. Now the product $\hat{\lambda}^{(t,l)'}_{i}(\tau)\hat{f}^{top}_{t}$ can consistently estimate $\lambda_{0,i}'(\tau)f_{0,t}$ for any $\tau$ because the rotation matrices for $\hat{\lambda}^{(t,l)}_{i}(\cdot)$ and $\hat{f}^{top}_{t}$ automatically cancel out by construction when doing the multiplication.

Generally, for an arbitrary fixed $(i,t)\in \mathcal{N}_{a}\times \mathcal{T}_{b}$, we estimate the inverse density by the following estimator:
\begin{equation}\label{eq.density.est.N}
  \widehat{\frac{1}{\textsf{f}_{\tau_{m},it}(0)}}\coloneqq\left[\frac{-\hat{\lambda}^{w}_{i}(\tau_{m}+2h_{d})+8\hat{\lambda}^{w}_{i}(\tau_{m}+h_{d})-8\hat{\lambda}^{w}_{i}(\tau_{m}-h_{d})+\hat{\lambda}^{w}_{i}(\tau_{m}-2h_{d})}{12h_{d}}\right]'\hat{f}_{t}^{v},
\end{equation}
where 
\begin{equation*}
  (w,v)=\begin{cases} \left((b,r),bottom\right),&\text{if }a=1,b=1;\\
  \left((b,l),bottom\right),&\text{if }a=1,b=2;\\
  \left((t,r),top\right),&\text{if }a=2,b=1;\\
  \left((t,l),top\right),&\text{if }a=2,b=2.
  \end{cases}
\end{equation*}

We now make the following assumption and derive the uniform rate of convergence of \eqref{eq.density.est.N}.

\begin{assum}\label{assum.subsample}
For sufficiently large $N$ and $T$, the $r$ largest eigenvalues of matrices $\sum_{m=1}^{M}\sum_{i\in\mathcal{N}_{a}}\lambda_{0,i}(\tau_{m})\lambda_{0,i}'(\tau_{m})/MN$ and $\sum_{t\in\mathcal{T}_{b}}f_{0,t}f_{0,t}'/T$
are bounded away from 0, and the $r$ nonzero eigenvalues of $F_{0}\sum_{m=1}^{M}\sum_{i\in\mathcal{N}_{a}}\lambda_{0,i}(\tau_{m})\lambda_{0,i}'(\tau_{m})F_{0}'/MNT$ are distinct for all $a,b\in\{1,2\}$.
\end{assum}
Assumption \ref{assum.subsample} is the subsample counterpart of Assumption \ref{assum.normalization}. Note that the matrices $\sum_{m=1}^{M}\sum_{i\in\mathcal{N}_{a}}\lambda_{0,i}(\tau_{m})\lambda_{0,i}'(\tau_{m})/MN$ and $\sum_{t\in\mathcal{T}_{b}}f_{0,t}f_{0,t}'/T$ are in general no longer diagonal. Let $\psi_{NT}\coloneqq \zeta_{NT}/(hh_{d})+h_{d}^{4}$. We have the following theorem.
\begin{thm}\label{thm.firststage}
Under Assumptions \ref{assum.foundation} to \ref{assum.subsample}, $\max_{m,i,t}|\widehat{1/\textsf{f}_{\tau_{m},it}(0)}-1/\textsf{f}_{\tau_{m},it}(0)|=O_{p}(\psi_{NT})$.
\end{thm}

Finally, let us revisit the problem raised in the beginning of this subsection. We can see that the second term in equation \eqref{eq.prob} can be split into two parts:
\[
  \frac{1}{N}\sum_{i\in \mathcal{N}_{1}}\sum_{m=1}^{M}\left(\widehat{\frac{1}{\textsf{f}_{\tau_{m},it}(0)}}-\frac{1}{\textsf{f}_{\tau_{m},it}(0)}\right)\eta_{h,\tau_{m},it}\lambda_{0,i}(\tau_{m})+\frac{1}{N}\sum_{i\in \mathcal{N}_{2}}\sum_{m=1}^{M}\left(\widehat{\frac{1}{\textsf{f}_{\tau_{m},it}(0)}}-\frac{1}{\textsf{f}_{\tau_{m},it}(0)}\right)\eta_{h,\tau_{m},it}\lambda_{0,i}(\tau_{m}).
\]
These two parts are correlated because the estimated density functions in each part are correlated with the $\eta_{h,\tau_{m},it}$ in the other part by construction. However, within each part, all the $\widehat{1/\textsf{f}_{\tau_{m},it}(0)}$s are independent of the $\eta_{h,\tau,it}$s. Hence, conditional on the $\widehat{1/\textsf{f}_{\tau_{m},it}(0)}$s, both parts are $o_{p}(1/\sqrt{N})$ by the Hoeffding's inequality.
\begin{remark}
Our particular sample splitting approach is designed to overcome the $(i,t)$-specific incidental nuisance parameter problem. It is different from the standard sample splitting approaches such as \cite{chernozhukov2018double} where the nuisance parameters to be estimated, for instance a conditional expectation function, stay constant across subsamples. Our approach also differs from the split-panel jackknife method in \cite{dhaene2015split} where estimation is done using half panels (in our nation, $Left$ and $Right$) to remove bias in maximum-likelihood estimation of nonlinear models with fixed effects. 
\end{remark}

\subsection{Inverse Density Weighted Universal Factor Analysis}

Once we obtain the estimated densities, we use \textit{the full sample} to estimate the factors and loadings. Hence, our estimator does not lose efficiency. Specifically, we define our inverse density weighted estimator (IDW-UFA) as follows:
  \begin{equation}\label{eq.estimator2}
  \left(\tilde{\Lambda}(\cdot),\tilde{F}\right)\coloneqq \arg\min_{\Lambda(\cdot)\in\Xi, F\in\mathcal{F}}\frac{1}{M}\sum_{m=1}^{M}\int\rho_{\tau_{m}}(s)\frac{1}{NTh}\sum_{i=1}^{N}\sum_{t=1}^{T}\widehat{\frac{1}{\textsf{f}_{\tau_{m},it}(0)}}k\left(\frac{s-(Y_{it}-\lambda_{i}'(\tau_{m})f_{t})}{h}\right)ds. 
\end{equation}

 Define $\tilde{R}_{h,t}(f;\Lambda(\cdot))$ and $\tilde{R}_{h,i,\tau}(\lambda;F)$ similar to $\hat{R}_{h,t}(f;\Lambda(\cdot))$ and $\hat{R}_{h,i,\tau}(\lambda;F)$ in Section \ref{sec.estimator1.estimator} with the estimated inverse densities as weights. We implement our estimator by the following algorithm.

\vspace{0.5em}
\begin{algorithm}[H]\label{alg.2}
\SetAlgoLined
\SetKw{State}{State}
\KwInput{Density estimation\;
\textbf{1.} Obtain $\hat{f}^{top}_{t}$ and $\hat{f}^{bottom}_{t}$ for all $t$ by running Algorithm \ref{alg.1} twice, using $Top$ and $Bottom$ respectively; $\hat{F}^{top}\leftarrow(\hat{f}^{top}_{t}),\hat{F}^{bottom}\leftarrow(\hat{f}^{bottom}_{t})$\;
\textbf{2.} Construct $\widehat{1/\textsf{f}_{\tau_{m},it}(0)}$ by equation \eqref{eq.density.est.N} for each $m,i,t$\;
\textbf{3.} Set the initial values $F^{0}$ and $\Lambda^{0}(\tau_{m})$ for every $\tau_{m}$. }
\While{not converged}{
Repeat steps 1.1 to 1.2 in Algorithm \ref{alg.1} using the full sample with $\hat{R}_{h,t}(f;\Lambda^{(k)}(\cdot))$ and $\hat{R}_{h,i,\tau}(\lambda;F_{temp}^{(k+1)})$ in the algorithm replaced by $\tilde{R}_{h,t}(f;\Lambda^{(k)}(\cdot))$ and $\tilde{R}_{h,i,\tau}(\lambda;F_{temp}^{(k+1)})$, respectively.
}
\Return $F_{temp}^{(\infty)},\Lambda_{temp}^{(\infty)}(\tau_{m})$ for every $m$\;
\textbf{4.} Follow step 2 in Algorithm \ref{alg.1} for normalization.

\KwOutput{$F^{(\infty)},\Lambda^{(\infty)}(\tau_{m})$ for every $m$.}
\caption{Inverse Density Weighted Universal Factor Analysis (IDW-UFA)}
\end{algorithm}

\subsection{Asymptotic Theory of IDW-UFA}
In this subsection, we show that our inverse density weighted estimator can both estimate the spaces spanned by the factors and loadings at $\sqrt{N}$ or $\sqrt{T}$ rate, and is pointwise asymptotically normal up to rotation.

One key step to achieve asymptotic normality is to recenter the estimator. Following a similar argument as in Theorem \ref{thm.roc}, we can show that, for instance, $\|\tilde{F}-F_{0}\tilde{H}_{NT,1}\|_{F}/\sqrt{T}=O_{p}(\zeta_{NT})$ for a diagonal matrix $\tilde{H}_{NT,1}$ whose $j$-th diagonal entry is $\text{sgn}(\tilde{F}_{j}'F_{0,j})$. However, to achieve asymptotic normality for $\tilde{f}_{t}$ for each $t$, letting $H_{NT,2}\coloneqq F_{0}'\tilde{F}/T$, we show that we need to recenter $\tilde{f}_{t}$ around $H_{NT,2}'f_{0,t}$ instead of $\tilde{H}_{NT,1}'f_{0,t}$. This result echoes \cite{bai2003inferential} and \cite{bai2023approximate} as $H_{NT,2}$ is equivalent to all their rotation matrices under $F_{0}'F_{0}/T=I_{r}$ as mentioned in Section \ref{sec.value_of_density}.

Recall $\Phi\coloneqq \sum_{m=1}^{M}\sum_{i=1}^{N}\lambda_{0,i}(\tau_{m})\lambda'_{0,i}(\tau_{m})/(MN)$. For any $\tau^{*}\in\{\tau_{1},\ldots,\tau_{M}\}$, let
\begin{align*}
\Sigma_{F,t}&\coloneqq \frac{1}{M^{2}N}\sum_{m=1}^{M}\sum_{m'=1}^{M}\sum_{i=1}^{N}\frac{\left(\min(\tau_{m},\tau_{m'})-\tau_{m}\tau_{m'}\right)\lambda_{0,i}(\tau_{m})\lambda_{0,i}'(\tau_{m'})}{\textsf{f}_{\tau_{m},it}(0)\textsf{f}_{\tau_{m'},it}(0)},\\
\Sigma_{\Lambda,\tau^{*},i}&\coloneqq \tau^{*}(1-\tau^{*})\cdot \frac{1}{T}\sum_{t=1}^{T}\frac{f_{0,t}f_{0,t}'}{\textsf{f}_{\tau^{*},it}^{2}(0)}.
\end{align*}

\begin{thm}\label{thm.dist2}
Under Assumptions \ref{assum.foundation} to \ref{assum.subsample}, if $h_{d}\propto T^{-d}$ and $\zeta_{NT}/h^{5}<h_{d}<h$, we have the following results.
\begin{itemize}
\item[(i)] Average rate: For all $m=1,\ldots,M$,
\begin{align*}
\frac{1}{T}\left\|\tilde{F}-F_{0}H_{NT,2}\right\|_{F}^{2}=O_{p}\left(\frac{1}{N}\right),\ \ \frac{1}{N}\left\|\tilde{\Lambda}(\tau_{m})-\Lambda_{0}(\tau_{m})H_{NT,2}^{'-1}\right\|_{F}^{2}=O_{p}\left(\frac{1}{T}\right).
\end{align*}
\item[(ii)] Limiting distributions: For all fixed $i,t$ and $\tau^{*}\in\{\tau_{1},\ldots,\tau_{M}\}$,
\begin{align*}
\Sigma_{F,t}^{-1/2}\Phi H_{NT,2}^{'-1}\cdot\sqrt{N}\left(\tilde{f}_{t}-H_{NT,2}'f_{0,t}\right)&\overset{d}{\to}\mathcal{N}\left(0,I_{r}\right),\\
\Sigma_{\Lambda,\tau^{*},i}^{-1/2}H_{NT,2}^{'-1}\cdot\sqrt{T}\left(\tilde{\lambda}_{i}(\tau^{*})-H_{NT,2}^{-1}\lambda_{0,i}(\tau^{*})\right)&\overset{d}{\to}\mathcal{N}\left(0,I_{r}\right),\\
\frac{\tilde{\lambda}_{i}'(\tau^{*})\tilde{f}_{t}-\lambda_{0,i}'(\tau^{*})f_{0,t}}{\sqrt{\frac{1}{N}\lambda_{0,i}'(\tau^{*})\Phi^{-1}\Sigma_{F,t}\Phi^{-1}\lambda_{0,i}(\tau^{*})+\frac{1}{T}f_{0,t}'\Sigma_{\Lambda,\tau^{*},i}f_{0,t}}}&\overset{d}{\to}\mathcal{N}(0,1).
\end{align*} 
\end{itemize}
The rotation matrix $H_{NT,2}=\tilde{H}_{NT,1}+o_{p}\left(1\right)$. In particular, when $r=1$, $H_{NT,2}=\tilde{H}_{NT,1}+O_{p}\left(\zeta_{NT}^{2}\right)$.
\end{thm}
\begin{remark}\label{rem.uniformrate.good}
Theorem \ref{thm.dist2}-(ii) implies $\sqrt{N}$ pointwise convergence rate for the factor, loading and common component estimators. We can also show that the uniform rate of these estimators is $\sqrt{\log N/N}$.
\end{remark}
\begin{remark}\label{rem.bandwidth}
The bandwidth requirement on $h_{d}$ is to guarantee that the inverse density  estimation error has an asymptotically negligible impact on the estimator of factors and loadings. Under Assumption \ref{assum.kernel}-(ii), the inequality constraints for $h_{d}$ are feasible and contain the optimal $h_{d}$ that satisfies $h_{d}^{4}\asymp \zeta_{NT}/(hh_{d})$.
\end{remark}
\begin{remark}\label{rem.covariance.estimator}
All the covariance matrices in Theorem \ref{thm.dist2}-(ii) can be consistently estimated by simply plugging our estimated factors, loadings and inverse densities. Using the variance of $\sqrt{N}(\tilde{f}_{t}-H_{NT,2}'f_{0,t})$ as an example, notice that 
\begin{align*}
&\left(\Sigma_{F,t}^{-1/2}\Phi H_{NT,2}^{'-1}\right)^{-1}\left(\Sigma_{F,t}^{-1/2}\Phi H_{NT,2}^{'-1}\right)^{'-1}\\
=&H_{NT,2}^{'}\Phi^{-1}\Sigma_{F,t}\Phi^{-1}H_{NT,2}\\
  =&H_{NT,2}^{'}\Phi^{-1}H_{NT,2}H_{NT,2}^{-1}\Sigma_{F,t}H_{NT,2}^{-1'}H_{NT,2}'\Phi^{-1}H_{NT,2}\\
  =&\left(H_{NT,2}^{-1}\Phi H_{NT,2}^{'-1}\right)^{-1}\left(H_{NT,2}^{-1}\Sigma_{F,t}H_{NT,2}^{-1'}\right)\left(H_{NT,2}^{-1}\Phi H_{NT,2}^{'-1}\right)^{-1}.
\end{align*}
Both $H_{NT,2}^{-1}\Phi H_{NT,2}^{'-1}$ and $H_{NT,2}^{-1}\Sigma_{F,t}H_{NT,2}^{-1'}$ can be consistently estimated by plugging in $\tilde{\lambda}_{i}(\tau_{m})$ and $\widehat{1/\textsf{f}_{\tau_{m},it}(0)}$ because they are respectively consistent of $H_{NT,2}^{-1}\lambda_{0,i}(\tau_{m})$ and $1/\textsf{f}_{\tau_{m},it}(0)$ uniformly in $m$ and $i$. Consistency of plug-in estimators of the other two variances follows similarly.
\end{remark}

\begin{remark}
As mentioned earlier, our $H_{NT,2}$ is identical to that in Lemma 3 in \cite{bai2023approximate} under $F_{0}'F_{0}/T=I_{r}$, and the latter is shown to be equivalent to the rotation matrix in \cite{bai2003inferential} for PCA in AFMs under strong factors. Indeed, our inverse density weighted estimator can be asymptotically equivalent to an \textit{infeasible} PCA estimator. For illustration, let $M=1$ and assume strong factor at $\tau\coloneqq \tau_{1}$. Define $Y^{*}_{\tau,it}\coloneqq \lambda_{0,i}'(\tau)f_{0,t}+\varepsilon^{*}_{\tau,it}$ with $\varepsilon^{*}_{\tau,it}\coloneqq -\eta_{h,\tau,it}/\textsf{f}_{\tau,it}(0)$ where recall that $\eta_{h,\tau_{m},it}\coloneqq K\left((\lambda_{0,i}'(\tau_{m})f_{0,t}-Y_{it})/h\right)-\mathbb{E}\left[ K\left((\lambda_{0,i}'(\tau_{m})f_{0,t}-Y_{it})/h\right)\right]$. Note that $Y^{*}_{\tau,it}$ is unobserved. Based on this model, since $\eta_{h,\tau,it}$ is bounded and has mean zero by construction, one can verify that the asymptotic distribution of the infeasible PCA estimator of $\lambda_{0,i}(\tau)$ and $f_{0,t}$ are equal to ours under $F_{0}'F_{0}/T=I_{r}$ and $\Phi\coloneqq \Lambda_{0}'(\tau)\Lambda_{0}(\tau)/N$ being diagonal with distinct diagonal entries. 
\end{remark}

 \section{Estimating the Mean Factor Loadings}\label{sec.mean}
In some applications, the parameters of interest are the factors and loadings affecting the mean, rather than the quantiles, of the outcome variable. However, directly estimating an AFM using PCA requires strong factors. Our method provides an alternative approach to estimate the mean factor loadings.

Under \eqref{eq.model}, we have 
\begin{equation}\label{eq.mean.model}
  Y_{it}=\bar{\lambda}_{0,i}'f_{0,t}+\nu_{it},\ \ \mathbb{E}(\nu_{it}|f_{0,t})=0,
\end{equation}
if $\mathbb{E}(Y_{it}|f_{0,t})$ exists. The mean factor loading $\bar{\lambda}_{0,i}$ is by construction $\int_{0}^{1}\lambda_{0,i}(\tau)d\tau$. We can estimate $\bar{\lambda}_{0,i}$ for each $i$ by solving the following least square problem:
\begin{equation}\label{eq.estimatormean}
  \min_{\lambda}\frac{1}{T}\sum_{t=1}^{T}\left(Y_{it}-\lambda'\tilde{f}_{t}\right)^{2},
\end{equation}
where $\tilde{f}_{t}$ is obtained by estimator \eqref{eq.estimator2}. Our normalization gives a simple analytical solution:
\begin{equation}\label{eq.meanlambda}
  \tilde{\bar{\lambda}}_{i}=\frac{1}{T}\sum_{t=1}^{T}\tilde{f}_{t}Y_{it}.
\end{equation}
Let the mean common component matrix be $\bar{L}_{0}\coloneqq \bar{\Lambda}_{0}F_{0}'$. The estimator $\tilde{\bar{\lambda}}_{i}$ has the following properties.
\begin{thm}\label{thm.dist3}
Suppose $Y_{it}$ is bounded. Then under the conditions in Theorem \ref{thm.dist2},
\begin{align*}
&\frac{1}{N}\left\|\tilde{\bar{\Lambda}}-\bar{\Lambda}_{0}H_{NT,2}^{'-1}\right\|_{F}^{2}=O_{p}\left(\frac{1}{T}\right),\\
&\bar{\Sigma}_{\Lambda,i}^{-1/2}H_{NT,2}^{'-1}\cdot \sqrt{T}\left(\tilde{\bar{\lambda}}_{i}-H_{NT,2}^{-1}\bar{\lambda}_{0,i}\right)\overset{d}{\to}\mathcal{N}\left(0,I_{r}\right),\\
&\frac{\tilde{\bar{\lambda}}_{i}'\tilde{f}_{t}-\bar{L}_{0,it}}{\sqrt{\frac{1}{N}\bar{\lambda}_{0,i}'\Phi^{-1}\Sigma_{F,t}\Phi^{-1}\bar{\lambda}_{0,i}+\frac{1}{T}f_{0,t}'\bar{\Sigma}_{\Lambda,i}f_{0,t}}}\overset{d}{\to}\mathcal{N}\left(0,1\right),
\end{align*}
where $H_{NT,2}$, $\Phi$ and $\Sigma_{F,t}$ are the same as in Theorem \ref{thm.dist2} and $\bar{\Sigma}_{\Lambda,i}\coloneqq \sum_{t=1}^{T}\mathbb{E}\left(\nu_{it}^{2}f_{0,t}f_{0,t}'\right)/T$. 
\end{thm}
Two remarks are in order. First, the boundedness assumption on $Y_{it}$ in the theorem is for simplicity; it can be replaced by, for example, the existence of higher order moments of $Y_{it}$. Second, all the variances can be consistently estimated by plugging in the estimated factors, loadings and conditional densities, under a similar argument as Remark \ref{rem.covariance.estimator}.  
 \section{Estimating the Number of (Strong) Factors}\label{sec.r}
 So far, we have assumed that $r$ is known or can be consistently estimated. In this section, we first propose a consistent estimator of $r$ that is robust to weak factors. We achieve this by estimating the common component $L_{0}(\tau_{m})$ for each $m$ by a nuclear norm penalized estimator that does not require strong factors. We also introduce estimators of the number of factors that have any tolerated level of influence on the conditional quantile or the conditional mean of the outcome variable; these ``strong'' factor selectors can be useful in applications where the researcher would like to include factors that have relatively large influence.

\subsection{Estimating the Number of Factors}
 For a constant $C>0$ and compact interval $\mathcal{B}_{L}\subset\mathbb{R}$, for each $m=1,\ldots,M$, define
 \begin{equation}\label{eq.obj.nuc}
   \hat{L}^{pel}(\tau_{m})=\arg\min_{L\in \mathcal{B}_{L}^{N\times T}}\frac{1}{NT}\sum_{i=1}^{N}\sum_{t=1}^{T}\rho_{\tau_{m}}\left(Y_{it}-L_{it}\right)+\frac{C\sqrt{\log(NT)}\max\{\sqrt{N},\sqrt{T}\}}{NT}\|L\|_{*}, 
 \end{equation}
where $\|\cdot\|_{*}$ is the nuclear norm of a matrix. Applying the results in \cite{feng2023nuclear} without regressors, we can show that $\hat{L}^{pel}(\tau_{m})$ is consistent of $L_{0}(\tau_{m})$ in the average squared Frobenius norm with the rate equal to $O_{p}(\log(NT)\max\{1/N,1/T\})$ uniformly in $m$, regardless of the order of the singular values of $L_{0}(\tau_{m})$. We can then estimate $\sum_{m=1}^{M}L_{0}'(\tau_{m})L_{0}(\tau_{m})/(MNT)$ by $\sum_{m=1}^{M}\hat{L}^{pel'}(\tau_{m})\hat{L}^{pel}(\tau_{m})/(MNT)$, where under $F_{0}'F_{0}/T=I_{r}$, all the nonzero eigenvalues of the former have order $O(1)$ by Lemma \ref{lem.discrete}. Therefore, between the $r$-th and the $(r+1)$-th largest  eigenvalues of $\sum_{m=1}^{M}\hat{L}^{pel'}(\tau_{m})\hat{L}^{pel}(\tau_{m})/(MNT)$, denoted by $\hat{\sigma}^{2}_{r}$ and $\hat{\sigma}^{2}_{r+1}$, we can show that there is a sufficiently large gap with probability approaching 1. We thus propose the following thresholding estimator for $r$:
\begin{equation}\label{eq.estimator.r}
  \hat{r}=\sum_{j=1}^{\min\{N,T\}}1\left(\hat{\sigma}_{j}^{2}\geq C_{r}\right),
\end{equation}
where $C_{r}$ is any sequence of $(N,T)$ satisfying $C_{r}\to 0$  and $\sqrt{\log(NT)}/(C_{r}\sqrt{\min\{N,T\}})\to 0$. The following theorem shows consistency of $\hat{r}$.
\begin{thm}\label{thm.r}
Under Assumptions \ref{assum.foundation}, \ref{assum.normalization}, \ref{assum.lip}, \ref{assum.iid} and \ref{assum.density}, $\Pr(\hat{r}=r)\to 1$.
\end{thm}
\begin{remark}
Our estimator of the number of factors differs from the ones in \cite{bai2002determining} and \cite{chen2021quantile} in three aspects. First, we do not require the factors to be strong in their sense. Second, all their estimators require a known upper bound on the number of factors while we do not need such prior information, taking advantage of the nuclear norm penalized estimation. Third, our estimator is based on a convex minimization problem that solves quickly and accurately.
\end{remark}
\begin{remark}\label{rem.initial.confirm}
Once $\hat{r}$ is obtained, we can use the associated estimated factors and loadings as the initial guess of $F_{0}$ and $\Lambda_{0}(\cdot)$ for our Algorithms \ref{alg.1} and \ref{alg.2}; see Remark \ref{rem.initial}. Specifically, let $\hat{F}^{pel}$ be $\sqrt{T}$ times the eigenvectors of $\sum_{m=1}^{M}\hat{L}^{pel'}(\tau_{m})\hat{L}^{pel}(\tau_{m})/(MNT)$ corresponding to the largest $\hat{r}$ eigenvalues. Then construct $\hat{\Lambda}^{pel}(\tau_{m})=\hat{L}^{pel}(\tau_{m})\hat{F}^{pel}/T$. One can show that these estimators are consistent in the average squared Frobenius norm. Starting from these ``nearby'' initial guesses leads to stable numerical performance when computing the UFA and IDW-UFA estimators.
\end{remark}
\subsection{Selecting Factors of Tolerated Strength}
In applications, researchers may be interested in an AFM or a QFM at a specific quantile level, and only wish to include sufficiently influential factors. In this section, we propose a method in a similar spirit of $\hat{r}$ to select factors in each model that have any tolerated strength.

We start from the AFM \eqref{eq.mean.model}. Let the nonzero singular values of $\bar{L}_{0}/\sqrt{NT}\coloneqq \bar{\Lambda}_{0}F'_{0}/\sqrt{NT}$ be $\bar{\sigma}_{1}\geq \ldots \geq \bar{\sigma}_{r}$. Under $F_{0}'F_{0}/T=I_{r}$, the strong factor condition in the literature of AFM (e.g. \cite{bai2003inferential}) refers to the case where $\bar{\sigma}_{j}$ has order $O(1)$, away from 0. A factor is weak in the sense of \cite{bai2023approximate} refers to the case where $\bar{\sigma}_{j}$ has order $O(N^{\alpha_{j}/2-1/2})$ for some $j$ and $0<\alpha_{j}<1$. 

We can show that, by the first part of Theorem \ref{thm.dist3} and by $N\asymp T$, $|\tilde{\bar{\sigma}}_{j}-\bar{\sigma}_{j}|$ is $O_{p}(N^{-1/2})=o_{p}(N^{\alpha_{j}/2-1/2})$ for any $\alpha_{j}>0$ uniformly in $j=1,\ldots,\min\{N,T\}$, where $\tilde{\bar{\sigma}}_{j}$ is the $j$-th largest singular value of $\tilde{\bar{\Lambda}}\tilde{F}'/\sqrt{NT}$. Hence, for any $\alpha\in (0,1]$, we estimate the number of factors that influence the conditional mean of $Y$ with strength at least $\alpha$ by 
\begin{equation}
   \tilde{\bar{r}}(\alpha)\coloneqq \sum_{j=1}^{r}1\left(\tilde{\bar{\sigma}}_{j}\geq \frac{CN^{\frac{\alpha-1}{2}}}{\log(N)}\right),
 \end{equation} 
 where $C$ is an arbitrary constant. 

 Similarly, let $\sigma_{j}(\tau_{m})$ be the $j$-th largest singular value of $L_{0}(\tau_{m})/\sqrt{NT}$. The strong factor assumption in \cite{chen2021quantile} is the case that $\sigma_{j}(\tau_{m})$ is $O(1)$ and away from 0 for all $j=1,\ldots,r$. Now similar to \cite{bai2023approximate}, consider weak factors such that $\sigma_{j}(\tau_{m})$ has order $O(N^{\alpha_{j}/2-1/2})$ for some $j$ and $0<\alpha_{j}<1$. We estimate the number of factors that influence the conditional quantile of $Y$ at $\tau_{m}$ with strength at least $\alpha$ by
\begin{equation}
   \tilde{r}_{\tau_{m}}(\alpha)\coloneqq \sum_{j=1}^{r}1\left(\tilde{\sigma}_{j}(\tau_{m})\geq \frac{CN^{\frac{\alpha-1}{2}}}{\log(N)}\right),
 \end{equation} 
 where $\tilde{\sigma}_{j}(\tau_{m})$ is the $j$-th largest singular value of $\tilde{\Lambda}(\tau_{m})\tilde{F}'/\sqrt{NT}$.
\

\begin{thm}\label{thm.rstrong}
Suppose for all $j=1,\ldots,r$, the $j$-th largest nonzero eigenvalues of $\bar{L}_{0}'\bar{L}_{0}/NT$ and $L_{0}'(\tau_{m})L_{0}(\tau_{m})/NT$ have order $N^{\bar{\alpha}_{j}-1}$ and $N^{\alpha_{\tau_{m},j}-1}$ with some constants $\bar{\alpha}_{j},\alpha_{\tau_{m},j}\in (0,1]$, respectively. Under the conditions in Theorem \ref{thm.dist2}, $\Pr(\tilde{\bar{r}}(\alpha)=\bar{r}(\alpha))\to 1$ and $\Pr(\tilde{r}_{\tau_{m}}(\alpha)=r_{\tau_{m}}(\alpha))=1$ for every $\alpha\in (0,1]$ and every fixed $\tau_{m}$, where $\bar{r}(\alpha)$ and $r_{\tau_{m}}(\alpha)$ are the numbers of factors whose $\bar{\alpha}_{j}\geq \alpha$ and $\alpha_{\tau_{m},j}\geq \alpha$, respectively.
\end{thm}

Our factor selectors provide alternative approaches to select empirically relevant factors compared to the existing methods. Unlike \cite{freyaldenhoven2022factor} who essentially imposes sparsity on $\bar{\Lambda}_{0}$ and requires that the smallest $\alpha_{j}$ to be greater than $1/2$ or \cite{bai2019rank} who use a ridge penalty to filter out factors that have relatively small influence on the mean of $Y$, our method can select out either mean or quantile factors with any $\alpha>0$.

\section{Monte Carlo Simulations}\label{sec.simu}
In this section, we demonstrate the finite sample performance of our estimators. We first compare the performance of UFA and IDW-UFA with QFA and PCA in estimating the factor space when there is a relatively weak quantile/mean factor. We then present the quality of Gaussian approximation of IDW-UFA in finite samples.
\subsection{Estimating the Factor Space}\label{sec.MC.space}
We let $r=1$ for simplicity. We consider sample size $(N,T)\in\{(50,50),(100,100),(150,150)\}$. We first draw $N\times 1$ and $T\times 1$ vectors $F^{*}_{0}$ and $\Lambda^{*}_{0,base}$ independently from $\text{Unif}[0,2]$. Draw an $N\times T$ matrix $U$ independently from $\text{Unif}[0,1]$. Let $\beta(U_{it})\coloneqq -0.99+2U_{it}$ and $\lambda_{0,i}^{*}(U_{it})\coloneqq \beta(U_{it})\lambda_{0,base,i}^{*}$. Construct $Y$ by $Y_{it}=\lambda_{0,i}^{*}(U_{it})f_{0,t}^{*}$.
By construction, $q_{Y|F_{0}^{*}}(\tau)=\Lambda^{*}(\tau)F_{0}^{*'},\tau\in (0,1)$. 

We can verify that $F_{0}^{*}$ is a strong universal factor because $\int_{0}^{1}\|\Lambda_{0}^{*}(\tau)\|^{2}_{F}d\tau/N$ is well bounded away from 0. However, $F_{0}^{*}$ is a ``relatively weak'' quantile factor near $\tau=0.5$ because $\|\Lambda_{0}^{*}(\tau)\|^{2}_{F}/N$ is close to 0 when $\tau$ is around 0.5, and a ``relatively weak''  mean factor because $\|\int_{0}^{1}\Lambda_{0}^{*}(\tau)d\tau)\|^{2}_{F}/N$ is close to 0, too;  ``relatively weak'' because these values, though close to zero, are still fixed, not diminishing as $N$ and $T$ grow to infinity. Consequently, when the sample size is sufficiently large, we should expect that QFA in \cite{chen2021quantile} and PCA still work. However, as shown below, these estimators' performance in the sample size we consider is not as well as UFA and IDW-UFA.
 
To implement our estimators, we first estimate $r$ by $\hat{r}$ proposed in Section \ref{sec.r}. We set $C$ in \eqref{eq.obj.nuc} equal to $0.2$ and $C_{r}$ in \eqref{eq.estimator.r} equal to $1/(12(\min(N,T))^{1/3})$. Table \ref{tab.r} presents the average, maximum and minimum $\hat{r}$ in 1000 simulation repetitions. The results show that our estimator of $r$ performs very good; for sample size $(50,50)$, only in 16 repetitions $r$ is overestimated by 1. Under larger sample size, $r$ is correctly estimated in all repetitions.

\begin{table}[htbp]
  \centering
  \caption{$\hat{r}$}
    \begin{tabular}{ccccc}
    \toprule
    \toprule
         \multicolumn{2}{c}{$(N,T)$} & (50,50) & (100,100)& (150,150)\\
    \midrule
    \multirow{3}{*}{$\hat{r}$} & Average &1.016 & 1 & 1\\
    & Max &2 & 1&1\\
    & Min & 1& 1&1 \\
    \bottomrule
    \end{tabular}
\label{tab.r}
\end{table}
 
Next, we estimate the factor by UFA and IDW-UFA. To satisfy Assumption \ref{assum.kernel}, we let $\gamma=14$, set $h=1/\min(N,T)^{1/13}$, and choose the following fourteenth-order Gaussian-based kernel (see \cite{wand1990gaussian} and \cite{marron1992exact}): $k(z)\coloneqq \left(\sum_{i=0}^{6}c_{2i}z^{2i}\right)\phi(z)$ where $\phi$ is the standard normal density and 
\begin{equation*}
c_{2i}=\frac{(-1)^{i}\cdot 2^{i-13}\cdot 14!}{7!(2i+1)!(6-i)!}.
\end{equation*}
As discussed below Assumption \ref{assum.kernel}, a fourth-order kernel is sufficient to deliver $\sqrt{N}$-consistency for factor space estimation. In fact, we find in simulations that even a second-order kernel has similar performance. To save space, we present results under the current kernel just to be coherent with the assumption. For the inverse density estimator, we set $h_{d}=0.04$. We fix $M=9$, $\tau_{m}\in \{0.1,0.2,\ldots,0.9\}$, and use $\hat{r}$ as the number of factors for computation. For the initial guess $F^{0}$ and $\Lambda^{0}(\cdot)$ in Algorithm \ref{alg.1} to compute UFA, we use the ones proposed in Remark \ref{rem.initial.confirm}, whereas we use $\hat{F}$ and $\hat{\Lambda}(\cdot)$ obtained from UFA as the initial guess for Algorithm \ref{alg.2}.
 
In addition to the two estimators proposed in this paper, we also estimate the factor by PCA and by smoothed QFA \citep{chen2021quantile} at each $\tau_{m}$ to see how weak mean/quantile factor affects the performance. For these two estimators, we directly use the true $r$ for the number of factors. For the smoothed QFA, we slightly modify the approach in \cite{chen2021quantile} by using our smoothed objective function.

Note that since the smoothed QFA is also obtained from a nonconvex optimization problem without an analytical solution, no matter how smoothing is conducted, the initial guess should play a crucial role. We use two sets of initial guesses for it. The first one is the same as the one for our UFA; denoted by $\text{ini}_{UFA}$. This initial guess, by utilizing the information across all $M$ quantile levels, is already consistent in the average squared Frobenius norm. In practice, however, if a researcher is to use QFA, she is taking a quantile-level-specific approach, so it is more likely that she only solves the nuclear norm penalized estimator \eqref{eq.obj.nuc} at one $\tau$, obtains $\hat{L}^{pel}(\tau)$, sets $\hat{F}^{pel}(\tau)$ equal to the eigenvectors corresponding to the $r$ largest eigenvalues of $\hat{L}^{pel'}(\tau)\hat{L}^{pel}(\tau)/NT$ multiplied by $\sqrt{T}$ and $\hat{\Lambda}^{pel}(\tau)=\hat{L}^{pel}(\tau)\hat{F}^{pel}(\tau)/T$, and uses $(\hat{F}^{pel}(\tau),\hat{\Lambda}^{pel}(\tau))$ as the initial guess, denoted by $\text{ini}_{\tau}$. We compute QFA under either initial guess.

 To evaluate the performance, we regress the true factor $F_{0}^{*}$ on the estimated factors and obtain the adjusted $R^{2}$. A higher adjusted $R^{2}$ indicates that the estimated factor space is closer to the true one.
 Table \ref{tab.Rsq} shows the results averaged over 1000 simulation repetitions. Since our estimators and PCA are not quantile-level-specific, they have identical results across $\tau_{m}$. 

 Three observations can be drawn from Table \ref{tab.Rsq}. First, the performance of UFA and IDW-UFA are very similar, confirming Theorems \ref{thm.roc} and \ref{thm.dist2}. The adjusted $R^{2}$, already high when $(N,T)=(50,50)$, increases as the sample size increases. Second, QFA, on the other hand, has inferior performance, especially when the factor becomes weaker around $\tau=0.5$. The initial guess is indeed important, but even under the already consistent initial guess $\text{ini}_{UFA}$, QFA still only has an adjusted $R^{2}$ of $0.51$ at $\tau=0.5$ under $(N,T)=(150,150)$. Under the initial guess $\text{ini}_{\tau}$ that is more likely to be adopted in practice, QFA no longer works at $\tau=0.5$ as the adjusted $R^{2}$ is close to 0. Third, PCA does not work due to the weak strength of $F_{0}^{*}$ as a mean factor.

\begin{table}[htbp]
  \centering
  \caption{Adjusted $R^{2}$ in Regressions of the True Factor on the Estimated Factor}
  \addtolength{\tabcolsep}{-0.35em}
    \begin{tabular}{cccccccccccccccc}
    \toprule
    \toprule
          & \multicolumn{3}{c}{UFA} & \multicolumn{3}{c}{IDW-UFA} & \multicolumn{3}{c}{QFA: $\text{ini}_{UFA}$} & \multicolumn{3}{c}{QFA: $\text{ini}_{\tau}$} & \multicolumn{3}{c}{PCA}\\

    \cmidrule(lr){2-4} \cmidrule(lr){5-7} \cmidrule(lr){8-10} \cmidrule(lr){11-13} \cmidrule(lr){14-16}
    $N=T=$ & 50 & 100 & 150 & 50 & 100 & 150 & 50 & 100 & 150 & 50 & 100 & 150 & 50 & 100 & 150 \\
    \midrule
    $\tau=0.1$   & \multirow{9}{*}{0.95} & \multirow{9}{*}{0.97} & \multirow{9}{*}{0.98} & \multirow{9}{*}{0.93}&\multirow{9}{*}{0.97} & \multirow{9}{*}{0.98}& 0.93 & 0.96& 0.98 & 0.93& 0.96 & 0.97& \multirow{9}{*}{0.01}& \multirow{9}{*}{0.01}& \multirow{9}{*}{0.01} \\
   $\tau=0.2$ &  &  &  & & & & 0.84 & 0.91& 0.94 & 0.84& 0.91 & 0.94& && \\
    $\tau=0.3$ &  &  &  & & & & 0.66 & 0.79& 0.85 & 0.64& 0.78 & 0.84& && \\
   $\tau=0.4$  &  &  &  & & & & 0.36 & 0.53& 0.64 & 0.11& 0.31 & 0.46& && \\
    $\tau=0.5$  &  &  &  & & & & 0.19 & 0.29& 0.51 & 0.01& 0.01 & 0.002& && \\
     $\tau=0.6$ &  &  &  & & & & 0.39 & 0.56& 0.67 & 0.16& 0.41 & 0.54& && \\
   $\tau=0.7$  &  &  &  & & & & 0.68 & 0.81& 0.86 & 0.66& 0.80 & 0.85& && \\
    $\tau=0.8$ &  &  &  & & & & 0.85 & 0.92& 0.94 & 0.85& 0.92 & 0.94& && \\
     $\tau=0.9$ &  &  &  & & & & 0.93 & 0.97& 0.98 & 0.93& 0.97 & 0.98& && \\
    \bottomrule
    \end{tabular}
\label{tab.Rsq}
\end{table}

\subsection{Gaussian Approximation}\label{sec.MC.dist}
We now demonstrate the quality of the normal approximation in Theorem \ref{thm.dist2} for our IDW-UFA estimator. We follow the same data generating process as in Section \ref{sec.MC.space}. However, now we only draw the factors and loadings once and keep them fixed across the 1000 simulation repetitions, the same as \cite{chen2021quantile}. We estimate the fixed factors and loadings using IDW-UFA with sample size $(N,T)\in\{(50,50),(75,75),(100,100),(150,150)\}$. The implementation details are identical to Section \ref{sec.MC.space} except that now we use the true number of factors $r=1$ to avoid the rare cases of $\hat{r}\neq r$ where the dimensionality of $\tilde{f}_{t}$ is different from that of $f_{0,t}$.

Before presenting the results, first note that in Theorem \ref{thm.dist2}, $\tilde{f}_{t}$ is asymptotically normal up to rotation. Same as \cite{bai2003inferential}, this rotation matrix $H_{NT,2}$ is unknown as it depends on the true $F_{0}$. However, under $r=1$, Theorem \ref{thm.dist2} shows that $H_{NT,2}=\tilde{H}_{NT,1}+o_{p}(1/\sqrt{N})$, and thus by normalizing the sign of $\tilde{f}_{t}$ and $f_{0,t}$ which leads to $\tilde{H}_{NT,1}=1$, we can drop the rotation matrix and directly look at $\tilde{f}_{t}-f_{0,t}$. To verify, Table \ref{tab.rotation} presents $|H_{NT,2}-1|$  averaged over 1000 repetitions under different sample size. It can be seen that the difference is indeed very small, and gets closer to 0 as the sample size increases.

\begin{table}[htbp]
  \centering
  \caption{$H_{NT,2}$}
    \begin{tabular}{ccccc}
    \toprule
    \toprule
      $N=T=$ & 50 & 75& 100 & 150\\
    \midrule
     $|H_{NT,2}-1|$& 0.008 &0.005& 0.004 & 0.003\\
    \bottomrule
    \end{tabular}
\label{tab.rotation}
\end{table}
 
Now, we consider the standardized estimated factor and common component. Define $\tilde{f}^{std}_{t}\coloneqq (\tilde{f}_{t}-f_{0,t})/se_{\tilde{f}_{t}}$,
where $se_{\tilde{f}}$ is the standard error, constructed by the plug-in estimator described in Remark \ref{rem.covariance.estimator}. Note that it is not the standard deviation of the estimate across simulation repetitions. Under Theorem \ref{thm.dist2}, the distribution of $\tilde{f}^{std}_{t}$ approaches to standard normal as $N,T\to\infty$ for any fixed $t$. Define $\tilde{L}_{it}^{std}(\tau)$ similarly. Let $i=\lfloor N/2\rfloor$, $t=\lfloor T/2\rfloor$, where $\lfloor a \rfloor$ is the largest integer that is no greater than $a$. Table \ref{tab.mean.std} presents the sample mean and standard deviation of $\tilde{f}^{std}_{t}$ and $\tilde{L}^{std}_{it}(\tau)$ for $\tau=0.2,0.5,0.8$ over 1000 simulation repetitions. Note that the factor is relatively weak at $\tau=0.5$ whereas $0.2$ and $0.8$ are quantile levels near the boundaries. From the results, the sample mean is in general close to 0; it is relatively large at $(N,T)=(50,50)$, but soon gets smaller when the sample size reaches $(75,75)$. The standard deviation is close to 1, especially for $N>50$ and $T>50$. 
\begin{table}[htbp]
  \centering
  \caption{Mean and Standard Deviation of the Standardized Estimators}
    \begin{tabular}{ccccccccc}
    \toprule
    \toprule
          & \multicolumn{4}{c}{Mean} & \multicolumn{4}{c}{Standard Dev.} \\
    \cmidrule(lr){2-5} \cmidrule(lr){6-9}
    $N=T=$ & 50 & 75 & 100 & 150 & 50 & 75 & 100 & 150 \\
    \midrule
    $\tilde{f}_{t}^{std}$  &$-0.64$ &$-0.04$ &0.06  & $-0.08$ & 1.42&1.16 &1.06 &1.07 \\
    $\tilde{L}_{it}^{std}(0.5)$           &0.04 & 0.004   &$-0.05$  &$-0.03$ &  1.23 &1.00&1.11 &1.00\\
    $\tilde{L}_{it}^{std}(0.2)$           &0.37 &0.12  &0.04  &0.05      & 1.13&0.96 &0.98 & 1.01\\
    $\tilde{L}_{it}^{std}(0.8)$           &$-0.39$& $-0.14$ & $-0.06$  &$-0.04$ &1.24 &0.95 &1.02 &1.01  \\
    \bottomrule
    \end{tabular}
\label{tab.mean.std}
\end{table}

Finally, we plot the histograms (scaled to be density functions) of these standardized estimates, superimposed by the standard normal density curve. Figure \ref{fig.dist} shows that the normal approximation is reasonably good in all cases, and performs better when the sample size increases.
\begin{figure}[h]
\begin{center}
\includegraphics[scale=0.4]{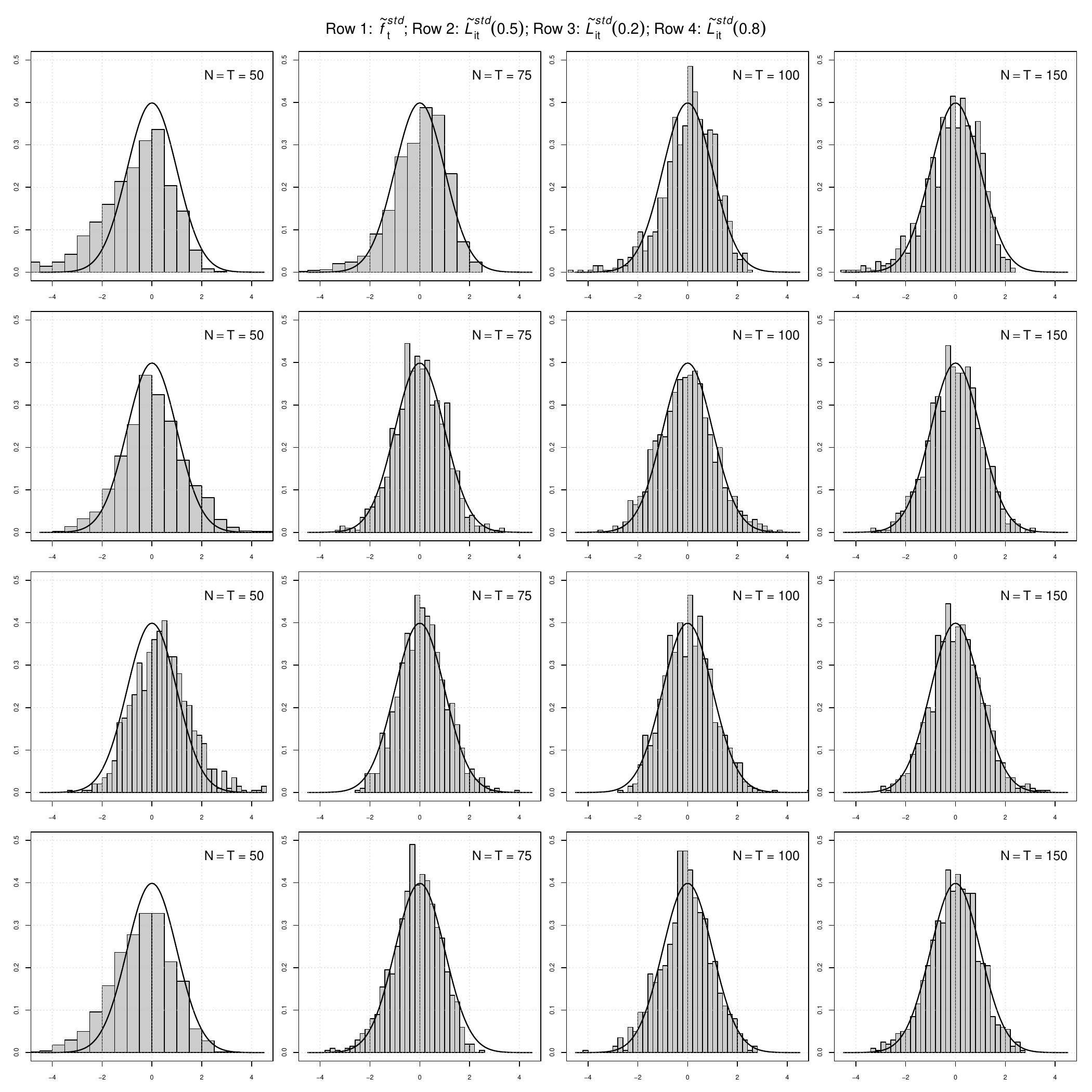}
\caption{Histograms of the standardized factor and common component estimates. Row 1 to row 4 plots the histograms of $\tilde{f}_{t}^{std}$, $\tilde{L}_{it}^{std}(0.5)$, $\tilde{L}_{it}^{std}(0.2)$, $\tilde{L}_{it}^{std}(0.8)$, respectively.}
\label{fig.dist}
\end{center}
\end{figure}

 \section{Conclusion}\label{sec.con}
This paper proposes a new factor analysis framework. By collecting all the factors that may have impact on the outcome, namely, the universal factors, our framework can induce an AFM and a QFM, but does not require the strong factor assumptions maintained in those models. 

We build two estimators for the universal factors and loadings. Both estimators achieve the optimal rate when estimating the space spanned by the factors, regardless of the strength of the factors with respect to the mean or quantile of the outcome at a given quantile level. Built on a novel sample splitting strategy, the inverse density weighted estimator further achieves $\sqrt{N}$-asymptotic normality for individual factors and loadings. Monte Carlo shows that our estimators have superior performance to QFA and PCA in the presence of weak factors.

In addition, we propose estimators concerning the number of factors. A weak-factor-robust consistent estimator for the total number of universal factors is constructed. We also develop consistent factor selectors that can select out factors having any tolerated level of influence in an AFM or a QFM.

\begin{appendices}
\section{Proof of Lemma \ref{lem.discrete}}
\begin{proof}[Proof of Lemma \ref{lem.discrete}]
Let $\mathcal{U}=[\underline{\tau},\bar{\tau}]$. Let $\tau_{0}=\underline{\tau}$.
 By the boundedness of $\lambda_{0,i}^{*}(\cdot)$ and Assumption \ref{assum.lip}, for some $C_{1},\ldots,C_{4}>0$,
\begin{align}
&\left\|\frac{1}{MN}\sum_{m=1}^{M}\Lambda^{*'}_{0}(\tau_{m})\Lambda^{*}_{0}(\tau_{m})-\frac{1}{N}\int_{\underline{\tau}}^{\bar{\tau}}\Lambda^{*'}_{0}(\tau)\Lambda^{*}_{0}(\tau)d\tau\right\|_{F}\notag\\
\leq &\max_{i=1,\ldots,N}\left\|\frac{1}{M}\sum_{m=1}^{M}\lambda_{0,i}^{*}(\tau_{m})\lambda_{0,i}^{*'}(\tau_{m})-\int_{\underline{\tau}}^{\bar{\tau}}\lambda_{0,i}^{*}(\tau)\lambda_{0,i}^{*'}(\tau)d\tau\right\|_{F}\notag\\
\leq &\max_{i=1,\ldots,N}\left\|\frac{1}{M}\sum_{m=1}^{M}\left(\lambda_{0,i}^{*}(\tau_{m})\lambda_{0,i}^{*'}(\tau_{m})-M\int_{\tau_{m-1}}^{\tau_{m}}\lambda_{0,i}^{*}(\tau)\lambda_{0,i}^{*'}(\tau)d\tau\right)\right\|_{F}+\max_{i=1,\ldots,N}\left\|\int_{\tau_{M}}^{\bar{\tau}}\lambda_{0,i}^{*}(\tau)\lambda_{0,i}^{*'}(\tau)d\tau\right\|_{F}\notag\\
\leq &\max_{i=1,\ldots,N}\left\|\sum_{m=1}^{M}\int_{\tau_{m-1}}^{\tau_{m}}\left(\lambda_{0,i}^{*}(\tau_{m})\lambda_{0,i}^{*'}(\tau_{m})-\lambda_{0,i}^{*}(\tau)\lambda_{0,i}^{*'}(\tau)\right)d\tau\right\|_{F}\notag\\
&+\left(\frac{1}{M}-\frac{1}{M+1}\right)\max_{i=1,\ldots,N}\left\|\sum_{m=2}^{M}\lambda_{0,i}^{*}(\tau_{m})\lambda_{0,i}^{*'}(\tau_{m})\right\|_{F}\notag\\
&+\left(\underline{\tau}+\frac{1}{M}-\frac{1}{M+1}\right)\cdot\max_{i=1,\ldots,N}\left\|\lambda_{0,i}^{*}(\tau_{1})\lambda_{0,i}^{*'}(\tau_{1})\right\|_{F}+\max_{i=1,\ldots,N}\left\|\int_{\tau_{M}}^{\bar{\tau}}\lambda_{0,i}^{*}(\tau)\lambda_{0,i}^{*'}(\tau)d\tau\right\|_{F}\notag\\
\leq &C_{1}\max_{i=1,\ldots,N}\sum_{m=1}^{M}\left[\left(\tau_{m}-\tau_{m-1}\right)\int_{\tau_{m-1}}^{\tau_{m}}d\tau\right]+\frac{C_{2}}{M}+C_{3}\left(\bar{\tau}-\tau_{M}\right)\leq \frac{C_{4}}{M},\label{eq.proof.discrete}
\end{align}
where the third inequality holds because for any $m$, $\int_{\tau_{m-1}}^{\tau_{m}}\lambda_{0,i}^{*}(\tau_{m})\lambda_{0,i}^{*'}(\tau_{m})d\tau=(\tau_{m}-\tau_{m-1})\lambda_{0,i}^{*}(\tau_{m})\lambda_{0,i}^{*'}(\tau_{m})$; note $\tau_{m}-\tau_{m-1}=1/(M+1)$ for $m\geq  2$ and equal to $1/(M+1)-\underline{\tau}$ for $m=1$. The penultimate inequality is due to Lipschitz continuity and boundedness of $\lambda_{0,i}^{*}(\cdot)$ on $\mathcal{U}$ and by $\underline{\tau}<\tau_{1}=1/(M+1)$. The last inequality is by $\bar{\tau}-\tau_{M}<1/(M+1)$. Therefore, arranging the eigenvalues of $\sum_{m=1}^{M}\Lambda_{0}^{*'}(\tau_{m})\Lambda_{0}^{*}(\tau_{m})/MN$ and those of $\int_{0}^{1}\Lambda_{0}^{*'}(\tau)\Lambda_{0}^{*}(\tau)d\tau/N$ in descending order, the difference between each pair is bounded by $C_{4}/M$ by the Weyl's inequality, giving the desired result for sufficiently large $M$ under Assumption \ref{assum.normalization}-(i).

 Next, we show the distinctiveness of the eigenvalues. 
Denote the $j$-th largest eigenvalues in $F_{0}^{*}\int_{\mathcal{U}}\Lambda_{0}^{*'}(\tau)\Lambda_{0}^{*}(\tau)d\tau F_{0}^{*'}/(NT)$ and $F_{0}^{*}\sum_{m=1}^{M}\Lambda_{0}^{*'}(\tau_{m})\Lambda_{0}^{*}(\tau_{m})F_{0}^{*'}/MNT$ by $\sigma^{*2}_{j}$ and $\sigma^{2}_{M,j}$, respectively. Let $c^{*}\coloneqq \min_{j=1,\ldots,r-1}(\sigma^{*2}_{j}-\sigma^{*2}_{j+1})$. By Assumption \ref{assum.normalization}, $c^{*}>0$ for all $N$. Therefore, there exists $C_{5}>0$ such that
\begin{align*}
  \min_{j=1,\ldots,r-1}&\left(\sigma_{M,j}^{2}-\sigma_{M,j+1}^{2}\right)\geq  \min_{j=1,\ldots,r-1}(\sigma^{*2}_{j}-\sigma^{*2}_{j+1})  -\max_{j=1,\ldots,r-1}\left|\sigma_{M,j}^{2}-\sigma_{j}^{*2}\right|-\max_{j=1,\ldots,r-1}\left|\sigma_{M,j+1}^{2}-\sigma_{j+1}^{*2}\right|\\
  \geq & c^{*}-2\left\|\frac{1}{MNT}F_{0}^{*}\sum_{m=1}^{M}\Lambda^{*'}_{0}(\tau_{m})\Lambda^{*}_{0}(\tau_{m})F_{0}^{*'}-\frac{1}{NT}F_{0}^{*}\int_{\underline{\tau}}^{\bar{\tau}}\Lambda^{*'}_{0}(\tau)\Lambda^{*}_{0}(\tau)d\tau F_{0}^{*'}\right\|_{F}\geq c^{*}-\frac{2C_{5}}{M},
\end{align*}
where the first and second inequality are by the triangle inequality and the Weyl's inequality, respectively. The last inequality is by \eqref{eq.proof.discrete}, by the Cauchy-Schwarz inequality and by boundedness of $f_{0,t}^{*}$. The desired results then follow under sufficiently large $M$.
\end{proof}
\section{Proof of Theorem \ref{thm.roc}}\label{appx.roc}
We start by introducing some notation. Recall that $\mathcal{B}$ is a compact subset of $\mathbb{R}$ such that $\lambda_{0,i}(\tau_{m}),f_{0,t}\in\mathcal{B}^{r}$ for all $m,i,t$. Let $\Theta_{\lambda}\coloneqq\{(\lambda_{1}'(\tau_{1}),\cdots,\lambda_{N}'(\tau_{1}),\ldots,\lambda_{1}'(\tau_{M}),\cdots,\lambda_{N}'(\tau_{M}))'\in \mathcal{B}^{MNr}:(\sum_{i=1}^{N}\lambda_{i}(\tau_{m})\lambda_{i}'(\tau_{m})/N)_{m}\in\Xi\}$ and $\Theta_{f}\coloneqq\{(f_{1}',\ldots,f_{T}')'\in \mathcal{B}^{Tr}:\sum_{t=1}^{T}f_{t}f_{t}'/T\in\mathcal{F}\}$.
Let $\Theta\coloneqq\Theta_{\lambda}\times\Theta_{f}$. Denote an arbitrary element in $\Theta$ by $\theta$. Let the vector of the diagonalized true loadings and factors be $\theta_{0}$ and the estimated loadings and factors by $\hat{\theta}$. Under Lemma \ref{lem.discrete}, by the definition of $F_{0}$ and $\Lambda_{0}(\cdot)$ and the definition of the estimator \eqref{eq.estimator}, $\theta_{0},\hat{\theta}\in\Theta$.

Define $\hat{R}_{h,\tau}(c;Y_{it})\coloneqq \int_{s}\rho_{\tau}(s)k\left((s-Y_{it}+c)/h\right)ds/h$. For $j=1,2,3$, $\hat{R}^{(j)}_{h,\tau}(c;Y_{it})=(\partial/\partial c)^{j}\hat{R}_{h,\tau}(c;Y_{it})$. When $c=\lambda_{0,i}'(\tau_{m})f_{0,t}$, denote $\hat{R}^{(j)}_{h,\tau_{m}}(c;Y_{it})$ by $\hat{R}^{(j)}_{h,\tau_{m},it}$. By \cite{fernandes2021smoothing}, $\hat{R}^{(1)}_{h,\tau_{m}}\left(c;Y_{it}\right) =K\left((c-Y_{it})/h\right)-\tau_{m}$, $\hat{R}^{(2)}_{h,\tau_{m},it}\left(c;Y_{it}\right)=k\left((c-Y_{it})/h\right)/h$, and $\hat{R}^{(3)}_{h,\tau_{m},it}\left(c;Y_{it}\right)=k^{(1)}\left((c-Y_{it})/h\right)/h^{2}$, where $K(z)\coloneqq\int_{-\infty}^{z}k(v)dv$ and $k^{(1)}(z)$ is the derivative of $k$ at $z$. 
Finally, for any $\theta\in\Theta$, let $\hat{R}_{h}(\Lambda(\cdot),F)\coloneqq \sum_{m=1}^{M}\sum_{i=1}^{N}\sum_{t=1}^{T}\hat{R}_{h,\tau_{m}}(\lambda_{i}(\tau_{m})'f_{t};Y_{it})/MNT$. Let $\Delta_{\hat{R}_{h}}(\theta)\coloneqq \hat{R}_{h}(\Lambda(\cdot),F)-\hat{R}_{h}(\Lambda_{0}(\cdot),F_{0})$, $\bar{\Delta}_{\hat{R}_{h}}(\theta)\coloneqq \mathbb{E}(\Delta_{\hat{R}_{h}}(\theta))$ and $\hat{S}_{h}(\theta)\coloneqq \Delta_{\hat{R}_{h}}(\theta)-\bar{\Delta}_{\hat{R}_{h}}(\theta)$.

We need the following lemmas to show Theorem \ref{thm.roc}. They are similar to \cite{chen2021quantile}, and we give their proofs in the Online Appendix. 
\begin{lem}\label{lem.nonpara}
Under Assumptions \ref{assum.foundation} to \ref{assum.kernel},
\begin{itemize}
  \item[(i)] There exists a constant $C>0$ such that $h^{j-1}\cdot\sup_{c\in\mathbb{R},\tau\in\mathcal{U},i,t}|\hat{R}_{h,\tau}^{(j)}(c;Y_{it})|\leq C$ for $j=1,2,3$.
  \item[(ii)]  $\max_{m,i,t}|\mathbb{E}(\hat{R}^{(1)}_{h,\tau_{m},it})|=O(h^{\gamma})$, $\sup_{m,i,t,\lambda\in\mathcal{B}^{r},f\in\mathcal{B}^{r}}|\mathbb{E}(\hat{R}^{(2)}_{h,\tau_{m}}(\lambda'f;Y_{it}))-\textsf{f}_{\tau_{m},it}(\lambda'f-\lambda_{0,i}'(\tau_{m})f_{0,t})|=O(h^{\gamma})$, and $\sup_{m,i,t,\lambda\in\mathcal{B}^{r},f\in\mathcal{B}^{r}}|\mathbb{E}(\hat{R}^{(3)}_{h,\tau_{m}}(\lambda'f;Y_{it}))+\textsf{f}_{\tau_{m},it}^{(1)}(\lambda'f-\lambda_{0,i}'(\tau_{m})f_{0,t})|=O(h^{\gamma})$.
  \item[(iii)] $\mathbb{E}(\hat{R}_{h,\tau_{m},it}^{(1)})^{2}=\tau_{m}(1-\tau_{m})+O(h)$, $\mathbb{E}(\hat{R}_{h,\tau_{m},it}^{(1)}\cdot \hat{R}_{h,\tau_{m'},it}^{(1)})=\min(\tau_{m},\tau_{m'})-\tau_{m}\tau_{m'}+o(1)$, and $h\mathbb{E}(\hat{R}_{h,\tau_{m},it}^{(2)})^{2}=O(1)$.
\end{itemize}
\end{lem}

\begin{lem}\label{lem.avgcons}
Under Assumptions \ref{assum.foundation} to \ref{assum.kernel},
\begin{equation*}
  \frac{1}{MNT}\sum_{m=1}^{M}\left\|\hat{L}(\tau_{m})-L_{0}(\tau_{m})\right\|_{F}^{2}=o_{p}(1).
\end{equation*}
\end{lem}

Let $\mathcal{H}$ be the set of $r\times r$ diagonal matrices whose diagonal entries only consist of $1$ and $-1$. Notice that $|\mathcal{H}|=2^{r}$, and for any element $H\in\mathcal{H}$, $H^{2}=I_{r}$. For any $\theta_{1},\theta_{2}\in\Theta$ and $H\in\mathcal{H}$, define 
\begin{align*}
d_{1}(\theta_{1},\theta_{2})\coloneqq
& \sqrt{\frac{1}{M}\sum_{m=1}^{M}\frac{1}{NT}\left\|L_{1}\left(\tau_{m}\right)-L_{2}\left(\tau_{m}\right)\right\|_{F}^{2}} ,\\
d_{2}(\theta_{1},\theta_{2};H)\coloneqq
&\sqrt{\frac{1}{MN}\sum_{m=1}^{M}\left\|\Lambda_{1}(\tau_{m})-\Lambda_{2}(\tau_{m})H\right\|_{F}^{2}+\frac{1}{T}\left\|F_{1}-F_{2}H\right\|_{F}^{2}},\\
d_{3}(\theta_{1},\theta_{2};H)\coloneqq&\sqrt{\frac{1}{MN}\sum_{m=1}^{M}\left\|\Lambda_{1}(\tau_{m})-\Lambda_{2}(\tau_{m})H\right\|_{F}^{2}}+\sqrt{\frac{1}{T}\left\|F_{1}-F_{2}H\right\|_{F}^{2}}.
\end{align*}

Let $\Theta(\delta,d_{1})\coloneqq \{\theta\in \Theta:d_{1}\left(\theta,\theta_{0}\right)\leq \delta\}$, and $\Theta(\delta,d_{k})\coloneqq \{\theta\in \Theta:\exists H\in\mathcal{H}\ s.t.\ d_{k}\left(\theta,\theta_{0};H\right)\leq \delta\},k=2,3$.
\begin{lem}\label{lem.metrics}
Under Assumptions \ref{assum.normalization} and \ref{assum.lip}, for any $\delta>0$, there exists a constant $C>0$ such that $d_{1}(\theta,\theta_{0})\leq Cd_{3}(\theta,\theta_{0};H(\theta))$ for all $\theta\in\Theta$ where $H(\theta):= \text{diag}(\text{sgn}(F'_{j}F_{0,j}))$, implying that $\Theta(\delta,d_{1})\subseteq \Theta(C\delta,d_{3})$. 
\end{lem}

\begin{lem}\label{lem.psi2}
Under Assumptions \ref{assum.foundation} to \ref{assum.density}, for sufficiently small $\delta>0$,
\begin{equation*}
  \mathbb{E}\left[\sup_{\theta\in\Theta(\delta,d_{1})}\left|\hat{S}_{h}(\theta)\right|\right]\lesssim \delta\zeta_{NT}.
\end{equation*}
\end{lem}

\begin{lem}\label{lem.uniformcons}
Under Assumptions \ref{assum.foundation} to \ref{assum.kernel}, 
\begin{align*}
   \max_{i,m}\left\|\hat{\lambda}_{i}(\tau_{m})-H_{NT,1}^{-1}\lambda_{0,i}(\tau_{m})\right\|_{F}=o_{p}\left(1\right),\ \ \max_{t}\left\|\hat{f}_{t}-H_{NT,1}'f_{0,t}\right\|_{F}=&o_{p}\left(1\right).
\end{align*}
\end{lem}

\begin{proof}[Proof of Theorem \ref{thm.roc}]
We first prove Theorem \ref{thm.roc}-(i).
By $h^{\gamma/2}=o(\zeta_{NT})$ under Assumption \ref{assum.kernel}, partition the parameter space $\Theta\subset\mathbb{R}^{(MN+T)r}$ into shells $S_{j}\coloneqq\{\theta\in \Theta: 2^{j-1}<d_{1}(\theta,\theta_{0})/\zeta_{NT}\leq 2^{j}\}$. Then by Lemmas \ref{lem.avgcons} and \ref{lem.psi2}, 
we can follow exactly the same steps as the proof of Lemma S.3 in \cite{chen2021quantile} and show that
\begin{equation*}
  \sqrt{\sum_{m=1}^{M}\frac{1}{MNT}\left\|\hat{L}\left(\tau_{m}\right)-L_{0}\left(\tau_{m}\right)\right\|_{F}^{2}}=O_{p}\left(\zeta_{NT}\right)+O\left(h^{\gamma/2}\right).
\end{equation*}
Then Lemma \ref{lem.metrics} and diagonality of $H_{NT,1}$ imply that
\begin{align*}
\sqrt{\frac{1}{MN}\sum_{m=1}^{M}\left\|\hat{\Lambda}(\tau_{m})-\Lambda_{0}(\tau_{m})H_{NT,1}^{-1'}\right\|_{F}^{2}}+\sqrt{\frac{1}{T}\left\|\hat{F}-F_{0}H_{NT,1}\right\|_{F}^{2}}=O_{p}\left(\zeta_{NT}+h^{\gamma/2}\right).
\end{align*}
The desired result is obtained.

Next, we consider the uniform rates. The rates of $\hat{\lambda}_{i}(\tau_{m})$ and $\hat{f}_{t}$ are derived similarly to Lemma S.5 in \cite{chen2021quantile}; the difference is that here we establish uniform rates. Once these rates are obtained, the rate of $\hat{\lambda}'_{i}(\tau_{m})\hat{f}_{t}$ is immediate by the triangle inequality and boundedness of $\hat{\lambda}_{i}(\tau_{m}),\hat{f}_{t},\lambda_{0,i}(\tau_{m})$ and $f_{0,t}$. To save space, we only derive the uniform rate of $\hat{\lambda}_{i}(\tau_{m})$ here for an arbitrary fixed $\tau_{m}$. 

For simplicity, denote $H_{NT,1}'f_{0,t}$ and $H_{NT,1}^{-1}\lambda_{0,i}(\tau_{m})$ by $f_{0,t}^{H}$ and $\lambda_{0,i}^{H}(\tau_{m})$, respectively. Let $\bar{\hat{R}}_{h,\tau}^{(j)}(c;Y_{it})\coloneqq \mathbb{E}\left(\hat{R}_{h,\tau}^{(j)}(c;Y_{it})\right)$ and $\bar{\hat{R}}_{h,\tau_{m},it}^{(j)}\coloneqq \mathbb{E}\left(\hat{R}_{h,\tau}^{(j)}(\lambda_{0,i}'(\tau_{m})f_{0,t};Y_{it})\right)$, $j=1,2,3$, where the expectation is taken with respect to $Y_{it}$. For any $i$, expanding $\sum_{t}\bar{\hat{R}}_{h,\tau_{m}}^{(1)}(\hat{\lambda}_{i}'(\tau_{m})\hat{f}_{t};Y_{it})\hat{f}_{t}/T$,\small
\begin{align}
&\frac{1}{T}\sum_{t}\bar{\hat{R}}_{h,\tau_{m}}^{(1)}\left(\hat{\lambda}_{i}'(\tau_{m})\hat{f}_{t};Y_{it}\right)\hat{f}_{t}\notag\\
=&\frac{1}{T}\sum_{t}\bar{\hat{R}}_{h,\tau_{m}}^{(1)}\left(\lambda_{0,i}^{H'}(\tau_{m})\hat{f}_{t};Y_{it}\right)\hat{f}_{t}+\frac{1}{T}\sum_{t}\bar{\hat{R}}_{h,\tau_{m}}^{(2)}\left(\lambda_{0,i}^{H'}(\tau_{m})\hat{f}_{t};Y_{it}\right)\hat{f}_{t}\hat{f}'_{t}\left(\hat{\lambda}_{i}(\tau_{m})-\lambda^{H}_{0,i}(\tau_{m})\right)\notag\\
&+\frac{0.5}{T}\sum_{t}\bar{\hat{R}}_{h,\tau_{m}}^{(3)}\left(\lambda_{i}^{*'}(\tau_{m})\hat{f}_{t};Y_{it}\right)\hat{f}_{t}\left[\hat{f}_{t}'\left(\hat{\lambda}_{i}(\tau_{m})-\lambda^{H}_{0,i}(\tau_{m})\right)\right]^{2}\notag\\
=&\underbrace{\frac{1}{T}\sum_{t}\bar{\hat{R}}_{h,\tau_{m},it}^{(1)}f^{H}_{0,t}}_{A_{1mi}}+\underbrace{\frac{1}{T}\sum_{t}\bar{\hat{R}}_{h,\tau_{m}}^{(1)}\left(\lambda_{i}^{H'}(\tau_{m})f^{*}_{t};Y_{it}\right)\left(\hat{f}_{t}-f^{H}_{0,t}\right)}_{A_{2mi}}\notag\\
&+\underbrace{\frac{1}{T}\sum_{t}\bar{\hat{R}}_{h,\tau_{m}}^{(2)}\left(\lambda_{0,i}^{H'}(\tau_{m})f^{*}_{t};Y_{it}\right)f^{*}_{t}\lambda^{H'}_{0,i}(\tau_{m})\left(\hat{f}_{t}-f^{H}_{0,t}\right)}_{A_{3mi}}+\underbrace{\frac{1}{T}\sum_{t}\bar{\hat{R}}_{h,\tau_{m},it}^{(2)}\hat{f}_{t}\hat{f}_{t}'\left(\hat{\lambda}_{i}(\tau_{m})-\lambda^{H}_{0,i}(\tau_{m})\right)}_{A_{4mi}}\notag\\
&+\underbrace{\frac{1}{T}\sum_{t}\bar{\hat{R}}_{h,\tau_{m}}^{(3)}\left(\lambda_{0,i}^{H'}(\tau_{m})f^{**}_{0,t};Y_{it}\right)\hat{f}_{t}\hat{f}_{t}'\left(\hat{\lambda}_{i}(\tau_{m})-\lambda^{H}_{0,i}(\tau_{m})\right)\lambda_{0,i}^{H'}(\tau_{m})\left(\hat{f}_{t}-f_{0,i}^{H}\right)}_{A_{5mi}} \notag\\
&+\underbrace{\frac{0.5}{T}\sum_{t}\bar{\hat{R}}_{h,\tau_{m}}^{(3)}\left(\lambda_{i}^{*'}(\tau_{m})\hat{f}_{t};Y_{it}\right)\hat{f}_{t}\left[\hat{f}_{t}'\left(\hat{\lambda}_{i}(\tau_{m})-\lambda^{H}_{0,i}(\tau_{m})\right)\right]^{2}}_{A_{6mi}},\label{eq.rate.expansion}
\end{align}
\normalsize
where $\lambda_{i}^{*}(\tau_{m})$ lies between $\hat{\lambda}_{i}(\tau_{m})$ and $\lambda^{H}_{0,t}(\tau_{m})$ and $f^{*}_{t}$ and $f^{**}_{t}$ lie between $\hat{f}_{t}$ and $f^{H}_{0,t}$. 

By Lemma \ref{lem.nonpara}-(ii), $A_{1mi}=O\left(h^{\gamma}\right)$ uniformly in $m$ and $i$. By the uniform boundedness of $\hat{R}_{h,\tau_{m}}^{(1)}(\cdot;Y_{it})$ and by Theorem \ref{thm.roc}-(i), $A_{2mi}=O_{p}\left(\zeta_{NT}\right)$ uniformly in $m$ and $i$. Lemma \ref{lem.nonpara}-(ii), uniform boundedness of $\textsf{f}_{\tau_{m},it}(\cdot)$ and Theorem \ref{thm.roc}-(i) imply that $A_{3mi}=O_{p}(\zeta_{NT})+O(h^{\gamma})$ uniformly in $t$. For $A_{5mi}$ and $A_{6mi}$, they are both $o_{p}(1)\cdot(\hat{\lambda}_{i}(\tau_{m})-\lambda^{H}_{0,i}(\tau_{m}))$ by Theorem \ref{thm.roc}-(i) and Lemma \ref{lem.uniformcons} where the term $o_{p}(1)$ is uniform in $m$ and $i$.

Finally, consider $A_{4mi}$. Let $Q_{\Lambda,mi}\coloneqq \sum_{t=1}^{T}\textsf{f}_{\tau_{m},it}(0)f^{H}_{0,t}f_{0,t}^{H'}/T$. By uniform consistency of $\hat{f}_{t}$ and Lemma \ref{lem.nonpara}-(ii),
\[
  A_{4mi}=\left(Q_{\Lambda,mi}+o_{p}(1)\right)\left(\hat{\lambda}_{i}(\tau_{m})-\lambda^{H}_{0,i}(\tau_{m})\right).
\]
Hence, equation \eqref{eq.rate.expansion} is now written as:
\begin{equation}\label{eq.rate.expansion.simplified}
  \left(Q_{\Lambda,,mi}+o_{p}\left(1\right)\right)\left(\hat{\lambda}_{i}(\tau_{m})-\lambda^{H}_{0,i}(\tau_{m})\right)=\frac{1}{T}\sum_{t}\bar{\hat{R}}_{h,\tau_{m}}^{(1)}\left(\hat{\lambda}'_{i}(\tau_{m})\hat{f}_{t}';Y_{it}\right)\hat{f}_{t}+O_{p}\left(\zeta_{NT}\right).
\end{equation}
For the right-hand side, note that we have the first order condition $\sum_{t}\hat{R}^{(1)}_{h,\tau_{m}}\left(\hat{\lambda}_{i}'(\tau_{m})\hat{f}_{t};Y_{it}\right)\hat{f}_{t}/T=0$. So we have\footnotesize
\begin{align*}
&\frac{1}{T}\sum_{t}\bar{\hat{R}}_{h,\tau_{m}}^{(1)}\left(\hat{\lambda}_{i}'(\tau_{m})\hat{f}_{t};Y_{it}\right)\hat{f}_{t}\\
=&-\frac{1}{T}\sum_{t}\left(\hat{R}_{h,\tau_{m},it}^{(1)}-\bar{\hat{R}}_{h,\tau_{m},it}^{(1)}\right)\hat{f}_{t}\\
&-\frac{1}{T}\sum_{t}\Bigg\{\left[\hat{R}_{h,\tau_{m}}^{(1)}\left(\hat{\lambda}_{i}'(\tau_{m})\hat{f}_{t};Y_{it}\right)-\bar{\hat{R}}_{h,\tau_{m}}^{(1)}\left(\hat{\lambda}_{i}'(\tau_{m})\hat{f}_{t};Y_{it}\right)\right]-\left(\hat{R}_{h,\tau_{m},it}^{(1)}-\bar{\hat{R}}_{h,\tau_{m},it}^{(1)}\right)\Bigg\}\hat{f}_{t}\\
=&-\frac{1}{T}\sum_{t}\left(\hat{R}_{h,\tau_{m},it}^{(1)}-\bar{\hat{R}}_{h,\tau_{m},it}^{(1)}\right)f_{0,t}^{H}\\
&-\frac{1}{T}\sum_{t}\Bigg\{\left[\hat{R}_{h,\tau_{m}}^{(1)}\left(\hat{\lambda}_{i}'(\tau_{m})\hat{f}_{t};Y_{it}\right)-\bar{\hat{R}}_{h,\tau_{m}}^{(1)}\left(\hat{\lambda}_{i}'(\tau_{m})\hat{f}_{t};Y_{it}\right)\right]-\left(\hat{R}_{h,\tau_{m},it}^{(1)}-\bar{\hat{R}}_{h,\tau_{m},it}^{(1)}\right)\Bigg\}f^{H}_{0,t}+O_{p}\left(\zeta_{NT}\right) \\
=&-\underbrace{\frac{1}{T}\sum_{t}\left(\hat{R}_{h,\tau_{m},it}^{(1)}-\bar{\hat{R}}_{h,\tau_{m},it}^{(1)}\right)f_{0,t}^{H}}_{B_{1mi}}\\
&-\frac{1}{T}\sum_{t}\underbrace{\Bigg\{\left[\hat{R}_{h,\tau_{m}}^{(1)}\left(\hat{\lambda}_{i}'(\tau_{m})\hat{f}_{t};Y_{it}\right)-\bar{\hat{R}}_{h,\tau_{m}}^{(1)}\left(\hat{\lambda}_{i}'(\tau_{m})\hat{f}_{t};Y_{it}\right)\right]-\left[\hat{R}_{h,\tau_{m}}^{(1)}\left(\hat{\lambda}_{i}'(\tau_{m})f_{0,t}^{H};Y_{it}\right) -\bar{\hat{R}}_{h,\tau_{m}}^{(1)}\left(\hat{\lambda}_{i}'(\tau_{m})f_{0,t}^{H};Y_{it}\right) \right]\Bigg\}}_{B_{2mit}}\\
&\cdot f^{H}_{0,t}\\
&-\underbrace{\frac{1}{T}\sum_{t}\Bigg\{\left[\hat{R}^{(1)}_{h,\tau_{m}}\left(\hat{\lambda}'_{0}(\tau_{m})f^{H}_{0,t};Y_{it}\right) -\bar{\hat{R}}^{(1)}_{h,\tau_{m}}\left(\hat{\lambda}'_{0}(\tau_{m})f^{H}_{0,t};Y_{it}\right) \right]-\left(\hat{R}^{(1)}_{h,\tau_{m},it}-\bar{\hat{R}}^{(1)}_{h,\tau_{m},it}\right)\Bigg\}f^{H}_{0,t}}_{B_{3mi}}\\
&+O_{p}\left(\zeta_{NT}\right),
\end{align*}
\normalsize
where $O_{p}(\zeta_{NT})$ is uniform in $m,i$ by the boundedness of $\hat{R}^{(1)}_{h,\tau_{m},it}$.

For $B_{1mi}$, since $H_{NT,1}\in\mathcal{H}$ which only consists of $r^{2}$ elements, we can show that it is $O_{p}(\sqrt{\log N}\zeta_{NT})$ uniformly in $m,i$ by the Hoeffding's inequality under the boundedness of $\hat{R}^{(1)}_{h,\tau_{m},it}$.

For $\sum_{t}B_{2mit}f_{0,t}^{H}/T$, by the mean value theorem, \small
\begin{align*}
&\max_{m,i}\left|\frac{1}{T}\sum_{t}B_{2mit}f_{0,t}^{H}\right| \\
=&\max_{m,i}\left|\frac{1}{T}\sum_{t}\left[\hat{R}_{h,\tau_{m}}^{(2)}\left(\hat{\lambda}_{i}'(\tau_{m}){f}^{***}_{t};Y_{it}\right)-\bar{\hat{R}}_{h,\tau_{m}}^{(2)}\left(\hat{\lambda}_{i}'(\tau_{m}){f}^{***}_{t};Y_{it}\right)\right]f^{H}_{0,t}\hat{\lambda}_{i}'(\tau_{m})\left(\hat{f}_{t}-f^{H}_{0,t}\right)\right|= O_{p}\left(\frac{\zeta_{NT}}{h}\right),
\end{align*}
\normalsize
where $f^{***}_{t}$ lies between $f_{0,t}^{H}$ and $\hat{f}_{t}$. The last equality is by the uniform boundedness of $h\hat{R}^{(2)}_{h,\tau_{m}}(\cdot;Y_{it})$ and $H_{NT,1}$ and by Theorem \ref{thm.roc}-(i).

For $B_{3mi}$, similarly,\footnotesize
\begin{align*}
B_{3mi}=\underbrace{\frac{1}{T}\sum_{t}\left[\hat{R}^{(2)}_{h,\tau_{m}}\left(\lambda_{i}^{**'}(\tau_{m})f_{0,t}^{H};Y_{it}\right) -\bar{\hat{R}}^{(2)}_{h,\tau_{m}}\left(\lambda_{i}^{**'}(\tau_{m})f_{0,t}^{H};Y_{it}\right) \right]f_{0,t}^{H}f_{0,t}^{H'}}_{C_{mi}}\cdot\left(\hat{\lambda}_{i}(\tau_{m})-\lambda^{H}_{0,i}(\tau_{m})\right).
\end{align*}
\normalsize
where $\lambda_{i}^{**}(\tau_{m})$ lies between $\lambda_{0,i}^{H}(\tau_{m})$ and $\hat{\lambda}_{i}(\tau_{m})$. Note that
\begin{align*}
\max_{m,i}\left|C_{t}\right|\leq&\max_{m,i,H\in\mathcal{H}}\sup_{\lambda\in\mathcal{B}^{r}}\left|\frac{1}{T}\sum_{t=1}^{T}\left(\hat{R}^{(2)}_{h,\tau_{m}}\left(\lambda'f_{0,t};Y_{it}\right) -\bar{\hat{R}}^{(2)}_{h,\tau_{m}}\left(\lambda'f_{0,t};Y_{it}\right) \right)Hf_{0,t}f_{0,t}'H\right|=o_{p}(1),
\end{align*}
where the equality is by the standard results for kernel density estimation and by the uniform boundedness of $f_{0,t}$. Substitute $B_{1mi}$ to $B_{3mi}$ to equation \eqref{eq.rate.expansion.simplified}, and we have
\begin{equation*}\label{eq.unif.rate.final}
  \max_{m,i}\left[\left(Q_{\Lambda,mi}+o_{p}\left(1\right)\right)\left(\hat{\lambda}_{i}(\tau_{m})-\lambda^{H}_{0,i}(\tau_{m})\right)\right]=O_{p}\left(\frac{\zeta_{NT}}{h}\right).
\end{equation*}

Finally, Assumptions \ref{assum.normalization} and \ref{assum.density} imply that all the eigenvalues of $Q_{\Lambda,mi}$ are bounded away from 0 uniformly in $m,i$ and in realizations of $H_{NT,1}$ for sufficiently large $N$ and $T$. 
The desired result thus follows.
\end{proof}

\section{Proof of Theorem \ref{thm.firststage}}
Note that
\begin{align*}
\max_{m,i,t}\left|\widehat{\frac{1}{\textsf{f}_{\tau_{m},it}(0)}}-\frac{1}{\textsf{f}_{\tau_{m},it}(0)}\right|=\max_{a,b\in\{1,2\}}\max_{i\in\mathcal{N}_{a},t\in\mathcal{N}_{b},m}\left|\widehat{\frac{1}{\textsf{f}_{\tau_{m},it}(0)}}-\frac{1}{\textsf{f}_{\tau_{m},it}(0)}\right|.
\end{align*}
For simplicity, we only consider the uniform rate for $(i,t)\in\mathcal{N}_{2}\times\mathcal{T}_{2}$. The other cases follow exactly the same argument and have the same rate. Then the result follows.

Note that $F_{0}/\sqrt{T}$ is in general no longer the eigenvector matrix of the top submatrix of $\sum_{m}F_{0}'\sum_{i\in\mathcal{N}_{1}}\lambda_{0,i}(\tau_{m})\lambda_{0,i}'(\tau_{m})F_{0}/MNT$. However, there exists a full-rank $r\times r$ matrix $H^{eigen}$ such that $F_{0}H^{eigen}/\sqrt{T}$ is the eigenvector matrix because $\sum_{m}\sum_{i\in\mathcal{N}_{1}}\lambda_{0,i}(\tau_{m})\lambda_{0,i}'(\tau_{m})/MN$ is full-rank by Assumption \ref{assum.subsample}. Correspondingly, $\sum_{m}\sum_{i\in\mathcal{N}_{1}}(H^{eigen})^{-1}\lambda_{0,i}(\tau_{m})\lambda_{0,i}'(\tau_{m})(H^{eigen'})^{-1}/MN$, being the eigenvalue matrix, is diagonal with distinct diagonal entries bounded away from 0. From the proof of Theorem \ref{thm.roc}, we can only consistently estimate $F_{0}H^{eigen}$ (up to column signs). Nevertheless, we note that the density to be estimated is invariant to full-rank transformations of $F_{0}$, i.e., for any full-rank $r\times r$ matrix $H$, 
\begin{align}
\frac{1}{\textsf{f}_{\tau,it}(0)}=\lambda^{(1)'}_{0,i}(\tau)f_{0,t}=\left(\lambda'_{0,i}(\tau)(H')^{-1}\right)^{(1)}H'f_{0,t}.\label{eq.density.invariant}
\end{align}
Therefore, we can without loss of generality derive the rate with respect to $F_{0}H^{eigen}$ and correspondingly $\Lambda_{0}(\tau_{m})(H^{eigen'})^{-1}$. With slight abuse of notation, still denote them by $F_{0}$ and $\Lambda_{0}(\tau_{m})$, respectively. 
Then by Assumption \ref{assum.subsample} and following the proof of Theorem \ref{thm.roc}, there exists a diagonal matrix $H_{NT}^{top}\in\mathcal{H}$, where recall that we defined $\mathcal{H}$ in Appendix \ref{appx.roc} as the set of diagonal matrices whose diagonal entries only contain $1$ and $-1$, such that
\begin{align*}
\max_{t=1,\ldots,T}\left\|\hat{f}^{top}_{t}-H_{NT}^{top'}f_{0,t}\right\|_{F}=&O_{p}\left(\frac{\zeta_{NT}}{h}\right),\\
\max_{\substack{i=1,\ldots,N\\m=1,\ldots,M\\c=-2,-1,1,2}}\left\|\hat{\lambda}_{i}^{(t,l)}(\tau_{m}+ch_{d})-\left(H_{NT}^{top}\right)^{-1}\lambda_{0,i}(\tau_{m}+ch_{d})\right\|_{F}=&O_{p}\left(\frac{\zeta_{NT}}{h}\right).
\end{align*}
Note that for the loading estimator, we have the same result as before even though now the quantile levels considered depend on the sample size via $h_{d}$. The reason is that the total number of quantile levels is still finite, equal to $4M$, and the rate $\zeta_{NT}/h$ is driven by the average rate of the factor estimator as can be seen in the proof of Theorem \ref{thm.roc}, unrelated to the quantile levels. The desired result then follows by the triangle inequality, by \eqref{eq.density.invariant} and by that the five-point difference formula has approximation error $O(h_{d}^{4})$.

\section{Proof of Theorem \ref{thm.dist2}}\label{appx.dist}
Similar as before, let $\tilde{F}_{j}$ and $F_{0,j}$ be the $j$-th column in the $T\times r$ matrices $\tilde{F}$ and $F_{0}$. Let $\tilde{H}_{NT,1}\coloneqq\text{diag}(\text{sgn}(\tilde{F}_{j}'F_{0,j}))$. The proofs of the following lemmas are in the Online Appendix.
\begin{lem}\label{lem.tilde.rate}
Under Assumptions \ref{assum.foundation} to \ref{assum.subsample}, 
\begin{align*}
\frac{1}{T}\left\|\tilde{F}-F_{0}\tilde{H}_{NT,1}\right\|_{F}^{2}=O_{p}\left(\zeta_{NT}^{2}\right),\ \frac{1}{MNT}\sum_{m=1}^{M}\left\|\tilde{\Lambda}(\tau_{m})-\Lambda_{0}(\tau_{m})\tilde{H}_{NT,1}^{'-1}\right\|_{F}^{2}=O_{p}\left(\zeta_{NT}^{2}\right),\\
  \max_{t}\left\|\tilde{f}_{t}-\tilde{H}'_{NT,1}f_{0,t}\right\|_{F}=O_{p}\left(\frac{\zeta_{NT}}{h}\right),\ 
   \max_{m,i}\left\|\tilde{\lambda}_{i}(\tau_{m})-\tilde{H}^{-1}_{NT,1}\lambda_{0,i}(\tau_{m})\right\|_{F}=O_{p}\left(\frac{\zeta_{NT}}{h}\right).
\end{align*}
\end{lem}

\begin{lem}\label{lem.inv}
Under Assumptions \ref{assum.foundation} to \ref{assum.subsample}, for $H_{NT,2}=F_{0}'\tilde{F}/T$, we have
\begin{align}
\text{diag}\left(H_{NT,2}\right)=\tilde{H}_{NT,1}+O_{p}\left(\zeta_{NT}^2\right),\label{eq.H2inv.diag}\\
H_{NT,2}=\tilde{H}_{NT,1}+O_{p}(\zeta_{NT}),\label{eq.H2inv.0}\\
H_{NT,2}H_{NT,2}'=I_{r}+O_{p}\left(\zeta_{NT}^{2}\right),\label{eq.H2inv.1}\\
H_{NT,2}'H_{NT,2}=I_{r}+O_{p}\left(\zeta_{NT}^{2}\right).\label{eq.H2inv.2}
\end{align}
\end{lem}

Recall that $\eta_{h,\tau_{m},it}=K\left((\lambda_{0,i}'(\tau_{m})f_{0,t}-Y_{it})/h\right)-\mathbb{E}\left[ K\left((\lambda_{0,i}'(\tau_{m})f_{0,t}-Y_{it})/h\right)\right]$.
\begin{lem}\label{lem.expansion}
Under the conditions in Theorem \ref{thm.dist2}, we have
\begin{align*}
  \max_{t=1,\ldots,T}\left\|H_{NT,2}^{-1}\Phi H_{NT,2}^{'-1}\left(\tilde{f}_{t}-H_{NT,2}'f_{0,t}\right)+\frac{1}{MN}\sum_{m=1}^{M}\sum_{i=1}^{N}\frac{\eta_{h,\tau_{m},it}}{\textsf{f}_{\tau_{m},it}(0)}H_{NT,2}^{-1}\lambda_{0,i}(\tau_{m})\right\|_{F}=&o_{p}\left(\frac{1}{\sqrt{N}}\right),\\
    \max_{i=1,\ldots,N;m=1,\ldots,M}\left\|\tilde{\lambda}_{i}(\tau_{m})-H_{NT,2}^{-1}\lambda_{0,t}(\tau_{m})+\frac{1}{T}\sum_{t=1}^{T}\frac{\eta_{h,\tau_{m},it}}{\textsf{f}_{\tau_{m},it}(0)}H_{NT,2}'f_{0,t}\right\|_{F}=&o_{p}\left(\frac{1}{\sqrt{T}}\right).
\end{align*}
\end{lem}

\begin{proof}[Proof of Theorem \ref{thm.dist2}]
 We first derive the average rate. To save space, we again only focus on $\tilde{F}$ since the result for $\tilde{\Lambda}(\tau_{m})$ follows by a similar argument. Consider
\[
  A\coloneqq \frac{1}{T}\sum_{t=1}^{T}\left[\frac{1}{N}\sum_{i=1}^{N}\left(\frac{1}{M}\sum_{m=1}^{M}\frac{\eta_{h,\tau_{m},it}}{\textsf{f}_{\tau_{m},it}(0)}\lambda^{'}_{0,i}(\tau_{m})\right)\right]\left[\frac{1}{N}\sum_{i=1}^{N}\left(\frac{1}{M}\sum_{m=1}^{M}\frac{\eta_{h,\tau_{m},it}}{\textsf{f}_{\tau_{m},it}(0)}\lambda_{0,i}(\tau_{m})\right)\right].
\]
Note that $\mathbb{E}(A)=O(1/N)$ in view that $\mathbb{E}(\eta_{h,\tau_{m},it})=0$ and the $\eta_{h,\tau_{m},it}$s are independent across $i$ and $t$. So, $A=O_{p}(1/N)$ by Markov's inequality. By Lemma \ref{lem.expansion} and by Cauchy-Schwarz inequality, 
\begin{align*}
&\frac{1}{T}\sum_{t=1}^{T}\left(\tilde{f}_{t}-H_{NT,2}'f_{0,t}\right)'H_{NT,2}^{-1}\Phi^{2}H_{NT,2}^{-1'}\left(\tilde{f}_{t}-H_{NT,2}'f_{0,t}\right)\\
\leq &A+\sqrt{A}\cdot o_{p}\left(\frac{1}{\sqrt{N}}\right)+o_{p}\left(\frac{1}{N}\right)=O_{p}\left(\frac{1}{N}\right).
\end{align*}
Since the smallest eigenvalue of $H_{NT,2}^{-1}\Phi^{2} H_{NT,2}^{'-1}$ is bounded away from zero with probability approaching 1, implied by Lemmas \ref{lem.discrete} and \ref{lem.inv}, we have
\[
  \frac{1}{T}\sum_{t=1}^{T}\left(\tilde{f}_{t}-H_{NT,2}'f_{0,t}\right)'\left(\tilde{f}_{t}-H_{NT,2}'f_{0,t}\right)=O_{p}\left(\frac{1}{N}\right).
\]

We now derive the asymptotic distribution of $\tilde{f}_{t}$. Let
\[
X_{i,T}(t)=\frac{1}{M}\sum_{m=1}^{M}\frac{\eta_{h,\tau_{m},it}}{\textsf{f}_{\tau_{m},it}(0)}\lambda_{0,i}\left(\tau_{m}\right),
\]
where the dependence of $X_{i,T}(t)$ on $T$ is through $h$. The variance of $X_{i,T}(t)$ satisfies
\begin{align*}
  Var(X_{i,T}(t))=&\frac{1}{M^{2}}\sum_{m=1}^{M}\sum_{m'=1}^{M}\frac{\mathbb{E}\left(\eta_{h,\tau_{m},it}\cdot \eta_{h,\tau_{m'},it}\right)}{\textsf{f}_{\tau_{m},it}(0)\textsf{f}_{\tau_{m'},it}(0)}\lambda_{0,i}(\tau_{m})\lambda_{0,i}'(\tau_{m'}) \\
  =&\frac{1}{M^{2}}\sum_{m=1}^{M}\sum_{m'=1}^{M}\frac{\min\{\tau_{m},\tau_{m'}\}-\tau_{m}\tau_{m'}}{\textsf{f}_{\tau_{m},it}(0)\textsf{f}_{\tau_{m'},it}(0)}\lambda_{0,i}(\tau_{m})\lambda_{0,i}'(\tau_{m'})+o(1),
\end{align*}
where the second equality is by Lemma \ref{lem.nonpara}-(iii). By independence of $X_{i,T}(t)$ across $i$, $Var(\sum_{i=1}^{N}X_{i,T}(t))=\sum_{i=1}^{N}Var(X_{i,T}(t))$.
Hence, by the definition of $\Sigma_{F,t}$ and by Theorem 2 in \cite{phillips1999linear}, we have
\begin{equation*}
  \Sigma_{F,t}^{-1/2}\Phi H_{NT,2}^{'-1}\sqrt{N}\left(\tilde{f}_{t}-H_{NT,2}'f_{0,t}\right)\overset{d}{\to}N\left(0,I_{r}\right).
\end{equation*}
The limiting distribution of $\tilde{\lambda}_{i}(\tau^{*})$ is derived similarly. 

For the limiting distribution of $\tilde{\lambda}_{i}'(\tau^{*})\tilde{f}_{t}$, first,
\begin{align*}
\tilde{\lambda}_{i}'(\tau^{*})\tilde{f}_{t}-L_{0,it}(\tau^{*})=\lambda_{0}'(\tau^{*})H_{NT,2}'^{-1}\left(\tilde{f}_{t}-H_{NT,2}'f_{0,t}\right)+\left(\tilde{\lambda}_{i}(\tau^{*})-H_{NT,2}^{-1}\lambda_{0}(\tau^{*})\right)'\tilde{f}_{t}.
\end{align*}
For the second term,
\begin{align*}
&\left(\tilde{\lambda}_{i}(\tau^{*})-H_{NT,2}^{-1}\lambda_{0}(\tau^{*})\right)'\tilde{f}_{t}\\
=&\left(\tilde{\lambda}_{i}(\tau^{*})-H_{NT,2}^{-1}\lambda_{0}(\tau^{*})\right)'H_{NT,2}'f_{0,t}+\left(\tilde{\lambda}_{i}(\tau^{*})-H_{NT,2}^{-1}\lambda_{0}(\tau^{*})\right)'\left(\tilde{f}_{t}-H_{NT,2}'f_{0,t}\right)\\
=&\left(\tilde{\lambda}_{i}(\tau^{*})-H_{NT,2}^{-1}\lambda_{0}(\tau^{*})\right)'H_{NT,2}'f_{0,t}+o_{p}\left(\frac{1}{\sqrt{T}}\right).
\end{align*}
Therefore, by $H_{NT,2}H_{NT,2}'=I_{r}+O_{p}(\zeta_{NT}^{2})$ (Lemma \ref{lem.inv}), since $N$ and $T$ have the same order, we have
\begin{align*}
&\left(\tilde{\lambda}_{i}'(\tau^{*})\tilde{f}_{t}-L_{0,it}(\tau^{*})\right)\\
=&\lambda_{0,i}'(\tau^{*})\Phi^{-1}\frac{1}{N}\sum_{j=1}^{N}\frac{1}{M}\sum_{m=1}^{M}\frac{\eta_{h,\tau^{*},jt}}{\textsf{f}_{\tau^{*},it}(0)} \lambda_{0,j}(\tau^{*})+f_{0,t}'\frac{1}{T}\sum_{s=1}^{T}\frac{\eta_{h,\tau^{*},it}}{\textsf{f}_{\tau^{*},is}(0)}f_{0,s}+o_{p}\left(1\right).
\end{align*}
Since the two leading terms are asymptotically independent, the desired result follows from the Cram\'er-Wold theorem.

Finally, when $r=1$, $H_{NT,2}=\tilde{H}_{NT,1}+O_{p}(\zeta_{NT}^{2})$ by equation \eqref{eq.H2inv.diag} in Lemma \ref{lem.inv}.
\end{proof}

 \section {Proof of Theorem \ref{thm.dist3}}
Substitute equation \eqref{eq.mean.model} and the stochastic expansion of $\tilde{f}_{t}$ in Lemma \ref{lem.expansion} into equation \eqref{eq.meanlambda}:
\begin{align*}
\tilde{\bar{\lambda}}_{i}=&\frac{1}{T}\sum_{t=1}^{T}\tilde{f}_{t}f_{0,t}'\bar{\lambda}_{0,i}+\frac{1}{T}\sum_{t}^{T}\nu_{it}H_{NT,2}'f_{0,t}\\
&-\left[\frac{1}{MN}\sum_{j=1}^{N}\sum_{m=1}^{M}H_{NT,2}^{-1}\lambda_{0,j}(\tau_{m})\lambda_{0,j}'(\tau_{m})H_{NT,2}^{'-1}\right]^{-1}\frac{1}{MNT}\sum_{t}^{T}\sum_{j=1}^{N}\sum_{m=1}^{M}\frac{\nu_{it}\eta_{h,\tau_{m},jt}}{\textsf{f}_{\tau_{m},jt}(0)}H_{NT,2}^{-1}\lambda_{0,j}(\tau_{m})\\
&+o_{p}\left(\frac{1}{\sqrt{T}}\right)\\
=&\frac{1}{T}\sum_{t=1}^{T}\tilde{f}_{t}f_{0,t}'\bar{\lambda}_{0,i}+\frac{1}{T}\sum_{t}^{T}\nu_{it}H_{NT,2}'f_{0,t}+o_{p}\left(\frac{1}{\sqrt{T}}\right),
\end{align*}
where the term $o_{p}(1/\sqrt{T})$ is uniform in $i$ by a similar argument as the proof of Lemma \ref{lem.expansion} and by the boundedness of $\nu_{it}$. Therefore, by $\sum_{t=1}^{T}\tilde{f}_{t}f_{0,t}'/T= H_{NT,2}'$ and $H_{NT,2}'=H_{NT,2}^{-1}+O_{p}(\zeta_{NT}^{2})$ implied by Lemma \ref{lem.inv},
\begin{equation*}
  \tilde{\bar{\lambda}}_{i}-H_{NT,2}^{-1}\bar{\lambda}_{0,i}=\frac{1}{T}\sum_{t}^{T}\nu_{it}H_{NT,2}'f_{0,t}+o_{p}\left(\frac{1}{\sqrt{T}}\right).
\end{equation*}
The desired results then follow similar arguments as in the proof of Theorem \ref{thm.dist2}.
 \section{Proof of Theorems \ref{thm.r} and \ref{thm.rstrong}}
\begin{proof}[Proof of Theorem \ref{thm.r}]
By the triangle inequality, there exist constants $C_{1},C_{2}>0$ such that
\begin{align*}
\left\|\frac{1}{MNT}\sum_{m=1}^{M}\left(\hat{L}^{pel'}\left(\tau_{m}\right)\hat{L}^{pel}\left(\tau_{m}\right)-L_{0}'\left(\tau_{m}\right)L_{0}\left(\tau_{m}\right)\right)\right\|_{F}^{2}\leq&\max_{m}\frac{C_{1}}{NT}\left\|\hat{L}^{pel}\left(\tau_{m}\right)-L_{0}\left(\tau_{m}\right)\right\|_{F}^{2}\\
\leq& \frac{C_{2}\log(NT)}{\min\{N,T\}},
\end{align*}
with probability approaching 1, where the last inequality follows \cite{feng2023nuclear}. By $F_{0}'F_{0}/T=I_{r}$, the diagonal entries of $\sum_{m=1}^{M}\Lambda_{0}'(\tau_{m})\Lambda_{0}(\tau_{m})/MN$, i.e., $\sigma_{1}^{2},\ldots,\sigma_{r}^{2}$, are equal to the nonzero eigenvalues of $\sum_{m=1}^{M}L_{0}'(\tau_{m})L_{0}(\tau_{m})/MNT$, all distinct and bounded away from 0 by Lemma \ref{lem.discrete}. Therefore, by Weyl's inequality,
\begin{align*}
\max_{j=1,\ldots,\min\{N,T\}}\left|\hat{\sigma}_{j}^{2}-\sigma_{j}^{2}\right|\leq \left\|\frac{1}{MNT}\sum_{m=1}^{M}\left(\hat{L}^{pel'}\left(\tau_{m}\right)\hat{L}^{pel}\left(\tau_{m}\right)-L_{0}'\left(\tau_{m}\right)L_{0}\left(\tau_{m}\right)\right)\right\|_{F}\leq \sqrt{\frac{C_{2}\log(NT)}{\min\{N,T\}}},
\end{align*}
with probability approaching 1. Now note that the event $1(\hat{\sigma}_{j}^{2}\geq C_{r})=1$ for all $j=1,\ldots,r$ and  $1(\hat{\sigma}_{j}^{2}\geq C_{r})=0$ for all $j>r$ implies $\hat{r}=r$. Therefore, by $\sigma^{2}_{r+1}=\cdots=\sigma^{2}_{\min(N,T)}=0$, 
\begin{align*}
\Pr\left(\hat{r}=r\right)\geq & \Pr\left(\hat{\sigma}_{j}^{2}\geq C_{r},\forall j\leq r\text{ and }\hat{\sigma}_{j}^{2}<C_{r},\forall j>r \right)\\
\geq &\Pr \left(\max_{j\leq r}|\hat{\sigma}_{j}^{2}-\sigma_{j}^{2}|\leq \sigma_{r}^{2}-C_{r}\text{ and }\max_{j>r}|\hat{\sigma}_{j}^{2}-\sigma_{j}^{2}|<C_{r}\right)\\
\geq &\Pr \left(\max_{j=1,\ldots,\min\{N,T\}}|\hat{\sigma}_{j}^{2}-\sigma_{j}^{2}|\leq \sqrt{\frac{C_{2}\log(NT)}{\min\{N,T\}}}\right)\to 1,
\end{align*}
where the last inequality holds because both $(\sigma_{r}^{2}-C_{r})$ and $C_{r}$ are greater than $\sqrt{\log(NT)/{\min\{N,T\}}}$ in order by $\sigma_{r}^{2}$ being bounded away from 0.
\end{proof}

\begin{proof}[Proof of Theorem \ref{thm.rstrong}]
We only prove consistency of $\tilde{\bar{r}}(\alpha)$ since consistency of $\tilde{r}_{\tau_{m}}(\alpha)$ follows exactly the same argument. Note that the singular values of $\bar{L}_{0}/\sqrt{NT}$ are the square root of the eigenvalues of $\bar{L}'_{0}\bar{L}_{0}/NT$. By Weyl's inequality,
\begin{align}
  \max_{j=1,\ldots,\min\{N,T\}}\left|\tilde{\bar{\sigma}}_{j}-\bar{\sigma}_{j}\right|\leq &\left\|\frac{\tilde{\bar{\Lambda}}\tilde{F}'}{\sqrt{NT}}-\frac{\bar{\Lambda}_{0}F_{0}'}{\sqrt{NT}}\right\|_{F}\notag\\
  \leq & \left\|\frac{\tilde{\bar{\Lambda}}}{\sqrt{N}}\right\|_{F}\cdot \left\|\frac{\tilde{F}-F_{0}H_{NT,2}}{\sqrt{T}}\right\|_{F}+\left\|\frac{F_{0}H_{NT,2}'}{\sqrt{T}}\right\|_{F}\left\|\frac{\tilde{\bar{\Lambda}}-\bar{\Lambda}_{0}H_{NT,2}^{'-1}}{\sqrt{N}}\right\|_{F}\notag\\
  =&O_{p}\left(\frac{1}{\sqrt{N}}\right)=o\left(\frac{N^{\frac{\alpha-1}{2}}}{\log N}\right),\forall \alpha_{j}\in (0,1].\label{eq.singularvalue}
\end{align}
where the penultimate equality is by Theorems \ref{thm.dist2} and \ref{thm.dist3}. Let $N^{(\tilde{\alpha}-1)/2}$ be the order of the $(\bar{r}(\alpha)+1)$-th singular value; if $\bar{r}=\bar{r}(\alpha)$, let $\tilde{\alpha}=-\infty$. By construction, $\tilde{\alpha}<\alpha$. By definition, $\bar{\sigma}_{1}\geq \cdots\geq \bar{\sigma}_{\bar{r}(\alpha)}\geq C_{1}N^{(\alpha-1)/2}$ for some $C_{1}$ whereas $\bar{\sigma}_{j}\leq C_{2}N^{(\tilde{\alpha}-1)/2}=o(N^{(\alpha-1)/2}/\log(N))$ for some $C_{2}$ for all $j>\bar{r}(\alpha)$ because $\tilde{\alpha}<\alpha$. Therefore, for any constant $C>0$, there exists a $C_{3}>0$ such that
\begin{align*}
&\Pr\left(\tilde{\bar{r}}(\alpha)=\bar{r}(\alpha)\right)\\
\geq & \Pr\left(\tilde{\bar{\sigma}}_{j}\geq \frac{CN^{\frac{\alpha-1}{2}}}{\log(N)},\forall j\leq \bar{r}(\alpha)\text{ and }\tilde{\bar{\sigma}}_{j}<\frac{CN^{\frac{\alpha-1}{2}}}{\log(N)},\forall j>\bar{r}(\alpha) \right)\\
\geq &\Pr \left(\max_{j\leq \bar{r}(\alpha)}|\tilde{\bar{\sigma}}_{j}-\bar{\sigma}_{j}|\leq C_{1}N^{\frac{\alpha-1}{2}}-\frac{CN^{\frac{\alpha-1}{2}}}{\log(N)}\text{ and }\max_{j>\bar{r}(\alpha)}|\tilde{\bar{\sigma}}_{j}-\bar{\sigma}_{j}|<\frac{CN^{\frac{\alpha-1}{2}}}{\log(N)}-C_{2}N^{\frac{\tilde{\alpha}-1}{2}}\right)\\
\geq &\Pr \left(\max_{j\leq \bar{r}(\alpha)}|\tilde{\bar{\sigma}}_{j}-\bar{\sigma}_{j}|\leq \frac{C_{3}N^{\frac{\alpha-1}{2}}}{\log(N)}\text{ and }\max_{j>\bar{r}(\alpha)}|\tilde{\bar{\sigma}}_{j}-\bar{\sigma}_{j}|\leq \frac{C_{3}N^{\frac{\alpha-1}{2}}}{\log(N)}\right)\to 1,
\end{align*}
where convergence in the last line follows from equation \eqref{eq.singularvalue}.
\end{proof}

\newpage
 \setcounter{equation}{0}
 \section*{\centering Supplemental Appendix to ``Universal Factor Models''}
\centering{\textbf{Abstract}}

This supplement gives proofs of the lemmas in the Appendices \ref{appx.roc} and \ref{appx.dist}.

\renewcommand{\thesection}{S.A}
\setcounter{equation}{0}
\section{Proofs of Lemmas in Appendix \ref{appx.roc}}
\begin{proof}[Proof of Lemma \ref{lem.nonpara}]
See for instance \cite{galvao2016smoothed} and \cite{fernandes2021smoothing}.
\end{proof}

\begin{proof}[Proof of Lemma \ref{lem.avgcons}]

By Assumption \ref{assum.density}, $\textsf{f}_{\tau_{m},it}(\cdot)$ on the compact interval between $\max_{\lambda_{1},\lambda_{2},f_{1},f_{2}\in\mathcal{B}^{r}}|\lambda_{1}'f_{1}-\lambda_{2}'f_{2}|$ and 0 is bounded away from $0$ uniform in $i,t$ and $m$. Let this lower bound be $\underline{\textsf{f}}$. By Lemma \ref{lem.nonpara}-(ii), 
\[
  \frac{\underline{\textsf{f}}}{2}\leq \inf_{(\lambda',f')\in\mathcal{B}^{2r},m,i,t}\textsf{f}_{\tau_{m},it}(\lambda'f-\lambda_{0,i}'(\tau_{m})f_{0,t})+O(h^{\gamma})\leq \inf_{(\lambda',f')'\in\mathcal{B}^{2r},m,i,t}\mathbb{E}\left(\hat{R}^{(2)}_{h,\tau_{m}}\left(\lambda'f;Y_{it}\right)\right),
\]
Meanwhile, $\bar{\Delta}_{\hat{R}_{h}}(\theta_{0})=0$ by definition and $\mathbb{E}\left(\hat{R}^{(1)}_{h,\tau_{m},it}\right)=O(h^{\gamma})$ uniformly in $m$, $i$ and $t$. Hence, expand $\bar{\Delta}_{\hat{R}_{h}}(\hat{\theta})$ around the $L_{0,it}(\tau_{m})$s, and we get
\begin{align}
  &\bar{\Delta}_{\hat{R}_{h}}(\hat{\theta})\geq \frac{1}{MNT}\sum_{m,i,t}\mathbb{E}\left(\hat{R}^{(1)}_{h,\tau_{m},it}\right)\left(\hat{\lambda}_{i}'(\tau_{m})\hat{f}_{t}-\lambda_{0,i}'(\tau_{m})f_{0,t}\right)\notag\\
  & +\frac{1}{4MNT}\sum_{m,i,t}\underline{\textsf{f}}\cdot\left(\hat{\lambda}_{i}'(\tau_{m})\hat{f}_{t}-\lambda_{0,i}'(\tau_{m})f_{0,t}\right)^{2}=\frac{\underline{\textsf{f}}}{4MNT}\sum_{m=1}^{M}\left\|\hat{L}(\tau_{m})-L_{0}(\tau_{m})\right\|_{F}^{2}+O\left(h^{\gamma}\right).\label{eq.expectation.lb}
  \end{align}
On the other hand, $\Delta_{\hat{R}_{h}}(\hat{\theta})\leq\Delta_{\hat{R}_{h}}\left(\theta_{0}\right)=0$ by the definition of the estimator. Therefore, 
\begin{align*}
&\frac{1}{MNT}\sum_{m=1}^{M}\left\|\hat{L}(\tau_{m})-L_{0}(\tau_{m})\right\|_{F}^{2}\leq \frac{4}{\underline{\textsf{f}}}\bar{\Delta}_{\hat{R}_{h}}\left(\hat{\theta}\right)+O\left(h^{\gamma}\right)\\
&\leq \frac{4}{\underline{\textsf{f}}}\left(\bar{\Delta}_{\hat{R}_{h}}\left(\hat{\theta}\right)-\Delta_{\hat{R}_{h}}\left(\hat{\theta}\right)\right)+O\left(h^{\gamma}\right)\leq \frac{4}{\underline{\textsf{f}}}\sup_{\theta\in\Theta}\left|\hat{S}_{h}(\theta)\right|+O\left(h^{\gamma}\right).
\end{align*}
It thus suffices to show that for any $\varepsilon>0$, $\Pr(\sup_{\theta\in\Theta}|\hat{S}_{h}(\theta)|>\varepsilon)\to 0$.

Since $\Theta\subseteq\mathcal{B}^{(MN+T)r}$ where the latter is a compact subset of $\mathbb{R}^{(MN+T)r}$, $\Theta$ can be covered by $K$ cubes $\mathcal{I}_{k}$ ($k=1,\ldots,K$) with center $\theta_{k}$ and length of edges $l=\varepsilon_{0}$ for any fixed $\varepsilon_{0}$. Specifically, each $\mathcal{I}_{k}=\prod_{m,i}[\lambda_{i,k}(\tau_{m})-\varepsilon_{0}/2,\lambda_{i,k}(\tau_{m})+\varepsilon_{0}/2]\times \prod_{t}[f_{t,k}-\varepsilon_{0}/2,f_{t,k}+\varepsilon_{0}/2]$. By construction, $K=(C/\varepsilon_{0})^{(MN+T)r}$ for some constant $C$. Thus,
\begin{align}
\sup_{\theta\in \Theta}\left|\hat{S}_{h}(\theta)\right|\leq\underbrace{\max_{k=1,\ldots,K}\sup_{\theta\in \Theta\cap \mathcal{I}_{k}}\left|\hat{S}_{h}(\theta)\right|}_{A_{1}}+\underbrace{\max_{k=1,\ldots,K}\left|\hat{S}_{h}(\theta_{k})\right|}_{A_{2}}.\label{eq.op1.A1A2}
\end{align} 

\sloppy By the uniform boundedness of $\hat{R}_{h,\tau}^{(1)}(c;Y_{it})$ in $c,Y_{it}$ and $\tau$, for any $(\lambda_{a,i}(\tau_{1})',\ldots,\lambda_{a,i}(\tau_{M})',f_{a,t}')'\in\mathcal{B}^{(M+1)r}$ and $(\lambda_{b,i}(\tau_{1})',\ldots,\lambda_{b,i}(\tau_{M})',f_{b,t}')'\in\mathcal{B}^{(M+1)r}$,
\begin{align}
  &\frac{1}{M}\sum_{m=1}^{M}\left|\hat{R}_{h,\tau_{m}}(\lambda_{a,i}'(\tau_{m})f_{a,t};Y_{it})-\hat{R}_{h,\tau_{m}}(\lambda_{b,i}'(\tau_{m})f_{b,t};Y_{it})\right|\notag\\
  \lesssim& \frac{1}{M}\sum_{m=1}^{M}\left|\lambda_{a,i}'(\tau_{m})f_{a,t}-\lambda_{b,i}'(\tau_{m})f_{b,t}\right|\notag\\
  \lesssim&\frac{1}{M}\sum_{m=1}^{M}\left(\left\|\lambda_{a,i}(\tau_{m})-\lambda_{b,i}'(\tau_{m})\right\|_{F}+\left\|f_{a,t}-f_{b,t}\right\|_{F}\right)\notag\\
  \leq&\sqrt{\frac{2}{M}\sum_{m=1}^{M}\left\|\lambda_{a,i}(\tau_{m})-\lambda_{b,i}'(\tau_{m})\right\|_{F}^{2}+2\left\|f_{a,t}-f_{b,t}\right\|_{F}^{2}},\label{eq.obj.lip}
\end{align}
where $\lesssim$ means ``{}left side bounded by a positive constant times the right side'' \citep{van1996weak}. The second inequality is by the boundedness of $\lambda_{a,i}(\tau_{m}),\lambda_{b,i}(\tau_{m}),f_{a,t}$ and $f_{b,t}$ and by triangle inequality. The last inequality follows from the fact that for a vector $(a_{1},\ldots,a_{m},\ldots,a_{M},b_{1},\ldots,b_{m}\ldots,b_{M})$ with $a_{m},b_{m}\geq 0$ for all $m$, $\sum_{m=1}^{M}(a_{m}+b_{m})\leq \sqrt{2M}\sqrt{\sum_{m=1}^{M}(a_{m}^{2}+b_{m}^{2})}$.

Now, for $A_{1}$ in \eqref{eq.op1.A1A2}, equation \eqref{eq.obj.lip} implies that there exist constants $C_{1},C_{2}>0$ such that
\begin{align}
  A_{1}\leq & C_{1}\max_{k}\sup_{\theta\in\Theta\cap \mathcal{I}_{k}}\frac{1}{NT}\sum_{i,t}\sqrt{\frac{1}{M}\sum_{m=1}^{M}\left\|\lambda_{i}(\tau_{m})-\lambda_{i,k}(\tau_{m})\right\|_{F}^{2}+\left\|f_{t}-f_{t,k}\right\|_{F}^{2}}\notag\\
  \leq &C_{1}\max_{k}\sup_{\theta\in\Theta\cap \mathcal{I}_{k}}\sqrt{\frac{1}{NT}}\sqrt{\sum_{i,t}\left(\frac{1}{M}\sum_{m=1}^{M}\left\|\lambda_{i}(\tau_{m})-\lambda_{i,k}(\tau_{m})\right\|_{F}^{2}+\left\|f_{t}-f_{t,k}\right\|_{F}^{2}\right)}\leq C_{2}\varepsilon_{0},\label{eq.op1.A1}
\end{align} 
where the last inequality is by the definition of $\mathcal{I}_{k}$.

Next, we consider $A_{2}$. Under Assumption \ref{assum.iid}, random variables $\sum_{m=1}^{M}(\hat{R}_{h,\tau_{m}}(\lambda_{i}'(\tau_{m})f_{t};Y_{it})-\hat{R}_{h,\tau_{m}}(\lambda_{0,i}'(\tau_{m})f_{0,t};Y_{it}))/M$ are independent across $i$ and $t$. So, by equation \eqref{eq.obj.lip} and the Hoeffding's inequality, for each fixed $k$ and $c>0$, there exist constants $C_{3},C_{4}>0$ such that
\begin{equation}\label{eq.obj.hoeffding}
  \Pr\left(\left|\sqrt{NT}S(\theta_{k})\right|>c\right)\leq 2e^{-\frac{2c^{2}}{C_{3}\left[\frac{1}{MN}\sum_{m,i}\left\|\lambda_{i,k}(\tau_{m})-\lambda_{0,i}(\tau_{m})\right\|_{F}^{2}+\frac{1}{T}\sum_{t}\left\|f_{t,k}-f_{0,t}\right\|_{F}^{2}\right]}}\leq 2e^{-\frac{2c^{2}}{C_{4}}},
\end{equation}
where the last inequality is by the uniform boundedness of $\lambda_{0,i}(\tau_{m}),f_{0,t},\lambda_{i,k}(\tau_{m})$ and $f_{t,k}$. Therefore, by Lemma 2.2.1 in \cite{van1996weak}, for any fixed $k=1,\ldots,K$, there exists a constant $C_{5}$ which does not depend on $k$ such that
\begin{equation}\label{eq.op1.psi2}
  \left\|\hat{S}_{h}(\theta_{k})\right\|_{\psi_{2}}\leq \frac{C_{5}}{\sqrt{NT}}.
\end{equation}

Hence,
\begin{align}
\Pr\left(A_{2}>\frac{\varepsilon}{2}\right)\leq &\frac{2}{\varepsilon}\mathbb{E}\left(\max_{k=1,\ldots,K}\left|\hat{S}_{h}(\theta_{k})\right|\right)\notag\leq \frac{2}{\varepsilon}\left\|\max_{k=1,\ldots,K}\left|\hat{S}_{h}(\theta_{k})\right|\right\|_{\psi_{2}}\notag\\
\lesssim &\sqrt{\log K}\max_{k=1,\ldots,K}\left\|\hat{S}_{h}(\theta_{k})\right\|_{\psi_{2}}\lesssim \frac{\sqrt{r(MN+T)}}{\sqrt{NT}}\to 0,\label{eq.op1.A2}
\end{align}
where the first inequality is by Markov's inequality. The third inequality is by Lemma 2.2.2 in \cite{van1996weak}. The fourth inequality is by equation \eqref{eq.op1.psi2} and by $K=(C/\varepsilon_{0})^{(MN+T)r}$. 

Combining \eqref{eq.op1.A1A2}, \eqref{eq.op1.A1} and \eqref{eq.op1.A2} and letting $\varepsilon_{0}$ be such that $C_{2}\varepsilon_{0}<\varepsilon/2$, we obtain
\begin{align}
\Pr\left(\sup_{\theta\in \Theta}\left|\hat{S}_{h}(\theta)\right|>\varepsilon\right)\leq \Pr\left(A_{1}+A_{2}>\varepsilon\right)\leq \Pr\left(A_{2}>\frac{\varepsilon}{2}\right)\to 0.\label{eq.op1.obj.final}
\end{align}
This completes the proof.
\end{proof}

\begin{proof}[Proof of Lemma \ref{lem.metrics}]
The lemma is shown if for any $\delta>0$ and $\theta\in\Theta(\delta,d_{1})$, we find a constant $C>0$ that does not depend on $\theta$ such that $\theta\in\Theta(C\delta,d_{3})$. Let $(L(\tau_{m}))$ be formed by this arbitrary $\theta\in\Theta(\delta,d_{1})$. First, there exists a $C_{1}>0$ such that
\begin{align}
&\left\|\frac{\sum_{m=1}^{M}L'(\tau_{m})L(\tau_{m})}{MNT}-\frac{\sum_{m=1}^{M}L_{0}'(\tau_{m})L_{0}\left(\tau_{m}\right)}{MNT}\right\|_{F}^{2}\notag\\
\leq &\frac{M}{(MNT)^{2}}\sum_{m=1}^{M}\left\|L'(\tau_{m})L(\tau_{m})-L_{0}'(\tau_{m})L_{0}\left(\tau_{m}\right)\right\|_{F}^{2}\notag\\
\leq &\frac{2}{M(NT)^{2}}\sum_{m=1}^{M}\left\|L(\tau_{m})\right\|_{F}^{2}\left\|L(\tau_{m})-L_{0}\left(\tau_{m}\right)\right\|_{F}^{2}+\frac{2}{M(NT)^{2}}\sum_{m=1}^{M}\left\|L_{0}(\tau_{m})\right\|_{F}^{2}\left\|L(\tau_{m})-L_{0}\left(\tau_{m}\right)\right\|_{F}^{2}\notag\\
\leq & \frac{C_{1}}{MNT}\sum_{m=1}^{M}\left\|L(\tau_{m})-L_{0}\left(\tau_{m}\right)\right\|_{F}^{2}\leq C_{1}\delta^{2},\label{eq.roc1}
\end{align}
where the third inequality is by the boundedness of $\theta$ and $\theta_{0}$. 

\sloppy By construction, $F/\sqrt{T}$ and $F_{0}/\sqrt{T}$ are eigenvectors of $\sum_{m}L'(\tau_{m})L(\tau_{m})/MNT$ and $\sum_{m}L'_{0}(\tau_{m})L_{0}(\tau_{m})/MNT$, respectively. Meanwhile, all the eigenvalues of $\sum_{m}L'_{0}(\tau_{m})L_{0}(\tau_{m})/MNT$ are distinct and bounded away from zero for sufficiently large $N$ and $M$ by Lemma \ref{lem.discrete}. Therefore, Corollary 1 in \cite{yu2015useful} implies that the $j$-th ($j=1,\ldots,r$) columns in $F$ and $F_{0}$ satisfies:
\begin{align*}
  &\frac{1}{\sqrt{T}}\left\|F_{j}-\text{sgn}\left(F_{j}'F_{0,j}\right)F_{0,j}\right\|_{F}\\
  \leq &\frac{2^{3/2}}{\min\left(\sigma_{j-1}^{2}-\sigma_{j}^{2},\sigma_{j}^{2}-\sigma_{j+1}^{2}\right)}\left\|\frac{\sum_{m}L'(\tau_{m})L(\tau_{m})}{MNT}-\frac{\sum_{m}L'_{0}(\tau_{m})L_{0}(\tau_{m})}{MNT}\right\|\\
  \leq & \frac{2^{3/2}}{\min\left(\sigma_{j-1}^{2}-\sigma_{j}^{2},\sigma_{j}^{2}-\sigma_{j+1}^{2}\right)}\left\|\frac{\sum_{m}L'(\tau_{m})L(\tau_{m})}{MNT}-\frac{\sum_{m}L'_{0}(\tau_{m})L_{0}(\tau_{m})}{MNT}\right\|_{F},
\end{align*}
where $\sigma_{j}^{2}$ is the $j$-th largest eigenvalue of $\sum_{m}L'_{0}(\tau_{m})L_{0}(\tau_{m})/MNT$; $\sigma_{0}^{2}$ is defined as $\infty$. The second inequality is because the operator norm of a matrix $\|\cdot\|$ is no greater than the Frobenius norm. Let $H(\theta)\coloneqq\text{diag}(\text{sgn}(F_{j}'F_{0,j}))$. So, $H(\theta)\in \mathcal{H}$ for all $\theta$. There exists a $C_{2}>0$ such that
\begin{align}
\frac{1}{T}\|F-&F_{0}H(\theta)\|_{F}^{2}=\frac{1}{T}\sum_{j=1}^{r}\left\|F_{j}-\text{sgn}\left(F_{j}'F_{0,j}\right)F_{0,j}\right\|_{F}^{2}\notag\\
\leq&C_{2}\left\|\frac{\sum_{m}L'(\tau_{m})L(\tau_{m})}{MNT}-\frac{\sum_{m}L'_{0}(\tau_{m})L_{0}(\tau_{m})}{MNT}\right\|_{F}^{2}\leq C_{1}C_{2}\delta^{2},\label{eq.roc2}
\end{align}
where the last inequality is by \eqref{eq.roc1}.

Next, we turn to the loadings. By $(H(\theta))^{2}=I_{r}$,
\begin{align}
&\frac{1}{MNT}\sum_{m=1}^{M}\left\|L(\tau_{m})-L_{0}(\tau_{m})\right\|_{F}^{2}\notag\\
=&\frac{1}{MNT}\sum_{m=1}^{M}\left\|\Lambda(\tau_{m})F'-\Lambda_{0}(\tau_{m})H(\theta)F'+\Lambda_{0}(\tau_{m})H(\theta)F'-\Lambda_{0}(\tau_{m})H(\theta)H(\theta)F'_{0}\right\|_{F}^{2}\notag\\
\geq &\underbrace{\frac{1}{2MNT}\sum_{m=1}^{M}\left\|\Lambda(\tau_{m})F'-\Lambda_{0}(\tau_{m})H(\theta)F'\right\|_{F}^{2}}_{A}-\underbrace{\frac{1}{MNT}\sum_{m=1}^{M}\left\|\Lambda_{0}(\tau_{m})H(\theta)\left(F-F_{0}H(\theta)\right)'\right\|_{F}^{2}}_{B}.\label{eq.roc3}
\end{align}

For $A$, denoting the trace of a matrix by $\text{Tr}$,
\begin{align}
&\left\|\Lambda(\tau_{m})F'-\Lambda_{0}(\tau_{m})H(\theta)F'\right\|_{F}^{2}=\text{Tr}\left[\left(\Lambda(\tau_{m})F'-\Lambda_{0}(\tau_{m})H(\theta)F'\right)\left(\Lambda(\tau_{m})F'-\Lambda_{0}(\tau_{m})H(\theta)F'\right)'\right]\notag\\
&=T\cdot \text{Tr}\left[\left(\Lambda(\tau_{m})-\Lambda_{0}(\tau_{m})H(\theta)\right)\left(\Lambda(\tau_{m})-\Lambda_{0}(\tau_{m})H(\theta)\right)'\right]=T\left\|\Lambda(\tau_{m})-\Lambda_{0}(\tau_{m})H(\theta)\right\|_{F}^{2},\label{eq.roc4}
\end{align}
where the first and last equality are by the definition of Frobenius norm. The second equality is by $F'F=T\cdot I_{r}$.

For $B$, since $\left\|H(\theta)\right\|_{F}^{2}=r$ and $\max_{m}\|\Lambda_{0}(\tau_{m})\|_{F}^{2}$ is bounded by a constant times $N$, there exists $C_{3}>0$ such that
\begin{equation}\label{eq.roc5}
  \frac{1}{MNT}\sum_{m=1}^{M}\left\|\Lambda_{0}(\tau_{m})H(\theta)\left(F-F_{0}H(\theta)\right)'\right\|_{F}^{2}\leq \frac{C_{3}}{T}\left\|F-F_{0}H(\theta)\right\|_{F}^{2}.
\end{equation}
Substitute \eqref{eq.roc4} and \eqref{eq.roc5} into \eqref{eq.roc3}, and we have
\begin{align*}
\frac{1}{MN}\sum_{m=1}^{N}\left\|\Lambda(\tau_{m})-\Lambda_{0}(\tau_{m})H(\theta)\right\|_{F}^{2}\leq  \frac{2}{MNT}\sum_{m=1}^{M}\left\|L(\tau_{m})-L_{0}(\tau_{m})\right\|_{F}^{2}+\frac{2C_{3}}{T}\left\|F-F_{0}H(\theta)\right\|_{F}^{2}.\label{eq.roc6}
\end{align*}
Combining it with \eqref{eq.roc2}, let $C=\sqrt{2+2C_{1}C_{2}C_{3}}+\sqrt{C_{1}C_{2}}$, and we have
\begin{align*}
&\sqrt{\frac{1}{MN}\sum_{m=1}^{N}\left\|\Lambda(\tau_{m})-\Lambda_{0}(\tau_{m})H(\theta)\right\|_{F}^{2}}+\sqrt{\frac{1}{T}\left\|F-F_{0}H(\theta)\right\|_{F}^{2}}\leq C\delta.
\end{align*}
Therefore, $\theta\in\Theta(C\delta,d_{3})$, implying that $\Theta(\delta,d_{1})\subseteq \Theta(C\delta,d_{3})$.
\end{proof}
\begin{proof}[Proof of Lemma \ref{lem.psi2}]
By definition, for all $\theta_{1},\theta_{2}\in\Theta$ and all $H\in\mathcal{H}$, $d_{2}(\theta_{1},\theta_{2};H)\leq d_{3}(\theta_{1},\theta_{2};H)$. Let $\Theta_{2}(\delta,d_{2};H)\coloneqq\{\theta\in\Theta:d_{2}(\theta,
\theta_{0};H)\leq \delta\}$ for all $\delta>0$ and $H\in\mathcal{H}$. Together with Lemma \ref{lem.metrics}, we have
\begin{equation*}
  \Theta(\delta,d_{1})\subseteq\Theta(C\delta,d_{3})\subseteq\Theta(C\delta,d_{2})=\cup_{H\in\mathcal{H}}\Theta(C\delta,d_{2};H).
\end{equation*}
Therefore, 
\begin{align}
\mathbb{E}\left[\sup_{\theta\in\Theta(\delta,d_{1})}\left|\hat{S}_{h}(\theta)\right|\right]\leq \mathbb{E}\left[\sup_{\theta\in\Theta(C\delta,d_{2})}\left|\hat{S}_{h}(\theta)\right|\right] \leq r^{2}\max_{H\in\mathcal{H}}\mathbb{E}\left[\sup_{\theta\in\Theta(C\delta,d_{2};H)}\left|\hat{S}_{h}(\theta)\right|\right].\label{eq.space.inclusion}
\end{align}
Since $H$ only varies in the signs of the diagonal entries, it is without loss of generality to consider $\Theta(C\delta,d_{2};I_{r})$.

By equation \eqref{eq.obj.hoeffding} and similar to \eqref{eq.op1.psi2}, Lemma 2.2.1 of \cite{van1996weak} implies that
\begin{equation}\label{eq.psi2}
  \left\|\hat{S}_{h}(\theta)\right\|_{\psi_{2}}\lesssim \frac{d_{2}\left(\theta,\theta_{0};I_{r}\right)}{\sqrt{NT}}.
\end{equation}

Now following a similar argument as the proof of Lemma 3 in \cite{chen2021quantile}, the $\varepsilon$-packing number of $\Theta(C\delta,d_{2};I_{r})$ under $d_{2}(\cdot,\cdot;I_{r})$ is upper bounded by $(C_{1}\delta/\varepsilon)^{r(MN+T)}$ for some $C_{1}>0$. Hence, by the separability of $\hat{S}_{h}(\theta)$, equations \eqref{eq.space.inclusion} and \eqref{eq.psi2} and Theorem 2.2.4 in \cite{van1996weak} imply
\[
  \mathbb{E}\left(\sup_{\theta\in\Theta(\delta,d_{1})}\left|\hat{S}_{h}(\theta)\right|\right)\lesssim \left\|\sup_{\theta\in\Theta(C\delta,d_{2};I_{r})}\left|\hat{S}_{h}(\theta)\right|\right\|_{\psi_{2}}\lesssim\delta\cdot\frac{\sqrt{MN+T}}{\sqrt{NT}}\lesssim\delta\zeta_{NT}.
\]
\end{proof}

\begin{proof}[Proof of Lemma \ref{lem.uniformcons}]
\sloppy We only prove uniform consistency of $\hat{f}_{t}$ to save space; the proof of uniform consistency of $\hat{\lambda}_{i}(\tau_{m})$ follows a similar argument. For any $(\lambda_{1}(\tau_{1})',\ldots,\lambda_{N}(\tau_{1})',\ldots,\lambda_{1}(\tau_{M})',\ldots,\lambda_{N}(\tau_{M})')\in\mathcal{B}^{MNr}$ and $f\in\mathcal{B}^{r}$. Let 
\begin{align*}
\Delta_{\hat{R}_{h},t}(f;\Lambda(\cdot))\coloneqq& \frac{1}{MN}\sum_{m=1}^{M}\sum_{i=1}^{N}\hat{R}_{h,\tau_{m}}\left(\lambda_{i}'(\tau_{m})f;Y_{it}\right)-\frac{1}{MN}\sum_{m=1}^{M}\sum_{i=1}^{N}\hat{R}_{h,\tau_{m}}\left(\lambda_{i}'(\tau_{m})H_{NT,1}^{'}f_{0,t};Y_{it}\right),\\
\bar{\Delta}_{\hat{R}_{h},t}(f;\Lambda(\cdot))\coloneqq&\mathbb{E}\left(\Delta_{\hat{R}_{h},t}(f;\Lambda(\cdot))\right).
\end{align*}

Similar to equation \eqref{eq.expectation.lb}, since $\Delta_{\hat{R}_{h},t}(H_{NT,1}^{'}f_{0,t};\Lambda_{0}(\cdot)H_{NT}^{'-1})=0$ and
\begin{align*}
\frac{\underline{\textsf{f}}}{2}\leq \inf_{f\in\mathcal{B}^{r},m,i,t}\textsf{f}_{\tau_{m},it}\left(\lambda_{0,i}'(\tau_{m})H_{NT}^{'-1}f-\lambda_{0,i}'(\tau_{m})f_{0,t}\right)+O\left(h^{\gamma}\right)\leq \inf_{f\in\mathcal{B}^{r},m,i,t}\mathbb{E}\left(\hat{R}^{(2)}_{h,\tau_{m}}\left(\lambda_{0,i}'(\tau_{m})H_{NT}^{'-1}f;Y_{it}\right)\right)
\end{align*}
by Lemma \ref{lem.nonpara} where $\underline{\textsf{f}}$ is defined in the proof of Lemma \ref{lem.avgcons}, there exists a constant $C>0$ such that the following holds by expanding $\bar{\Delta}_{\hat{R}_{h},t}(\hat{f}_{t};\Lambda_{0}(\cdot)H_{NT}^{'-1})$ around $H'_{NT,1}f_{0,t}$:
\begin{align*}
  &\bar{\Delta}_{\hat{R}_{h},t}\left(\hat{f}_{t};\Lambda_{0}(\cdot)H_{NT}^{'-1}\right)\geq  \frac{1}{MN}\sum_{m,i}\mathbb{E}\left(\hat{R}^{(1)}_{h,\tau_{m},it}\right)\lambda_{0,i}'(\tau_{m})H_{NT,1}^{'-1}\left(\hat{f}_{t}-H'_{NT,1}f_{0,t}\right) \\
  &+\frac{\underline{\textsf{f}}}{4MN}\sum_{m,i}\left(\left(\hat{f}_{t}-H_{NT,1}'f_{0,t}\right)'H_{NT,1}^{-1}\lambda_{i}(\tau_{m})\right)^{2}\\
  = &\frac{\underline{\textsf{f}}}{4}\left(\hat{f}_{t}-H_{NT,1}'f_{0,t}\right)'H_{NT,1}^{-1}\left(\frac{1}{MN}\sum_{m,i}\lambda_{i}(\tau_{m})\lambda_{0,i}'(\tau_{m})\right)H_{NT,1}^{-1'}\left(\hat{f}_{t}-H_{NT,1}'f_{0,t}\right)+O\left(h^{\gamma}\right)\\
  \geq & \frac{\underline{\textsf{f}}\sigma_{r}^{2}}{4}\left\|\hat{f}_{t}-H_{NT,1}'f_{0,t}\right\|_{F}^{2}+O\left(h^{\gamma}\right),
\end{align*}
where the term $O\left(h^{\gamma}\right)$ is uniform in $t$. The last inequality holds because $H_{NT,1}^{-1}H_{NT,1}^{'-1}=I_{r}$ and $\sum_{m,i}\lambda_{i}(\tau_{m})\lambda_{0,i}'(\tau_{m})/MN$ is diagonal with the smallest diagonal entry $\sigma_{r}^{2}>0$ by Assumption \ref{assum.normalization} and Lemma \ref{lem.discrete}.
Therefore, 
\begin{equation}\label{eq.op1.ind.1}
  \max_{t=1,\ldots,T}\left\|\hat{f}_{t}-H_{NT,1}'f_{0,t}\right\|_{F}\leq \frac{4}{\underline{\textsf{f}}\sigma_{r}^{2}}\max_{t=1,\ldots,T} \bar{\Delta}_{\hat{R}_{h},t}\left(\hat{f}_{t};\Lambda_{0}(\cdot)H_{NT,1}^{'-1}\right)+O\left(h^{\gamma}\right).
\end{equation}

Meanwhile, by the definition of the estimator, for every $t$,
\begin{equation}\label{eq.op1.ind.2}
  \Delta_{\hat{R}_{h},t}\left(\hat{f}_{t};\hat{\Lambda}(\cdot)\right)\leq \Delta_{\hat{R}_{h},t}\left(H_{NT,1}'f_{0,t};\hat{\Lambda}(\cdot)\right)=0.
\end{equation}

Finally, by the mean value theorem,
\begin{align*}
  \Delta_{\hat{R}_{h},t}&\left(\hat{f}_{t};\Lambda_{0}(\cdot)H_{NT,1}^{'-1}\right)=\Delta_{\hat{R}_{h},t}\left(\hat{f}_{t};\hat{\Lambda}(\cdot)\right)\\
  &+\frac{1}{MNT}\sum_{m,i}\left[\hat{R}^{(1)}_{h,\tau_{m}}\left(\lambda_{i}^{*'}(\tau_{m})\hat{f}_{t};Y_{it}\right)-\hat{R}^{(1)}_{h,\tau_{m}}\left(\lambda_{i}^{*'}(\tau_{m})H_{NT,1}'f_{0,t};Y_{it}\right)\right]\left(\hat{\lambda}_{i}(\tau_{m})-\lambda_{0,i}(\tau_{m})\right)\\
  =&\Delta_{\hat{R}_{h},t}\left(\hat{f}_{t};\hat{\Lambda}(\cdot)\right)+O_{p}\left(\zeta_{NT}\right),
\end{align*}
where $\lambda_{i}^{*'}(\tau_{m})$ is the mean value. The term $O_{p}\left(\zeta_{NT}\right)$, which is uniform in $t$, is by the uniform boundedness of $R^{(1)}_{h,\tau_{m}}(\cdot;Y_{it})$ and by Theorem \ref{thm.roc}-(i). This implies that 
\begin{equation}\label{eq.op1.ind.3}
  \max_{t}\left|\Delta_{\hat{R}_{h},t}\left(\hat{f}_{t};\Lambda_{0}(\cdot)H_{NT,1}^{'-1}\right)-\Delta_{\hat{R}_{h},t}\left(\hat{f}_{t};\hat{\Lambda}(\cdot)\right)\right|=o_{p}(1).
\end{equation}

Equations \eqref{eq.op1.ind.1}, \eqref{eq.op1.ind.2} and \eqref{eq.op1.ind.3} imply that
\begin{align*}
 &\max_{t}\left\|\hat{f}_{t}-H_{NT,1}'f_{0,t}\right\|_{F}\\
 \lesssim&\max_{t}\left(\bar{\Delta}_{\hat{R}_{h},t}\left(\hat{f}_{t};\Lambda_{0}(\cdot)H_{NT,1}^{'-1}\right)-\Delta_{\hat{R}_{h},t}\left(\hat{f}_{t};\Lambda_{0}(\cdot)H_{NT,1}^{'-1}\right)+\Delta_{\hat{R}_{h},t}\left(\hat{f}_{t};\Lambda_{0}(\cdot)H_{NT,1}^{'-1}\right)\right)+o_{p}(1)\\
 =&\max_{t}\left(\bar{\Delta}_{\hat{R}_{h},t}\left(\hat{f}_{t};\Lambda_{0}(\cdot)H_{NT,1}^{'-1}\right)-\Delta_{\hat{R}_{h},t}\left(\hat{f}_{t};\Lambda_{0}(\cdot)H_{NT,1}^{'-1}\right)+\Delta_{\hat{R}_{h},t}\left(\hat{f}_{t};\hat{\Lambda}(\cdot)\right)\right)+o_{p}(1)\\
 \leq &\max_{t}\left(\bar{\Delta}_{\hat{R}_{h},t}\left(\hat{f}_{t};\Lambda_{0}(\cdot)H_{NT,1}^{'-1}\right)-\Delta_{\hat{R}_{h},t}\left(\hat{f}_{t};\Lambda_{0}(\cdot)H_{NT,1}^{'-1}\right)\right)+o_{p}(1)\\
 =&\max_{t}\sup_{f\in\mathcal{B}}\left|\Delta_{\hat{R}_{h},t}\left(f;\Lambda_{0}(\cdot)H_{NT,1}^{'-1}\right)-\bar{\Delta}_{\hat{R}_{h},t}\left(f;\Lambda_{0}(\cdot)H_{NT,1}^{'-1}\right)\right|+o_{p}(1)=o_{p}(1),
\end{align*}
where the last equality follows a similar argument as equation \eqref{eq.op1.obj.final}.
\end{proof}

\renewcommand{\thesection}{S.B}
\section{Proofs of Lemmas in Appendix \ref{appx.dist}}
\begin{proof}[Proof of Lemma \ref{lem.tilde.rate}]
We can prove these results by almost the same argument for Theorem \ref{thm.roc} except that we need to handle the estimated inverse density weights in the objective function. To save space, here we only present how we handle them by showing the counterpart for Lemma \ref{lem.avgcons}, i.e.,
\[
  \frac{1}{MNT}\left\|\tilde{L}(\tau_{m})-L_{0}(\tau_{m})\right\|_{F}^{2}=o_{p}\left(1\right).
\]
 We can use a similar argument to prove the counterparts of Lemmas \ref{lem.psi2}, \ref{lem.uniformcons} and thus Theorem \ref{thm.roc} since Lemma \ref{lem.metrics} holds regardless of the weights.

Let $\Omega^{(a,b)}$ ($a,b=1,2$) be the event that $\max_{m,i\in\mathcal{N}_{a},t\in\mathcal{T}_{b}}|\widehat{1/\textsf{f}_{\tau_{m},it}(0)}-1/\textsf{f}_{\tau_{m},it}(0)|\leq \log(NT)\psi_{NT}$. Theorem \ref{thm.firststage} implies that $\max_{a,b}|\Pr(\Omega^{(a,b)})-1|\to 0$.
Let $W^{(a,b)}=\{\widehat{1/\textsf{f}_{\tau_{m},it}(0)}:m=1,\ldots,M,i\in\mathcal{N}_{a},t\in\mathcal{T}_{b}\}$. Let 
\[\tilde{R}_{h}^{(a,b)}(\Lambda(\cdot),F)\coloneqq \frac{1}{MNT}\sum_{m=1}^{M}\sum_{i\in\mathcal{N}_{a}}\sum_{t\in\mathcal{T}_{b}}\widehat{\frac{1}{\textsf{f}_{\tau_{m},it}(0)}}\hat{R}_{h,\tau_{m}}(\lambda_{i}(\tau_{m})'f_{t};Y_{it}).\]
For any $\theta\in\Theta$, let $\Delta_{\tilde{R}_{h}}^{(a,b)}(\theta)\coloneqq \tilde{R}_{h}^{(a,b)}(\Lambda(\cdot),F)-\tilde{R}_{h}^{(a,b)}(\Lambda_{0}(\cdot),F_{0})$, 
$\bar{\Delta}^{(a,b)}_{\tilde{R}_{h}}(\theta)\coloneqq \mathbb{E}(\Delta_{\tilde{R}_{h}}^{(a,b)}(\theta)|W^{(a,b)},\Omega^{(a,b)})$,
and $\tilde{S}_{h}^{(a,b)}(\theta)\coloneqq \Delta_{\tilde{R}_{h}}^{(a,b)}(\theta)-\bar{\Delta}_{\tilde{R}_{h}}^{(a,b)}(\theta)$.

First, there exists a constant $C_{1}>0$ such that for any $(i,t)\in\mathcal{N}_{a}\times \mathcal{T}_{b}$ for all $a,b\in\{1,2\}$,
\begin{align}
  &\inf_{(\lambda',f')'\in\mathcal{B}^{2r},m,i\in\mathcal{N}_{a},t\in\mathcal{T}_{b}}\mathbb{E}\left(\widehat{\frac{1}{\textsf{f}_{\tau_{m},it}(0)}}\hat{R}^{(2)}_{h,\tau_{m}}(\lambda'f;Y_{it})\Bigg|W^{(a,b)},\Omega^{(a,b)}\right)\notag\\
   =&\inf_{(\lambda',f')'\in\mathcal{B}^{2r},m,i\in\mathcal{N}_{a},t\in\mathcal{T}_{b}}\left[\mathbb{E}\left(\widehat{\frac{1}{\textsf{f}_{\tau_{m},it}(0)}}\Bigg|W^{(a,b)},\Omega^{(a,b)}\right)\cdot \mathbb{E}\left(\hat{R}^{(2)}_{h,\tau_{m}}(\lambda'f;Y_{it})\right)\right]\notag\\
 \geq   &\inf_{(\lambda',f')'\in\mathcal{B}^{2r},m,i\in\mathcal{N}_{a},t\in\mathcal{T}_{b}}\mathbb{E}\left(\widehat{\frac{1}{\textsf{f}_{\tau_{m},it}(0)}}\Bigg|W^{(a,b)},\Omega^{(a,b)}\right)\cdot \inf_{(\lambda',f')'\in\mathcal{B}^{2r},m,i\in\mathcal{N}_{a},t\in\mathcal{T}_{b}}\mathbb{E}\left(\hat{R}^{(2)}_{h,\tau_{m}}(\lambda'f;Y_{it})\right)\notag\\
\geq &\left(\min_{m,i\in\mathcal{N}_{a},t\in\mathcal{T}_{b}}\frac{1}{\textsf{f}_{\tau_{m},it}(0)}-\log(NT)\psi_{NT}\right)\left(\underline{\textsf{f}}+O(h^{\gamma})\right) >C_{1},\label{eq.tilde.lowerbound1}
\end{align}
where the first equality is by $\{Y_{it}:i\in\mathcal{N}_{a},t\in\mathcal{T}_{b}\}\perp W^{(a,b)}$ and that the event $\Omega$ only depends on $W^{(a,b)}$. The inequality following it is because for sufficiently large $N$ and $T$, the two expectations are positive uniformly in $\lambda,f,m,i,t$ by Assumption \ref{assum.density}, Lemma \ref{lem.nonpara} and $\Omega^{(a,b)}$. Similarly,
\begin{align*}
  &\max_{m,i\in\mathcal{N}_{a},t\in\mathcal{T}_{b}}\left|\mathbb{E}\left(\widehat{\frac{1}{\textsf{f}_{\tau_{m},it}(0)}}\hat{R}^{(1)}_{h,\tau_{m},it}\Bigg|W^{(a,b)},\Omega^{(a,b)} \right)\right|\\
  \leq &\max_{m,i\in\mathcal{N}_{a},t\in\mathcal{T}_{b}}\left(\frac{1}{\textsf{f}_{\tau_{m},it}(0)}+\log(NT)\psi_{NT}\right)\cdot \max_{m,i\in\mathcal{N}_{a},t\in\mathcal{T}_{b}}\left|\mathbb{E}\left(\hat{R}^{(1)}_{h,\tau_{m},it}\right)\right|=O\left(h^{\gamma}\right).
\end{align*}

Therefore, by $\bar{\Delta}^{(a,b)}_{\tilde{R}_{h}}(\theta_{0})=0$, letting $\tilde{\theta}$ be the vector of all $\tilde{\lambda}_{i}(\tau_{m})$ and $\tilde{f}_{t}$, we have the following inequality by Taylor expansion similar to \eqref{eq.expectation.lb}, 
\begin{equation}
\bar{\Delta}^{(a,b)}_{\tilde{R}_{h}}\left(\tilde{\theta}\right)\geq \frac{C_{1}}{2MNT}\sum_{m=1}^{M}\sum_{i\in\mathcal{N}_{a}}\sum_{t\in\mathcal{T}_{b}}\left(\tilde{\lambda}_{i}'(\tau_{m})\tilde{f}_{t}-L_{0,it}(\tau_{m})\right)^{2}+O\left(h^{\gamma}\right).
\end{equation}

On the other hand,
\begin{align*}
\sum_{a=1}^{2}\sum_{b=1}^{2}\Delta^{(a,b)}_{\tilde{R}_{h}}\left(\tilde{\theta}\right)\leq \sum_{a=1}^{2}\sum_{b=1}^{2}\Delta^{(a,b)}_{\tilde{R}_{h}}\left(\theta_{0}\right)=0,
\end{align*}
where the inequality is by the definition of the estimator. Therefore,
\begin{align}
\frac{1}{MNT}&\left\|\tilde{L}(\tau_{m})-L_{0}(\tau_{m})\right\|_{F}^{2}=\frac{1}{MNT}\sum_{a=1}^{2}\sum_{b=1}^{2}\sum_{m=1}^{M}\sum_{i\in\mathcal{N}_{a}}\sum_{t\in\mathcal{T}_{b}}\left(\tilde{\lambda}_{i}'(\tau_{m})\tilde{f}_{t}-L_{0,it}(\tau_{m})\right)^{2}\notag\\
\leq&\frac{2}{C_{1}}\sum_{a=1}^{2}\sum_{b=1}^{2}\bar{\Delta}^{(a,b)}_{\tilde{R}_{h}}\left(\tilde{\theta}\right)+O\left(h^{\gamma}\right)\leq \frac{2}{C_{1}}\sum_{a=1}^{2}\sum_{b=1}^{2}\left(\bar{\Delta}^{(a,b)}_{\tilde{R}_{h}}\left(\tilde{\theta}\right)-\Delta^{(a,b)}_{\tilde{R}_{h}}\left(\tilde{\theta}\right)\right)+O\left(h^{\gamma}\right)\notag\\
\leq &\frac{2}{C_{1}}\sum_{a=1}^{2}\sum_{b=1}^{2}\sup_{\theta\in\Theta}\left|\tilde{S}^{(a,b)}_{h}\left(\theta\right) \right|+O\left(h^{\gamma}\right).\label{eq.tilde.bound}
\end{align}
Following the same argument for \eqref{eq.op1.obj.final}, by $\Omega^{(a,b)}$ and by independence between $\{Y_{it}:i\in\mathcal{N}_{a},t\in\mathcal{T}_{b}\}$ and $W^{(a,b)}$, we can show that for any $\varepsilon>0$,
\[
\Pr\left(\sup_{\theta\in\Theta}\left|\tilde{S}^{(a,b)}_{h}\left(\theta\right) \right|>\frac{\varepsilon}{4}\Bigg|W^{(a,b)},\Omega^{(a,b)} \right)=o(1),
\]
where $o(1)$ is uniform in the realization of $W^{(a,b)}$ under $\Omega^{(a,b)}$. Note that the $\psi_{2}$-norm involved in the proof needs to be defined in terms of expectation conditional on $(W^{(a,b)},\Omega^{(a,b)})$.

Therefore,
\begin{align*}
\Pr&\left(\sum_{a=1}^{2}\sum_{b=1}^{2}\sup_{\theta\in\Theta}\left|\tilde{S}^{(a,b)}_{h}\left(\theta\right) \right|>\varepsilon\right)\leq 4\max_{a,b}\Pr\left(\left|\tilde{S}^{(a,b)}_{h}\left(\theta\right) \right|>\frac{\varepsilon}{4}\right)\\
\leq &4\max_{a,b}\Pr\left(\left|\tilde{S}^{(a,b)}_{h}\left(\theta\right) \right|>\frac{\varepsilon}{4}\Bigg|\Omega^{(a,b)}\right)\Pr(\Omega^{(a,b)})+4\left(1-\min_{a,b}\Pr(\Omega^{(a,b)})\right)\\
=&4\max_{a,b}\mathbb{E}\left[\Pr\left(\left|\tilde{S}^{(a,b)}_{h}\left(\theta\right) \right|>\frac{\varepsilon}{4}\Bigg|W^{(a,b)},\Omega^{(a,b)}\right)\Bigg|\Omega^{(a,b)}\right]+4\left(1-\min_{a,b}\Pr(\Omega^{(a,b)})\right)=o(1),
\end{align*}
where the first equality is by the union bound. The third equality is by the law of iterated expectation.  Hence, equation \eqref{eq.tilde.bound} implies that
\[
  \frac{1}{MNT}\left\|\tilde{L}(\tau_{m})-L_{0}(\tau_{m})\right\|_{F}^{2}=o_{p}\left(1\right).
\]

We can similarly show the counterparts of Lemmas \ref{lem.psi2}\footnote{The expectation in Lemmas \ref{lem.psi2} needs to be changed to \[\mathbb{E}\left[\sup_{\theta\in\Theta(\delta,d_{1})}\left|\tilde{S}^{(a,b)}_{h}(\theta)\right|\Bigg|W^{(a,b)},\Omega^{(a,b)}\right].\] }, \ref{lem.uniformcons} and thus Theorem \ref{thm.roc} using the same argument.
\end{proof}

\begin{proof}[Proof of Lemma \ref{lem.inv}]
For simplicity, define $\tilde{f}_{0,t}\coloneqq \tilde{H}_{NT,1}'f_{0,t}$ and $\tilde{F}_{0}=F_{0}\tilde{H}_{NT,1}$. By orthogonality of $\tilde{H}_{NT,1}$ and Assumption \ref{assum.normalization}, we have $\tilde{H}_{NT,1}'F'_{0}F_{0}\tilde{H}_{NT,1}/T=I_{r}$. So, for all $j,j'=1,\ldots,r$ such that $j\neq j'$,
\begin{align*}
  \frac{1}{T}\sum_{t=1}^{T}\tilde{f}_{0,tj}^{2}=1,\ \ \ \ \frac{1}{T}\sum_{t=1}^{T}\tilde{f}_{0,tj}\tilde{f}_{0,tj'}=0.
\end{align*}
Similarly, by normalization, we have the following for the estimator:
\begin{align*}
  \frac{1}{T}\sum_{t=1}^{T}\tilde{f}_{tj}^{2}=1,\ \ \ \ \frac{1}{T}\sum_{t=1}^{T}\tilde{f}_{tj}\tilde{f}_{tj'}=0.
\end{align*}

Therefore,
\begin{align}
\frac{1}{T}\sum_{t=1}^{T}\tilde{f}_{tj}\tilde{f}_{0,tj}=&1-\frac{1}{2T}\sum_{t=1}^{T}\left(\tilde{f}_{tj}-\tilde{f}_{0,tj}\right)^{2}=1+O_{p}\left(\zeta_{NT}^{2}\right),\label{eq.rotation.same}
\end{align}
where $O_{p}(\zeta_{NT}^{2})$ is by Lemma \ref{lem.tilde.rate} and is uniform in $j=1,\ldots,r$. By $\tilde{f}_{0,tj}=f_{0,tj}\tilde{H}_{NT,1}(j,j)$ where $\tilde{H}_{NT,1}(j,j)$ is the $j$-th diagonal entry in $\tilde{H}_{NT,1}$ satisfying $\tilde{H}_{NT,1}(j,j)=1/\tilde{H}_{NT,1}(j,j)$, 
\begin{equation*}
  \frac{1}{T}\sum_{t=1}^{T}\tilde{f}_{tj}f_{0,tj}=\tilde{H}_{NT,1}(j,j)+O_{p}\left(\zeta_{NT}^{2}\right).
\end{equation*}
Equation \eqref{eq.H2inv.diag} is thus shown since the above equation holds for all $j$.

Similarly,
\begin{align}
&\left|\frac{1}{T}\sum_{t=1}^{T}\tilde{f}_{tj}\tilde{f}_{0,tj'}+\frac{1}{T}\sum_{t=1}^{T}\tilde{f}_{tj'}\tilde{f}_{0,tj}\right|\notag\\
=&\left|\frac{1}{T}\sum_{t=1}^{T}\left(\tilde{f}_{tj}-\tilde{f}_{0,tj}\right)\left(\tilde{f}_{tj'}-\tilde{f}_{0,tj'}\right)\right|\lesssim\frac{1}{T}\left\|\tilde{F}-\tilde{F}_{0}\right\|_{F}^{2}=O_{p}\left(\zeta_{NT}^{2}\right),\label{eq.rotation.diff1}
\end{align}

Meanwhile,
\begin{equation}\label{eq.rotation.diff2}
  \frac{1}{T}\sum_{t=1}^{T}\tilde{f}_{tj}\tilde{f}_{0,tj'}=\frac{1}{T}\sum_{t=1}^{T}\left(\tilde{f}_{tj}-\tilde{f}_{0,tj}\right)\tilde{f}_{0,tj'}=O_{p}\left(\zeta_{NT}\right).
\end{equation}
Combining equations \eqref{eq.rotation.same} and \eqref{eq.rotation.diff2} yields
\begin{align*}
H_{NT,2}'\tilde{H}_{NT,1}-I_{r}=&\frac{1}{T}\sum_{t=1}^{T}\tilde{f}_{t}f_{0,t}'\tilde{H}_{NT,1}-I_{r}=\frac{1}{T}\sum_{t=1}^{T}\tilde{f}_{t}\tilde{f}_{0,t}'-I_{r}=O_{p}\left(\zeta_{NT}\right).
\end{align*}
Therefore, by orthogonality of $\tilde{H}_{NT,1}$,
\begin{equation*}
H_{NT,2}'=\tilde{H}_{NT,1}^{-1}+O_{p}(\zeta_{NT})=\tilde{H}_{NT,1}'+O_{p}(\zeta_{NT}).
\end{equation*}
Equation \eqref{eq.H2inv.0} is proved.

Now we prove equation \eqref{eq.H2inv.1}. Define $\bar{H}_{NT,2}\coloneqq \tilde{H}_{NT,1}'H_{NT,2}H_{NT,2}'\tilde{H}_{NT,1}$, that is,
\begin{align*}
  \bar{H}_{NT,2}=&\tilde{H}_{NT,1}'\left(\frac{F_{0}'\tilde{F}}{T}\right)\left(\frac{\tilde{F}'F_{0}}{T}\right)\tilde{H}_{NT,1}=\left(\frac{\tilde{F}_{0}'\tilde{F}}{T}\right)\left(\frac{\tilde{F}'\tilde{F}_{0}}{T}\right).
\end{align*}
We first show that $\bar{H}_{NT,2}=I_{r}+O_{p}(\zeta_{NT}^{2})$. The desired result will then follow from orthogonality of $\tilde{H}_{NT,1}$.

The $(j,k)$-th entry of $\bar{H}_{NT,2}$, denoted by $\bar{H}_{NT,2}(j,k)$, satisfies the following:
\begin{align*}
\bar{H}_{NT,2}\left(j,k\right)=&\sum_{s=1}^{r}\left(\frac{1}{T}\sum_{t=1}^{T}\tilde{f}_{0,tj}\tilde{f}_{ts}\right)\left(\frac{1}{T}\sum_{t=1}^{T}\tilde{f}_{ts}\tilde{f}_{0,tk}\right).
\end{align*}

If $j=k$,
\begin{align*}
\sum_{s=1}^{r}\left(\frac{1}{T}\sum_{t=1}^{T}\tilde{f}_{0,tj}\tilde{f}_{ts}\right)\left(\frac{1}{T}\sum_{t=1}^{T}\tilde{f}_{ts}\tilde{f}_{0,tk}\right)=&\left(\frac{1}{T}\sum_{t=1}^{T}\tilde{f}_{0,tj}\tilde{f}_{tj}\right)^{2}+\sum_{s\neq j}\left(\frac{1}{T}\sum_{t=1}^{T}\tilde{f}_{0,tj}\tilde{f}_{ts}\right)\left(\frac{1}{T}\sum_{t=1}^{T}\tilde{f}_{0,tj}\tilde{f}_{ts}\right)\\
=&1+O_{p}\left(\zeta_{NT}^{2}\right),
\end{align*}
where the last equality follows equations \eqref{eq.rotation.same} and \eqref{eq.rotation.diff2}.

If $j\neq k$,
\begin{align*}
&\sum_{s=1}^{r}\left(\frac{1}{T}\sum_{t=1}^{T}\tilde{f}_{0,tj}\tilde{f}_{ts}\right)\left(\frac{1}{T}\sum_{t=1}^{T}\tilde{f}_{ts}\tilde{f}_{0,tk}\right)\\
=&\left(\frac{1}{T}\sum_{t=1}^{T}\tilde{f}_{0,tj}\tilde{f}_{tj}\right)\left(\frac{1}{T}\sum_{t=1}^{T}\tilde{f}_{tj}\tilde{f}_{0,tk}\right)+
\left(\frac{1}{T}\sum_{t=1}^{T}\tilde{f}_{0,tj}\tilde{f}_{tk}\right)\left(\frac{1}{T}\sum_{t=1}^{T}\tilde{f}_{tk}\tilde{f}_{0,tk}\right)\\
&+\sum_{s\neq j,k}\left(\frac{1}{T}\sum_{t=1}^{T}\tilde{f}_{0,tj}\tilde{f}_{ts}\right)\left(\frac{1}{T}\sum_{t=1}^{T}\tilde{f}_{ts}\tilde{f}_{0,tk}\right)\\
=&\left(1+O_{p}\left(\zeta_{NT}^{2}\right)\right)\left(\frac{1}{T}\sum_{t=1}^{T}\tilde{f}_{tj}\tilde{f}_{0,tk}\right)+
\left(\frac{1}{T}\sum_{t=1}^{T}\tilde{f}_{0,tj}\tilde{f}_{tk}\right)\left(1+O_{p}\left(\zeta_{NT}^{2}\right)\right)+O_{p}\left(\zeta_{NT}^{2}\right)\\
=&\frac{1}{T}\sum_{t=1}^{T}\tilde{f}_{tj}\tilde{f}_{0,tk}+
\frac{1}{T}\sum_{t=1}^{T}\tilde{f}_{0,tj}\tilde{f}_{tk}+O_{p}\left(\zeta_{NT}^{2}\right)=O_{p}\left(\zeta_{NT}^{2}\right),
\end{align*}
where the second equality follows equations \eqref{eq.rotation.same} and \eqref{eq.rotation.diff2}, and the last equality is by \eqref{eq.rotation.diff1}. Therefore,
\begin{equation*}
  \bar{H}_{NT,2}=I_{r}+O_{p}\left(\zeta_{NT}^{2}\right).
\end{equation*}
Hence,
\begin{align*}
  H_{NT,2}H_{NT,2}'=&\left(\tilde{H}_{NT,1}'\right)^{-1}\bar{H}_{NT,2}\tilde{H}_{NT,1}^{-1}=\tilde{H}_{NT,1}\bar{H}_{NT,2}\tilde{H}_{NT,1}'\\
=&\tilde{H}_{NT,1}\tilde{H}_{NT,1}'+O_{p}\left(\zeta_{NT}^{2}\right)=I_{r}+O_{p}\left(\zeta_{NT}^{2}\right),
\end{align*}
where the second and last equality are by orthogonality of $\tilde{H}_{NT,1}$. Equation \eqref{eq.H2inv.1} is proved.

Finally, to show \eqref{eq.H2inv.2}, left and right multiplying the two sides of \eqref{eq.H2inv.1} by $H_{NT,2}'$ and $H_{NT,2}$ leads to
\begin{align}
H_{NT,2}'H_{NT,2}H_{NT,2}'H_{NT,2}=H_{NT,2}'H_{NT,2}+O_{p}\left(\zeta_{NT}^{2}\right)\label{eq.rotation.finaltransform}
\end{align}
By \eqref{eq.H2inv.0} and by $\tilde{H}_{NT,1}'\tilde{H}_{NT,1}=I_{r}$, we have the preliminary rate $(H_{NT,2}'H_{NT,2}-I_{r})=O_{p}(\zeta_{NT})$. Substitute it into \eqref{eq.rotation.finaltransform} and rearrange the terms:
\begin{align*}
\left(I_{r}+O_{p}(\zeta_{NT})\right)\left(H_{NT,2}'H_{NT,2}-I_{r}\right)=O_{p}\left(\zeta_{NT}^{2}\right),
\end{align*}
which implies \eqref{eq.H2inv.2}.
\end{proof}

\begin{proof}[Proof of Lemma \ref{lem.expansion}] 
First, note that by \eqref{eq.H2inv.0} in Lemma \ref{lem.inv} and by the boundedness of $f_{0,t}$ and $\lambda_{0,i}(\tau_{m})$, the results in Lemma \ref{lem.tilde.rate} hold by replacing $\tilde{H}_{NT,1}$ with $H_{NT,2}$.
Now consider the first order conditions with respect to $\tilde{\lambda}_{i}(\tau_{m})$ and $\tilde{f}_{t}$:
\begin{align}
\frac{1}{T}\sum_{s=1}^{T}\widehat{\frac{1}{\textsf{f}_{\tau_{m},is}(0)}}\hat{R}^{(1)}_{h,\tau_{m}}\left(\tilde{\lambda}'_{i}(\tau_{m})\tilde{f}_{s};Y_{is}\right)\tilde{f}_{s}&=0,\label{eq.foc.lambda}
\\
  \frac{1}{MN}\sum_{i=1}^{n}\sum_{m=1}^{M}\widehat{\frac{1}{\textsf{f}_{\tau_{m},it}(0)}}\hat{R}^{(1)}_{h,\tau_{m}}\left(\tilde{\lambda}'_{i}(\tau_{m})\tilde{f}_{t};Y_{it}\right)\tilde{\lambda}_{i}(\tau_{m})&=0.\label{eq.foc.f}
\end{align}
We first Taylor expand equation \eqref{eq.foc.lambda} around $H_{NT,2}'f_{0,t}$ and $H_{NT,2}^{-1}\lambda_{0,i}(\tau_{m})$ to the second order. For simplicity, denote the rotated true parameters $F_{0}H_{NT,2}$, $\Lambda_{0}(\tau_{m})H_{NT,2}^{'-1},H_{NT,2}'f_{0,t}$ and $H_{NT,2}^{-1}\lambda_{0,i}(\tau_{m})$ by $F_{0}^{H}$, $\Lambda^{H}_{0}(\tau_{m})$, $f_{0,t}^{H}$ and $\lambda^{H}_{0,i}(\tau_{m})$, respectively.
\begin{align*}
0=&\frac{1}{T}\sum_{s=1}^{T}\widehat{\frac{1}{\textsf{f}_{\tau_{m},is}(0)}}\hat{R}^{(1)}_{h,\tau_{m},is}\ \cdot f^{H}_{0,s}+\frac{1}{T}\sum_{s=1}^{T}\widehat{\frac{1}{\textsf{f}_{\tau_{m},is}(0)}}\hat{R}^{(1)}_{h,\tau_{m},is}\cdot\left(\tilde{f}_{s}-f^{H}_{0,s}\right)\\
&+\frac{1}{T}\sum_{s=1}^{T}\widehat{\frac{1}{\textsf{f}_{\tau_{m},is}(0)}}\hat{R}^{(2)}_{h,\tau_{m},is}\cdot f^{H}_{0,s}\lambda^{H'}_{0,i}(\tau_{m})\left(\tilde{f}_{s}-f^{H}_{0,s}\right)\\
&+\frac{1}{T}\sum_{s=1}^{T}\widehat{\frac{1}{\textsf{f}_{\tau_{m},is}(0)}}\hat{R}^{(2)}_{h,\tau_{m},is}\cdot f^{H}_{0,s}f^{H'}_{0,s}\left(\tilde{\lambda}_{i}(\tau_{m})-\lambda^{H}_{0,i}(\tau_{m})\right)\\
&+O_{p}\left(\max_{m,i,t}\left|\widehat{\frac{1}{\textsf{f}_{\tau_{m},it}(0)}}\right|\right) O_{p}\left(\frac{1}{Th^{2}}\sum_{t=1}^{T}\left\|\tilde{F}-F^{H}_{0}\right\|_{F}^{2}\right)\\
&+O_{p}\left(\max_{m,i,t}\left|\widehat{\frac{1}{\textsf{f}_{\tau_{m},it}(0)}}\right|\right)O_{p}\left(\frac{1}{h^{2}}\left\|\tilde{\lambda}_{i}(\tau_{m})-\lambda^{H}_{0,i}(\tau_{m})\right\|_{F}^{2}\right)\\
&+O_{p}\left(\max_{m,i,t}\left|\widehat{\frac{1}{\textsf{f}_{\tau_{m},it}(0)}}\right|\right)O_{p}\left(\frac{1}{h^{2}}\sqrt{\frac{1}{T}\sum_{s=1}^{T}\left\|\tilde{F}-F^{H}_{0}\right\|_{F}^{2}}\left\|\tilde{\lambda}_{i}(\tau_{m})-\lambda^{H}_{0,i}(\tau_{m})\right\|_{F}\right),
\end{align*}
where the last three terms are due to the second order Taylor expansion.  By Lemma \ref{lem.tilde.rate} with $\tilde{H}_{NT,1}$ replaced by $H_{NT,2}$, by Theorem \ref{thm.firststage} and by Assumptions \ref{assum.density} and \ref{assum.kernel}, the last three terms are $o_{p}(h^{2}/\sqrt{T})$ uniformly in $m,i,t$.

We next consider the first three terms on the right-hand side. We first show that replacing $\widehat{1/\textsf{f}_{\tau_{m},is}(0)}$ and $\hat{R}^{(1)}_{h,\tau_{m},is}$ with $1/\textsf{f}_{\tau_{m},is}(0)$ and $\eta_{h,\tau_{m},is}$ respectively in these terms only causes a difference of $o_{p}(h^{2}/\sqrt{T})$.

For the first term,
\begin{align*}
&\max_{m,i}\left\|\frac{1}{T}\sum_{s=1}^{T}\widehat{\frac{1}{\textsf{f}_{\tau_{m},is}(0)}}\hat{R}^{(1)}_{h,\tau_{m},is}\cdot f^{H}_{0,s}-\frac{1}{T}\sum_{s=1}^{T}\frac{1}{\textsf{f}_{\tau_{m},is}(0)}\eta_{h,\tau_{m},is}\cdot f^{H}_{0,s}\right\|_{F}\\{}
\leq&\max_{m,i}\left\|\frac{1}{T}\sum_{s=1}^{T}\widehat{\frac{1}{\textsf{f}_{\tau_{m},is}(0)}}\eta_{h,\tau_{m},is} \cdot f^{H}_{0,s}-\frac{1}{T}\sum_{s=1}^{T}\frac{1}{\textsf{f}_{\tau_{m},is}(0)}\eta_{h,\tau_{m},is} \cdot f^{H}_{0,s}\right\|_{F}+O_{p}\left(h^{\gamma}\right)\\
\leq &\sum_{a=1}^{2}\sum_{b=1}^{2}\max_{m,i\in\mathcal{N}_{a}}\left\|\frac{1}{T}\sum_{s\in\mathcal{T}_{b}}\left(\widehat{\frac{1}{\textsf{f}_{\tau_{m},is}(0)}}-\frac{1}{\textsf{f}_{\tau_{m},is}(0)}\right)\eta_{h,\tau_{m},is} \cdot f_{0,s}\right\|_{F}\cdot\left\|H_{NT,2}\right\|_{F} +O_{p}\left(h^{\gamma}\right)\\
=&O_{p}\left(\frac{\sqrt{\log\left(MN\right)}\psi_{NT}}{\sqrt{T}}\right)+O_{p}\left(h^{\gamma}\right)=o_{p}\left(\frac{h^{2}}{\sqrt{T}}\right),
\end{align*}
where the first inequality is by Lemma \ref{lem.nonpara} and by the uniform boundedness of $\widehat{1/\textsf{f}_{\tau_{m},is}(0)}$ with probability approaching 1 by Theorem \ref{thm.firststage} and Assumption \ref{assum.density}. The last equality is by Assumption \ref{assum.kernel} and by the requirements on $h_{d}$ in Theorem \ref{thm.dist2}. To see the penultimate equality, first note that $\|H_{NT,2}\|_{F}$ is bounded with probability 1. Denote the event that $\max_{m,i\in\mathcal{N}_{a},t\in\mathcal{T}_{b}}|\widehat{1/\textsf{f}_{\tau_{m},it}(0)}-1/\textsf{f}_{\tau_{m},it}(0)|\leq C_{1}\psi_{NT}$ by $\Omega^{(a,b)}$ and $W^{(a,b)}=\{\widehat{1/\textsf{f}_{\tau_{m},it}(0)}:m=1,\ldots,M,i\in\mathcal{N}_{a},t\in\mathcal{T}_{b}\}$. For any fixed $\varepsilon>0$, there exist $C_{1},C_{2}>0$ such that for sufficiently large $N$ and $T$, we have
\begin{align}
&\Pr\left(\max_{m,i\in\mathcal{N}_{a}}\left\|\frac{1}{T}\sum_{s\in\mathcal{T}_{b}}\left(\widehat{\frac{1}{\textsf{f}_{\tau_{m},is}(0)}}-\frac{1}{\textsf{f}_{\tau_{m},is}(0)}\right)\eta_{h,\tau_{m},is} \cdot f_{0,s}\right\|\geq \frac{C_{2}\sqrt{\log\left(MN\right)}\psi_{NT}}{\sqrt{T}} \right)\notag\\
\leq &\Pr\left(\max_{m,i\in\mathcal{N}_{a}}\left\|\frac{1}{T}\sum_{s\in\mathcal{T}_{b}}\left(\widehat{\frac{1}{\textsf{f}_{\tau_{m},is}(0)}}-\frac{1}{\textsf{f}_{\tau_{m},is}(0)}\right)\eta_{h,\tau_{m},is} \cdot f_{0,s}\right\|\geq \frac{C_{2}\sqrt{\log\left(MN\right)}\psi_{NT}}{\sqrt{T}}\Bigg|\Omega^{(a,b)}  \right)\Pr\left(\Omega^{(a,b)}\right)\notag\\
&+\left(1-\Pr(\Omega^{(a,b)})\right)\notag\\
\leq& MN\max_{m,i\in\mathcal{N}_{a}}\mathbb{E}\Bigg[\Pr\Bigg(\left\|\frac{1}{T}\sum_{s\in\mathcal{T}_{b}}\left(\widehat{\frac{1}{\textsf{f}_{\tau_{m},is}(0)}}-\frac{1}{\textsf{f}_{\tau_{m},is}(0)}\right)\eta_{h,\tau_{m},is} \cdot f_{0,s}\right\|\notag\\
&\ \ \ \ \ \ \ \ \ \ \ \ \ \ \ \ \ \ \ \ \ \ \ \ \ \ \ \ \geq \frac{C_{2}\sqrt{\log\left(MN\right)}\psi_{NT}}{\sqrt{T}}\Bigg|\Omega^{(a,b)},W^{(a,b)}  \Bigg)\Bigg|\Omega^{(a,b)} \Bigg]+\left(1-\Pr(\Omega^{(a,b)})\right)\notag\\
\leq &\varepsilon,\label{eq.subsample}
\end{align}
where the second inequality is by the law of iterated expectation. The last inequality is by Theorem \ref{thm.firststage} and by the Hoeffding's inequality in view that the random variables $\eta_{h,\tau_{m},is}$ are bounded and independent across $s$ for all $(i,s)\in\mathcal{N}_{a}\times \mathcal{T}_{b}$ conditional on $(W^{(a,b)},\Omega^{(a,b)})$, implied by Assumption \ref{assum.iid} and independence between $\{\eta_{h,\tau_{m},is}:i\in\mathcal{N}_{a},s\in\mathcal{T}_{b}\}$ and $W^{(a,b)}$. 

For the second and third terms, the differences caused by replacing $\widehat{1/\textsf{f}_{\tau_{m},is}(0)}$ and $\hat{R}^{(1)}_{h,\tau_{m},is}$ with $1/\textsf{f}_{\tau_{m},is}(0)$ and $\eta_{h,\tau_{m},is}$ respectively  are $O_{p}\left(\zeta_{NT}\psi_{NT}/h\right)+O_{p}(h^{\gamma}\zeta_{NT})=o_{p}(h^{2}/\sqrt{T})$ uniformly in $m$ and $i$ by the average rate of convergence of $\tilde{F}$ in Lemma \ref{lem.tilde.rate} by replacing $\tilde{H}_{NT,1}$ with $H_{NT,2}$, by the boundedness of the kernel function $k(\cdot)$, by Theorem \ref{thm.firststage} and by the requirements on $h$ and $h_{d}$ in Assumption \ref{assum.kernel} and Theorem \ref{thm.dist2}.

Now we consider the fourth term.
\begin{align*}
&\frac{1}{T}\sum_{s=1}^{T}\widehat{\frac{1}{\textsf{f}_{\tau_{m},is}(0)}}\hat{R}^{(2)}_{h,\tau_{m},is}\cdot f^{H}_{0,s}f^{H'}_{0,s}\left(\tilde{\lambda}_{i}(\tau_{m})-\lambda^{H}_{0,i}(\tau_{m})\right)\\
=&\frac{1}{T}\sum_{s=1}^{T}\frac{1}{\textsf{f}_{\tau_{m},is}(0)}\hat{R}^{(2)}_{h,\tau_{m},is}\cdot f^{H}_{0,s}f^{H'}_{0,s}\left(\tilde{\lambda}_{i}(\tau_{m})-\lambda^{H}_{0,i}(\tau_{m})\right)+O_{p}\left(\frac{\psi_{NT}}{h}\right)\left(\tilde{\lambda}_{i}(\tau_{m})-\lambda^{H}_{0,i}(\tau_{m})\right)\\
=&\frac{1}{T}\sum_{s=1}^{T}\frac{1}{\textsf{f}_{\tau_{m},is}(0)}\left(\hat{R}^{(2)}_{h,\tau_{m},is}-\mathbb{E}\left(\hat{R}^{(2)}_{h,\tau_{m},is}\right) \right)\cdot f^{H}_{0,s}f^{H'}_{0,s}\left(\tilde{\lambda}_{i}(\tau_{m})-\lambda^{H}_{0,i}(\tau_{m})\right)\\
&+\frac{1}{T}\sum_{s=1}^{T}\frac{\textsf{f}_{\tau_{m},is}(0)}{\textsf{f}_{\tau_{m},is}(0)}f^{H}_{0,s}f^{H'}_{0,s}\left(\tilde{\lambda}_{i}(\tau_{m})-\lambda^{H}_{0,i}(\tau_{m})\right)\\
&+\left(O\left(h^{\gamma}\right) +O_{p}\left(\frac{\psi_{NT}}{h}\right)\right)\left(\tilde{\lambda}_{i}(\tau_{m})-\lambda^{H}_{0,i}(\tau_{m})\right)\\
=&\left(1+O_{p}\left(\zeta_{NT}^{2}\right)+O_{p}\left(\frac{\sqrt{\log(MN)}}{\sqrt{Th}}\right)+O\left(h^{\gamma}\right) +O_{p}\left(\frac{\psi_{NT}}{h}\right)\right)\left(\tilde{\lambda}_{i}(\tau_{m})-\lambda^{H}_{0,i}(\tau_{m})\right)\\
=&\tilde{\lambda}_{i}(\tau_{m})-\lambda^{H}_{0,i}(\tau_{m})+o_{p}\left(\frac{h^{2}}{\sqrt{T}}\right),
\end{align*}
where the first equality is by Theorem \ref{thm.firststage} and by the boundedness of $h\hat{R}^{(2)}_{h,\tau_{m},is}$. The second equality is by Lemma \ref{lem.nonpara}. The penultimate equality is by $\sum_{s=1}^{T}f^{H}_{0,s}f^{H'}_{0,s}/T=H_{NT,2}'H_{NT,2}=I_{r}+O_{p}(\zeta_{NT}^{2})$ by Lemma \ref{lem.inv} and by the Bernstein's inequality. The last equality is by Assumption \ref{assum.kernel}, by Lemma \ref{lem.tilde.rate} with $\tilde{H}_{NT,1}$ replaced with $H_{NT,2}$ and by the requirements on $h$ and $h_{d}$; the term $o_{p}(h^{2}/\sqrt{T})$ is uniform in $m,i,t$.

Combining these results, we have
\begin{align}
\tilde{\lambda}_{i}(\tau_{m})-\lambda^{H}_{0,i}(\tau_{m})=&-\frac{1}{T}\sum_{s=1}^{T}\frac{1}{\textsf{f}_{\tau_{m},is}(0)}\eta_{h,\tau_{m},is} \cdot f^{H}_{0,s}-\frac{1}{T}\sum_{s=1}^{T}\frac{1}{\textsf{f}_{\tau_{m},is}(0)}\eta_{h,\tau_{m},is}\cdot\left(\tilde{f}_{s}-f^{H}_{0,s}\right)\notag\\
&-\frac{1}{T}\sum_{s=1}^{T}\frac{1}{\textsf{f}_{\tau_{m},is}(0)}\hat{R}^{(2)}_{h,\tau_{m},is}\cdot f^{H}_{0,s}\lambda^{H'}_{0,i}(\tau_{m})\left(\tilde{f}_{s}-f^{H}_{0,s}\right)+o_{p}\left(\frac{h^{2}}{\sqrt{T}}\right),\label{eq.expansion.lambda}
\end{align}
where the term $o_{p}(h^{2}/\sqrt{T})$ is uniform in $m,i,t$.

Similarly, we can expand \eqref{eq.foc.f} and obtain
\begin{align}
&\left(\frac{1}{MN}\sum_{m=1}^{M}\Lambda_{0}^{H'}(\tau_{m})\Lambda_{0}^{H}(\tau_{m})\right)\left(\tilde{f}_{t}-f_{0,t}^{H}\right)\notag\\
=&-\frac{1}{MN}\sum_{i=1}^{N}\sum_{m=1}^{M}\frac{1}{\textsf{f}_{\tau_{m},it}(0)}\eta_{h,\tau_{m},it} \cdot \lambda^{H}_{0,i}(\tau_{m})\notag\\
&-\underbrace{\frac{1}{MN}\sum_{i=1}^{N}\sum_{m=1}^{M}\frac{1}{\textsf{f}_{\tau_{m},it}(0)}\eta_{h,\tau_{m},it}\cdot\left(\tilde{\lambda}_{i}(\tau_{m})-\lambda^{H}_{0,i}(\tau_{m})\right)}_{A_{1t}}\notag\\
&-\underbrace{\frac{1}{MN}\sum_{i=1}^{N}\sum_{m=1}^{M}\frac{1}{\textsf{f}_{\tau_{m},it}(0)}\hat{R}^{(2)}_{h,\tau_{m},it}\cdot \lambda^{H}_{0,i}(\tau_{m})f^{H'}_{0,t}\left(\tilde{\lambda}_{i}(\tau_{m})-\lambda^{H}_{0,i}(\tau_{m})\right)}_{A_{2t}}+o_{p}\left(\frac{h^{2}}{\sqrt{N}}\right),\label{eq.expansion.f}
\end{align}
where the term $o_{p}(h^{2}/\sqrt{N})$ is uniform in $m,i,t$. We now substitute equation \eqref{eq.expansion.lambda} into $A_{1t}$ and $A_{2t}$.

For $A_{1t}$,
\begin{align*}
A_{1t}=&-\underbrace{\frac{1}{MN}\sum_{i=1}^{N}\sum_{m=1}^{M}\frac{1}{\textsf{f}_{\tau_{m},it}(0)}\eta_{h,\tau_{m},it}\cdot\frac{1}{T}\sum_{s=1}^{T}\frac{1}{\textsf{f}_{\tau_{m},is}(0)}\eta_{h,\tau_{m},is} \cdot f^{H}_{0,s}}_{B_{1t}}\\
&-\underbrace{\frac{1}{MN}\sum_{i=1}^{N}\sum_{m=1}^{M}\frac{1}{\textsf{f}_{\tau_{m},it}(0)}\eta_{h,\tau_{m},it}\cdot\frac{1}{T}\sum_{s=1}^{T}\frac{1}{\textsf{f}_{\tau_{m},is}(0)}\eta_{h,\tau_{m},is}\cdot\left(\tilde{f}_{s}-f^{H}_{0,s}\right)}_{B_{2t}}\\
&-\underbrace{\frac{1}{MN}\sum_{i=1}^{N}\sum_{m=1}^{M}\frac{1}{\textsf{f}_{\tau_{m},it}(0)}\eta_{h,\tau_{m},it}\cdot\frac{1}{T}\sum_{s=1}^{T}\frac{1}{\textsf{f}_{\tau_{m},is}(0)}\hat{R}^{(2)}_{h,\tau_{m},is}\cdot f^{H}_{0,s}\lambda^{H'}_{0,i}(\tau_{m})\left(\tilde{f}_{s}-f^{H}_{0,s}\right)}_{B_{3t}}+o_{p}\left(\frac{h^{2}}{\sqrt{N}}\right).
\end{align*}
We first consider $B_{1t}$.
\begin{align}
B_{1t}=&\frac{1}{MNT}\sum_{i=1}^{N}\sum_{m=1}^{M}\frac{\eta_{h,\tau_{m},it}^{2}}{\textsf{f}_{\tau_{m,it}}^{2}(0)}f^{H}_{0,t}+\frac{1}{M}\sum_{m=1}^{M}\frac{1}{NT}\sum_{i=1}^{N}\sum_{s\neq t}\frac{\eta_{h,\tau_{m},is}\cdot\eta_{h,\tau_{m},it}f^{H}_{0,s}}{\textsf{f}_{\tau_{m},is}(0)\textsf{f}_{\tau_{m},it}(0)}\notag\\
=&O_{p}\left(\frac{1}{T}\right)+O_{p}\left(\frac{\sqrt{\log MT}}{\sqrt{NT}}\right)=o_{p}\left(\frac{h^{2}}{\sqrt{N}}\right)\label{eq.b1}
\end{align}
uniformly in $t$. The order of the second term on the right-hand side of the first equality is by the Hoeffding's inequality under a similar argument as \eqref{eq.subsample}, noting that for each $m$, conditional on $\{\eta_{h,\tau_{m},it}:i=1,\ldots,N\}$, $\eta_{h,\tau_{m},is}$ are independent across $i,s$ and mean zero. The last equality is by Assumption \ref{assum.kernel}.

We then consider the $r\times 1$ vector $B_{2}$. For any $j=1.,,,.r$, let $\Delta_{\tilde{F},t}(j)$ be an $N\times (T-1)$ matrix such that each column in it is constructed by replicating the scalar $(\tilde{f}_{js}-f^{H}_{0,js})$ for $N$ times for each $s\neq t$. By construction, the rank of $\Delta_{\tilde{F},t}(j)$ is 1. Let $\bm{\eta}_{h,\tau_{m},t}$ be the $N\times (T-1)$ matrix formed by $\eta_{h,\tau_{m},it}\eta_{h,\tau_{m},is}/(\textsf{f}_{\tau_{m},is}(0)\textsf{f}_{\tau_{m},it}(0))$ for all $i$ and all $s\neq t$. For each $m$, conditional on $\{\eta_{h,\tau_{m},it}:i=1,\ldots,N\}$, entries in $\bm{\eta}_{h,\tau_{m},t}$ are independent, bounded, and mean zero. Hence, by Theorem 4.4.5 in \cite{vershynin2018high} (p.85) and by the law of iterated expectation, one can show that for some $C>0$, the operator norm of $\bm{\eta}_{h,\tau_{m},t}$ satisfies
\begin{align}
&\Pr\left(\max_{m,t}\left\|\bm{\eta}_{h,\tau_{m},t}\right\|\geq C\left(\sqrt{N}+\sqrt{T}\right)\right)\notag\\
\leq &MT\max_{m,t}\mathbb{E}\left[\Pr\left(\left\|\bm{\eta}_{h,\tau_{m},t}\right\|\geq C\left(\sqrt{N}+\sqrt{T}\right)\Bigg|\{\eta_{h,\tau_{m},it}:i=1,\ldots,N\}\right)\right]\to 0\label{eq.operator}
\end{align}

 Now for the $j$-th row in $B_{2t}$,
\begin{align}
B_{2t,j}=&\frac{1}{MN}\sum_{i=1}^{N}\sum_{m=1}^{M}\frac{1}{\textsf{f}_{\tau_{m},it}(0)}\eta_{h,\tau_{m},it}\cdot\frac{1}{T}\sum_{s=1}^{T}\frac{1}{\textsf{f}_{\tau_{m},is}(0)}\eta_{h,\tau_{m},is}\cdot\left(\tilde{f}_{js}-f^{H}_{0,js}\right)\notag\\
=&\frac{1}{M}\sum_{m=1}^{M}\frac{1}{NT}\sum_{i=1}^{N}\sum_{s\neq t}\frac{\eta_{h,\tau_{m},is}\cdot\eta_{h,\tau_{m},it}\left(\tilde{f}_{js}-f^{H}_{0,js}\right)}{\textsf{f}_{\tau_{m},is}(0)\textsf{f}_{\tau_{m},it}(0)}+\frac{1}{MNT}\sum_{m=1}^{M}\sum_{i=1}^{N}\frac{1}{\textsf{f}_{\tau_{m,it}}^{2}(0)}\eta_{h,\tau_{m},it}^{2}\left(\tilde{f}_{jt}-f_{0,jt}^{H}\right)\notag\\
\leq& \frac{1}{NT}\max_{m,t}\left|\left\langle \bm{\eta}_{h,\tau_{m},t},\Delta_{\tilde{F},t}(j)\right\rangle\right|+O_{p}\left(\frac{\zeta_{NT}}{Th}\right)\notag\\
\leq&\frac{1}{NT}\max_{m,t}\left\|\bm{\eta}_{h,\tau_{m},t}\right\| \cdot \max_{t}\left\|\Delta_{\tilde{F},t}(j)\right\|_{*}+O_{p}\left(\frac{\zeta_{NT}}{Th}\right)\notag\\
=&\frac{1}{NT}\max_{m,t}\left\|\bm{\eta}_{h,\tau_{m},t}\right\| \cdot \max_{t} \left\|\Delta_{\tilde{F},t}(j)\right\|_{F}+O_{p}\left(\frac{\zeta_{NT}}{Th}\right)\notag\\
\leq &\frac{1}{NT}\max_{m,t}\left\|\bm{\eta}_{h,\tau_{m},t}\right\| \cdot \sqrt{N}\left\|\tilde{F}-F_{0}^{H}\right\|_{F}+O_{p}\left(\frac{\zeta_{NT}}{Th}\right)\notag\\
=&O_{p}\left(\frac{\zeta_{NT}}{\sqrt{T}}\right)+O_{p}\left(\frac{\zeta_{NT}}{Th}\right),\label{eq.terry}
\end{align}
where the fifth equality is by $\|\Delta_{\tilde{F},t}(j)\|_{*}\leq \sqrt{\text{rank}(\Delta_{\tilde{F},t}(j))}\|\Delta_{\tilde{F},t}(j)\|_{F}$ and $\text{rank}(\Delta_{\tilde{F},t}(j))=1$. The last equality is by \eqref{eq.operator}. Hence, $B_{2t}=o_{p}(h^{2}/\sqrt{N})$ uniformly in $t$.

For $B_{3t}$,
\begin{align*}
B_{3t}=&\frac{1}{MN}\sum_{i=1}^{N}\sum_{m=1}^{M}\frac{1}{\textsf{f}_{\tau_{m},it}(0)}\eta_{h,\tau_{m},it}\cdot\frac{1}{T}\sum_{s=1}^{T}\frac{1}{\textsf{f}_{\tau_{m},is}(0)}\left(\hat{R}^{(2)}_{h,\tau_{m},is}-\mathbb{E}\left(\hat{R}^{(2)}_{h,\tau_{m},is}\right)\right)\cdot f^{H}_{0,s}\lambda^{H'}_{0,i}(\tau_{m})\left(\tilde{f}_{s}-f^{H}_{0,s}\right)\\
&+\frac{1}{MN}\sum_{i=1}^{N}\sum_{m=1}^{M}\frac{1}{\textsf{f}_{\tau_{m},it}(0)}\eta_{h,\tau_{m},it}\cdot\frac{1}{T}\sum_{s=1}^{T} f^{H}_{0,s}\lambda^{H'}_{0,i}(\tau_{m})\left(\tilde{f}_{s}-f^{H}_{0,s}\right)+O_{p}\left(\zeta_{NT}h^{\gamma}\right).
\end{align*}
The first term is $O_{p}(\zeta_{NT}/(\sqrt{N}h))=o_{p}(h^{2}/\sqrt{N})$ uniformly in $t$ following a similar argument for $B_{2t,j}$ (\eqref{eq.terry}) by treating $h\cdot \left(\hat{R}^{(2)}_{h,\tau_{m},is}-\mathbb{E}\left(\hat{R}^{(2)}_{h,\tau_{m},is}\right)\right)$ as $\eta_{h,\tau_{m},is}$. For the second term,
\begin{align*}
&\frac{1}{MN}\sum_{i=1}^{N}\sum_{m=1}^{M}\frac{1}{\textsf{f}_{\tau_{m},it}(0)}\eta_{h,\tau_{m},it}\cdot\frac{1}{T}\sum_{s=1}^{T} f^{H}_{0,s}\lambda^{H'}_{0,i}(\tau_{m})\left(\tilde{f}_{s}-f^{H}_{0,s}\right)\\
=&\left[\frac{1}{T}\sum_{s=1}^{T}f^{H}_{0,s}\left(\tilde{f}_{s}-f^{H}_{0,s}\right)'\right] \cdot H_{NT,2}^{-1}\left[\frac{1}{N}\sum_{i=1}^{n}\left(\frac{1}{M}\sum_{m=1}^{M}\frac{\eta_{h,\tau_{m},it}\lambda_{0,i}(\tau_{m})}{\textsf{f}_{\tau_{m},it}(0)}\right)\right]\\
=&O_{p}\left(\zeta_{NT}\right)\cdot O_{p}(1)\cdot O_{p}\left(\frac{\sqrt{\log(T)}}{\sqrt{N}}\right)=o_{p}\left(\frac{h^{2}}{\sqrt{N}}\right),
\end{align*}
where the second equality is by the Hoeffding's inequality under boundedness and independence of $\sum_{m=1}^{M}\eta_{h,\tau_{m},it}\lambda_{0,i}(\tau_{m})/M$ across $i$; the term $\log(T)$ is for uniformity in $t$. Combining all these results, $A_{1t}=o_{p}(h^{2}/\sqrt{N})$ uniformly in $t$. 

Now we consider $A_{2t}$.
\begin{align*}
A_{2t}=&-\underbrace{\frac{1}{MN}\sum_{i=1}^{N}\sum_{m=1}^{M}\frac{\hat{R}^{(2)}_{h,\tau_{m},it}}{\textsf{f}_{\tau_{m},it}(0)}\lambda^{H}_{0,i}(\tau_{m})f^{H'}_{0,t}\cdot \frac{1}{T}\sum_{s=1}^{T}\frac{\eta_{h,\tau_{m},is}}{\textsf{f}_{\tau_{m},is}(0)} \cdot f^{H}_{0,s}}_{B_{4t}}\\
&-\underbrace{\frac{1}{MN}\sum_{i=1}^{N}\sum_{m=1}^{M}\frac{\hat{R}^{(2)}_{h,\tau_{m},it}}{\textsf{f}_{\tau_{m},it}(0)}\lambda^{H}_{0,i}(\tau_{m})f^{H'}_{0,t}\cdot \frac{1}{T}\sum_{s=1}^{T}\frac{\eta_{h,\tau_{m},is}}{\textsf{f}_{\tau_{m},is}(0)}\cdot\left(\tilde{f}_{s}-f^{H}_{0,s}\right)}_{B_{5t}}\\
&-\underbrace{\frac{1}{MN}\sum_{i=1}^{N}\sum_{m=1}^{M}\frac{\hat{R}^{(2)}_{h,\tau_{m},it}}{\textsf{f}_{\tau_{m},it}(0)}\lambda^{H}_{0,i}(\tau_{m})f^{H'}_{0,t}\cdot \frac{1}{T}\sum_{s=1}^{T}\frac{\hat{R}^{(2)}_{h,\tau_{m},is}}{\textsf{f}_{\tau_{m},is}(0)}\cdot f^{H}_{0,s}\lambda^{H'}_{0,i}(\tau_{m})\left(\tilde{f}_{s}-f^{H}_{0,s}\right)}_{B_{6t}}+o_{p}\left(\frac{h}{\sqrt{N}}\right).
\end{align*}
Following the same argument for $B_{1t}$, we can show that $B_{4t}=o_{p}(h^{2}/\sqrt{N})$ uniformly in $t$. Similarly, by the same argument for $B_{2t}$, $B_{5t}=O_{p}(\zeta_{NT}/(\sqrt{N}h))=o_{p}(h^{2}/\sqrt{N})$ uniformly in $t$. For $B_{6t}$,
\begin{align*}
B_{6t}=&\frac{1}{MN}\sum_{i=1}^{N}\sum_{m=1}^{M}\frac{\hat{R}^{(2)}_{h,\tau_{m},it}}{\textsf{f}_{\tau_{m},it}(0)}\lambda^{H}_{0,i}(\tau_{m})f^{H'}_{0,t}\cdot \frac{1}{T}\sum_{s=1}^{T}\frac{\hat{R}^{(2)}_{h,\tau_{m},is}}{\textsf{f}_{\tau_{m},is}(0)}\cdot f^{H}_{0,s}\lambda^{H'}_{0,i}(\tau_{m})\left(\tilde{f}_{s}-f^{H}_{0,s}\right)\\
=&\frac{1}{MN}\sum_{i=1}^{N}\sum_{m=1}^{M}\frac{\hat{R}^{(2)}_{h,\tau_{m},it}}{\textsf{f}_{\tau_{m},it}(0)}\lambda^{H}_{0,i}(\tau_{m})f^{H'}_{0,t}\cdot \frac{1}{T}\sum_{s=1}^{T}\frac{\hat{R}^{(2)}_{h,\tau_{m},is}-\mathbb{E}\left(\hat{R}^{(2)}_{h,\tau_{m},is}\right) }{\textsf{f}_{\tau_{m},is}(0)}\cdot f^{H}_{0,s}\lambda^{H'}_{0,i}(\tau_{m})\left(\tilde{f}_{s}-f^{H}_{0,s}\right)\\
&+\frac{1}{MN}\sum_{i=1}^{N}\sum_{m=1}^{M}\frac{\hat{R}^{(2)}_{h,\tau_{m},it}}{\textsf{f}_{\tau_{m},it}(0)}\lambda^{H}_{0,i}(\tau_{m})f^{H'}_{0,t}\cdot \frac{1}{T}\sum_{s=1}^{T} f^{H}_{0,s}\lambda^{H'}_{0,i}(\tau_{m})\left(\tilde{f}_{s}-f^{H}_{0,s}\right)+O_{p}\left(h^{\gamma-1}\zeta_{NT}\right).
\end{align*}
The first term on the right-hand side of the second equality is $O_{p}(\zeta_{NT}/(\sqrt{N}h^{2}))=o_{p}(h^{2}/\sqrt{N})$ uniformly in $t$ following the same argument as for $B_{2t}$. For the second term,
\begin{align*}
&\frac{1}{MN}\sum_{i=1}^{N}\sum_{m=1}^{M}\frac{\hat{R}^{(2)}_{h,\tau_{m},it}}{\textsf{f}_{\tau_{m},it}(0)}\lambda^{H}_{0,i}(\tau_{m})f^{H'}_{0,t}\cdot \frac{1}{T}\sum_{s=1}^{T} f^{H}_{0,s}\lambda^{H'}_{0,i}(\tau_{m})\left(\tilde{f}_{s}-f^{H}_{0,s}\right)\\
=&\frac{1}{MN}\sum_{i=1}^{N}\sum_{m=1}^{M}\lambda^{H}_{0,i}(\tau_{m})f^{H'}_{0,t}\cdot \frac{1}{T}\sum_{s=1}^{T} f^{H}_{0,s}\lambda^{H'}_{0,i}(\tau_{m})\left(\tilde{f}_{s}-f^{H}_{0,s}\right)\\
&+\frac{1}{MN}\sum_{i=1}^{N}\sum_{m=1}^{M}\frac{\hat{R}^{(2)}_{h,\tau_{m},it}-\mathbb{E}\left(\hat{R}^{(2)}_{h,\tau_{m},it}\right)}{\textsf{f}_{\tau_{m},it}(0)}\lambda^{H}_{0,i}(\tau_{m})f^{H'}_{0,t}\cdot \frac{1}{T}\sum_{s=1}^{T} f^{H}_{0,s}\lambda^{H'}_{0,i}(\tau_{m})\left(\tilde{f}_{s}-f^{H}_{0,s}\right)+O_{p}\left(\zeta_{NT}h^{\gamma}\right)\\
=&\frac{1}{MN}\sum_{i=1}^{N}\sum_{m=1}^{M}\lambda^{H}_{0,i}(\tau_{m})f^{H'}_{0,t}\cdot \frac{1}{T}\sum_{s=1}^{T} f^{H}_{0,s}\lambda^{H'}_{0,i}(\tau_{m})\left(\tilde{f}_{s}-f^{H}_{0,s}\right)+O_{p}\left(\zeta_{NT}h^{\gamma}\right)\\
&+\left[\frac{1}{MN}\sum_{i=1}^{N}\sum_{m=1}^{M}\frac{\hat{R}^{(2)}_{h,\tau_{m},it}-\mathbb{E}\left(\hat{R}^{(2)}_{h,\tau_{m},it}\right)}{\textsf{f}_{\tau_{m},it}(0)}\lambda^{H}_{0,i}(\tau_{m})\lambda^{H'}_{0,i}(\tau_{m})\right]\cdot\left[\frac{1}{T}\sum_{s=1}^{T}\left(\tilde{f}_{s}-f^{H}_{0,s}\right) f^{H'}_{0,s}f_{0,t}^{H}\right]\\
=&\frac{1}{MN}\sum_{i=1}^{N}\sum_{m=1}^{M}\lambda^{H}_{0,i}(\tau_{m})f^{H'}_{0,t}\cdot \frac{1}{T}\sum_{s=1}^{T} f^{H}_{0,s}\lambda^{H'}_{0,i}(\tau_{m})\left(\tilde{f}_{s}-f^{H}_{0,s}\right)\\
&+O_{p}\left(\frac{\sqrt{\log(MT)}}{\sqrt{Nh}}\right)O_{p}\left(\zeta_{NT}\right) +O_{p}\left(\zeta_{NT}h^{\gamma}\right),
\end{align*}
where the last equality is by the Bernstein's inequality. 

Therefore,
\begin{align*}
  A_{2t}=&-\frac{1}{MN}\sum_{i=1}^{N}\sum_{m=1}^{M}\lambda^{H}_{0,i}(\tau_{m})f^{H'}_{0,t}\cdot \frac{1}{T}\sum_{s=1}^{T} f^{H}_{0,s}\lambda^{H'}_{0,i}(\tau_{m})\left(\tilde{f}_{s}-f^{H}_{0,s}\right)+o_{p}\left(\frac{h}{\sqrt{N}}\right)\\
  =&-\Bigg\{\left[\frac{1}{MN}\sum_{i=1}^{N}\sum_{m=1}^{M}\lambda^{H}_{0,i}(\tau_{m})\lambda^{H'}_{0,i}(\tau_{m})\right]\left(\frac{1}{T}\sum_{s=1}^{T}\tilde{f}_{s}f^{H'}_{0,s}\right)f^{H}_{0,t}\\
  &-\left[\frac{1}{MN}\sum_{i=1}^{N}\sum_{m=1}^{M}\lambda^{H}_{0,i}(\tau_{m})\lambda^{H'}_{0,i}(\tau_{m})\right]H_{NT,2}'H_{NT,2}f_{0,t}^{H}\Bigg\}+o_{p}\left(\frac{h}{\sqrt{N}}\right)\\
  =&-\Bigg\{\left[\frac{1}{MN}\sum_{i=1}^{N}\sum_{m=1}^{M}\lambda^{H}_{0,i}(\tau_{m})\lambda^{H'}_{0,i}(\tau_{m})\right]H_{NT,2}'H_{NT,2}f^{H}_{0,t}\\
  &-\left[\frac{1}{MN}\sum_{i=1}^{N}\sum_{m=1}^{M}\lambda^{H}_{0,i}(\tau_{m})\lambda^{H'}_{0,i}(\tau_{m})\right]H_{NT,2}'H_{NT,2}f^{H}_{0,t}\Bigg\}+o_{p}\left(\frac{h}{\sqrt{N}}\right)\\
  =&o_{p}\left(\frac{h}{\sqrt{N}}\right),
\end{align*}
where the second equality is by $\sum_{s=1}^{T}f_{0,s}f_{0,s}'/T=I_{r}$. The third equality is because $H_{NT,2}$ is defined as $\sum_{s=1}^{T}f_{0,s}\tilde{f}_{0,s}'/T$. This step shows that it is crucial to expand the first order conditions around $H_{NT,2}'f_{0,t}$ and $H_{NT,2}^{-1}\lambda_{0,i}(\tau_{m})$ instead of $\tilde{H}_{NT,1}'f_{0,t}$ and $\tilde{H}_{NT,1}^{-1}\lambda_{0,i}(\tau_{m})$.

Substitute $A_{1t}$ and $A_{2t}$ into \eqref{eq.expansion.f}, then we have the following expansion by $f_{0,t}^{H}= H_{NT,2}'f_{0,t}$:
\begin{align}
  \tilde{f}_{t}-H_{NT,2}'f_{0,t}=&-\left[\frac{1}{MN}\sum_{i=1}^{n}\sum_{m=1}^{M}\lambda^{H}_{0,i}(\tau_{m})\lambda^{H'}_{0,i}(\tau_{m})\right]^{-1}\frac{1}{MN}\sum_{i=1}^{N}\sum_{m=1}^{M}\frac{\eta_{h,\tau_{m},it}}{\textsf{f}_{\tau_{m},it}(0)}\cdot \lambda^{H}_{0,i}(\tau_{m})\notag\\
  &+o_{p}\left(\frac{h}{\sqrt{N}}\right).\label{eq.final.f.H}
\end{align}
The desired expansion of $\tilde{f}_{t}$ is thus obtained by plugging in $\lambda_{0,i}^{H}(\tau_{m})\coloneqq H_{NT,2}^{-1}\lambda_{0,i}(\tau_{m})$.

Substitute equation \eqref{eq.final.f.H} into \eqref{eq.expansion.lambda}. By similar argument as equations \eqref{eq.b1} and \eqref{eq.terry}, we have
\begin{equation*}
  \tilde{\lambda}_{i}(\tau_{m})-H_{NT,2}^{-1}\lambda_{0,i}(\tau_{m})=-\frac{1}{T}\sum_{t=1}^{T}\frac{\eta_{h,\tau_{m},it}}{\textsf{f}_{\tau_{m},it}(0)}H_{NT,2}'f_{0,t}+o_{p}\left(\frac{1}{\sqrt{T}}\right).
\end{equation*}
The desired results are obtained by the uniformity of $o_{p}(h/\sqrt{N})$ and $o_{p}(1/\sqrt{T})$ in $m,i,t$.
\end{proof}

\end{appendices}

\bibliographystyle{chicago}
\bibliography{references}
\end{document}